\documentclass{article}
\usepackage{fullpage}
\usepackage[utf8]{inputenc}

\usepackage{graphicx}
\usepackage{amsmath,amsthm,amssymb}
\usepackage{algorithm}
\usepackage{algpseudocode}
\usepackage{multirow}
\usepackage{makecell}

\newtheorem{theorem}{Theorem}[section]

\newtheorem{lemma}[theorem]{Lemma}

\theoremstyle{definition}
\newtheorem{remark}{Remark}
\newtheorem{definition}{Definition}[section]

\newcommand{\istrut}[2][0]{\rule[- #1 mm]{0mm}{#1 mm}\rule{0mm}{#2 mm}}
\newcommand{\rb}[2]{\raisebox{#1 mm}[0mm][0mm]{#2}}

\newcommand{\paren}[1]{\mathopen{}\left( #1 \right)\mathclose{}}

\newcommand{\out}{\operatorname{out}}
\newcommand{\dist}{\operatorname{dist}}
\newcommand{\MSSP}{\textsf{MSSP}}
\newcommand{\VD}{\mathrm{VD}}

\newcommand{\VDout}{\VD^*_{\out}}
\newcommand{\VDfarout}{\VD^*_{\operatorname{farout}}}
\newcommand{\Vor}{\operatorname{Vor}}
\newcommand{\boundary}{\partial}
\newcommand{\bydef}{\stackrel{\operatorname{def}}{=}}
\newcommand{\flag}{\operatorname{flag}}
\newcommand{\parent}{\operatorname{par}}
\newcommand{\ET}{\mathsf{ET}}

\newcommand{\Dist}{\boldsymbol{\mathsf{Dist}}}

\newcommand{\PointLocate}{\boldsymbol{\mathsf{PointLocate}}}
\newcommand{\CentroidSearch}{\boldsymbol{\mathsf{CentroidSearch}}}
\newcommand{\Navigation}{\boldsymbol{\mathsf{Navigation}}}
\newcommand{\SitePathIndicator}{\boldsymbol{\mathsf{SitePathIndicator}}}
\newcommand{\ChordIndicator}{\boldsymbol{\mathsf{ChordIndicator}}}
\newcommand{\MaximalChord}{\boldsymbol{\mathsf{MaximalChord}}}
\newcommand{\AdjacentPiece}{\boldsymbol{\mathsf{AdjacentPiece}}}
\newcommand{\PieceSearch}{\boldsymbol{\mathsf{PieceSearch}}}

\newcommand{\chord}[1]{\overrightarrow{#1}}

\newcommand{\Patrascu}{P\v{a}tra\c{s}cu}

\newcommand{\ignore}[1]{}

\title{Planar Distance Oracles with Better 
Time-Space Tradeoffs\thanks{This work was supported by NSF grants CCF-1637546 and CCF-1815316, and a grant from IIIS, Tsinghua University.}}

\author{Yaowei Long\\
IIIS, Tsinghua University
\and
Seth Pettie\\ University of Michigan \\ 
}

\date{}

\begin{document}
\maketitle

\begin{abstract}
In a recent breakthrough, Charalampopoulos, Gawrychowski, Mozes, and Weimann~\cite{CharalampopoulosGMW19} showed that
exact distance queries on planar graphs could be answered in $n^{o(1)}$ time by a data structure 
occupying $n^{1+o(1)}$ space, i.e., up to $o(1)$ terms, 
optimal exponents in time (0) and space (1) can be achieved \emph{simultaneously}.
Their distance query algorithm is recursive: it makes successive calls to a 
point-location algorithm for 
planar Voronoi diagrams, which involves many recursive distance queries.  
The depth of this recursion is non-constant
and the branching factor logarithmic, leading to $(\log n)^{\omega(1)} = n^{o(1)}$ query times.

In this paper we present a new way to do 
point-location in planar Voronoi diagrams, 
which leads to a new exact distance oracle.  
At the two extremes of our space-time tradeoff curve
we can achieve either
\[
\begin{array}{lll}
\mbox{$n^{1+o(1)}$ space } 
 & \mbox{and}
 & \mbox{ $\log^{2+o(1)} n$ query time, or }\\
&&\\
\mbox{$n\log^{2+o(1)} n$ space }
 &\mbox{and}
 &\mbox{$n^{o(1)}$ query time.}\\
\end{array}
\]
All previous oracles 
with $\tilde{O}(1)$ query time
occupy space $n^{1+\Omega(1)}$, 
and all previous oracles
with space $\tilde{O}(n)$ 
answer queries in $n^{\Omega(1)}$ time.
\end{abstract}

\thispagestyle{empty}
%\clearpage
\setcounter{page}{0}

\newpage

\section{Introduction}

A \emph{distance oracle} is a data structure that answers distance
queries (or approximate distance queries) 
w.r.t.~some underlying graph or metric space.
On general graphs there
are many well known distance oracles that pit
space against multiplicative approximation~\cite{ThorupZ05},
space against mixed multiplicative/additive approximation~\cite{PatrascuR14,AbrahamG11}, 
and, in sparse graphs, space against query time~\cite{Agarwal14,SommerVY09}.
Refer to Sommer \cite{Sommer14} for a survey on distance oracles.

Whereas \emph{approximation} seems to be a necessary ingredient 
to achieve any reasonable space/query time on general graphs,
structured graph classes may admit \emph{exact} distance oracles
with attractive time-space tradeoffs.  In this paper we continue a long line of work~\cite{ArikatiCCDSZ96,Djidjev96,ChenX00,FakcharoenpholR06,Klein05,Wulff-Nilsen10,Nussbaum11,Cabello12,MozesS2012,Cohen-AddadDW17,GawrychowskiMWW18,CharalampopoulosGMW19} 
focused on exact distance oracles for weighted, directed planar graphs.

\paragraph{History.} Between 1996-2012, work of 
Arikati et al.~\cite{ArikatiCCDSZ96},
Djidjev~\cite{Djidjev96}, Chen and Xu~\cite{ChenX00}, 
Fakcharoenphol and Rao~\cite{FakcharoenpholR06}, Klein~\cite{Klein05},
Wulff-Nilsen~\cite{Wulff-Nilsen10}, Nussbaum~\cite{Nussbaum11},
Cabello~\cite{Cabello12}, and Mozes and Sommer~\cite{MozesS2012}
achieved space $\tilde{O}(S)$ and query time $\tilde{O}(n/\sqrt{S})$, 
for various ranges of $S$ that ultimately covered the full range $[n,n^2]$.

In 2017, Cabello~\cite{Cabello19} 
introduced \emph{planar} Voronoi diagrams as a tool for 
solving metric problems in planar graphs, such as diameter and sum-of-distances.
This idea was incorporated into new planar distance oracles, leading
to $\tilde{O}(n^{5/2}/S^{3/2})$ query time~\cite{Cohen-AddadDW17} 
for $S\in [n^{3/2},n^{5/3}]$ and $\tilde{O}(n^{3/2}/S)$ query time~\cite{GawrychowskiMWW18} 
for $S\in [n,n^{3/2}]$.
Finally, in a major 
breakthrough Charalampopoulos, Gawrychowski, Mozes, and Weimann~\cite{CharalampopoulosGMW19} 
demonstrated that up to $n^{o(1)}$ factors, 
\emph{there is no tradeoff} between space and query time, i.e., 
space $n^{1+o(1)}$ and query time 
$n^{o(1)}$ can be achieved simultaneously.
In more detail, 
they proved that space $O(n^{4/3}\sqrt{\log n})$ allows
for query time $O(\log^2 n)$, 
space $\tilde{O}(n^{1+\epsilon})$ allows for query time 
$O(\log n)^{1/\epsilon-1}$, 
and space 
$O(n\log^{2+1/\epsilon} n)$  allows for query time
$O(n^{2\epsilon})$.

The Charalampopoulos et al.~structure is based on a hierarchical $\vec{r}$-decomposition of the graph, 
$\vec{r}=(n,n^{(m-1)/m},\ldots,n^{1/m})$.  
(See Section~\ref{sect:prelims}.)
Given $u,v$, it iteratively finds the last
boundary vertex $u_i$ on the shortest $u$-$v$ path
that lies on the boundary of the level-$i$ region
containing $u$.  Given $u_{i-1}$, finding $u_i$
amounts to solving a \emph{point location} 
problem on an \emph{external} Voronoi diagram,
i.e., a Voronoi diagram of the \emph{complement} 
of a region in the hierarchy.
Each point location query is solved via a kind of binary search,
and each step of the binary search involves 
3 \emph{recursive} distance queries that begin
at a ``higher'' level in the hierarchy.  
This leads to a tradeoff between 
space $\tilde{O}(n^{1+1/m})$ and 
query time $O(\log n)^{m-1}$. 

See Table~\ref{tab:priorwork} for a summary of 
the space-time tradeoffs exact and approximate planar 
distance oracles.
\begin{table}[]
    \centering
\begin{tabular}{|l|l|l|}
\multicolumn{1}{l}{\large\sc Reference} &
\multicolumn{1}{l}{\large\sc Space} &
\multicolumn{1}{l}{\large\sc Query Time}\\
     \hline\hline
     \istrut[4]{5.5}\parbox{110pt}{\small Arikati, Chen, Chew\\
     Das, Smid \& Zaroliagis}\hfill{\small 1996} & $S\in[n^{3/2},n^{2}]$ & $O\paren{\frac{n^{2}}{S}}$  \\
     \hline
     \multirow{2}*{\rb{-1}{\small Djidjev}}\hfill\multirow{2}*{\rb{-1}{\small 1996}} 
               & \istrut[2.5]{4.5}$S\in[n,n^{2}]$ & $O\paren{\frac{n^{2}}{S}}$ \\\cline{2-3}
             ~ & \istrut[2.5]{4.5}$S\in[n^{4/3},n^{3/2}]$ & $O\paren{\frac{n}{\sqrt{S}}\log n}$ \\
     \hline
     \small Chen \& Xu \hfill{\small 2000} & \istrut[2.5]{4.5}$S\in[n^{4/3},n^{2}]$ & $O\paren{\frac{n}{\sqrt{S}}\log\paren{\frac{n}{\sqrt{S}}}}$  \\
     \hline
     \small Fakcharoenphol \& Rao \hfill{\small 2006} & \istrut[2.5]{4.5}$O(n\log n)$ & $O(\sqrt{n}\log^{2} n)$  \\\hline
%     \small Klein \hfill{\small 2005} & $O(n\log n)$ & $O(\sqrt{n}\log^{2}n)$ & $O(n\log^{2} n)$\\
%     \hline 
     \small Wulff-Nilsen\hfill{\small 2010} & \istrut[2.5]{4.5}$O(n^{2}\frac{\log^4\log n}{\log n})$ & $O(1)$ \\
     \hline
     \multirow{2}*{\small Nussbaum}\hfill\multirow{2}*{\small 2011}
                                 & \istrut[2.5]{4.5} $O(n)$ &
                                 \istrut[2.5]{4.5} $O(n^{1/2+\epsilon})$
                                 \\\cline{2-3}
     ~                           & \istrut[2.5]{4.5} $S\in[n^{4/3},n^{2}]$ &
                                 $O\paren{\frac{n}{\sqrt{S}}}$\\
     \hline
     %\small Nussbaum \hfill{\small 2011} & $O(n)$ & \istrut[2.5]{4.5}$O(n^{1/2+\epsilon})$ \\
     %\hline
     \small Cabello\hfill{\small 2012} & \istrut[2.5]{4.5}$S\in[n^{4/3}\log^{1/3}n,n^{2}]$ & $O\paren{\frac{n}{\sqrt{S}}\log^{3/2} n}$ \\
     \hline
     \multirow{2}*{\small Mozes \& Sommer}\hfill\multirow{2}*{\small 2012} 
                                        & \istrut[2.5]{4.5}$S\in[n\log\log n, n^{2}]$ & $O\paren{\frac{n}{\sqrt{S}}\log^{2}n\log^{3/2}\log n}$ \\\cline{2-3}
     ~                                  & \istrut[2.5]{4.5}$O(n)$ & $O(n^{1/2+\epsilon})$ \\
     \hline
     \istrut[4]{5}\parbox{110pt}{\small Cohen-Addad, Dahlgaard\\ \& Wulff-Nilsen}\hfill{\small 2017} 
                                    & $S\in[n^{3/2},n^{5/3}]$  & $O\paren{\frac{n^{5/2}}{S^{3/2}}\log n}$  \\
     \hline
     \istrut[4]{5}\parbox{110pt}{\small Gawrychowski, Mozes,\\
     Weimann \& Wulff-Nilsen}\hfill{\small 2018} 
            & $\tilde{O}(S)$ for $S\in[n,n^{3/2}]$ & $\tilde{O}\paren{\frac{n^{3/2}}{S}}$  \\
     \hline
      \multirow{2}*{\parbox{110pt}{\small Charalampopoulos,                
      Gawrychowski, Mozes\\ \& Weimann}}\hfill\multirow{2}*{\small 2019} 
                    & \istrut[3]{4}$O(n^{4/3}\sqrt{\log n})$  & $O(\log^2 n)$  \\\cline{2-3}
                    & \istrut[2]{4}$n^{1+o(1)}$          & $n^{o(1)}$ \\
     \hline
     \multirow{2}*{{\bf new}}\hfill\multirow{2}*{\small 2020} 
            & $n^{1+o(1)}$  & $\log^{2+o(1)} n$  \\\cline{2-3}
            & $n\log^{2+o(1)} n$ & $n^{o(1)}$ \\\hline\hline
\multicolumn{3}{l}{}\\
\multicolumn{1}{l}{\large\sc $(1+\epsilon)$-Approx. Oracles} &
\multicolumn{1}{l}{\large\sc Space} &
\multicolumn{1}{l}{\large\sc Query Time}\\\hline\hline
\multirow{2}*{\small Thorup}\hfill \multirow{2}*{\small 2001} & \istrut[2.5]{4.5}$O(n\epsilon^{-1}\log^2 n)$ & $O(\log\log n + \epsilon^{-1})$\\\cline{2-3}
~ & \istrut[2.5]{4.5}$O(n\epsilon^{-1}\log n)$ & $O(\epsilon^{-1})$\hfill (Undir.)\\\hline
{\small Klein}\hfill {\small 2002} & \istrut[2.5]{4.5}$O(n(\log n + \epsilon^{-1}\log\epsilon^{-1}))$ & $O(\epsilon^{-1})$\hfill (Undir.)\\\hline
\parbox{110pt}{\small Kawarabayashi,\\                
      Klein, \& Sommer}\hfill{\small 2011} & \istrut[3.5]{5.5}$O(n)$ & $O(\epsilon^{-2}\log^2 n)$\hfill (Undir.)\\\hline
\multirow{2}*{\parbox{110pt}{\small Kawarabayashi,\\                
      Sommer, \& Thorup}}\hfill \multirow{2}*{\small 2013} 
    & \istrut[2]{4.5}$\overline{O}(n\log n)$ & $\overline{O}(\epsilon^{-1})$\hfill (Undir.)\\\cline{2-3}
~   & \istrut[2]{4.5}$\overline{O}(n)$     & $\overline{O}(\epsilon^{-1})$\hfill \ \ \ (Undir.,Unweight.)\\\hline\hline
\end{tabular}
    \caption{Space-query time tradeoffs for exact and approximate planar distance oracles.
    $\overline{O}$ hides $\log(\epsilon^{-1}\log n)$ factors.}
    \label{tab:priorwork}
\end{table}
\nocite{Thorup04}
\nocite{KawarabayashiKS11}
\nocite{KawarabayashiST13}
\nocite{Klein02}

\paragraph{New Results.} 
In this paper we develop a more direct and more efficient way 
to do point location in external Voronoi diagrams.
It uses a new persistent data structure for maintaining sets 
of non-crossing systems of \emph{chords}, which are paths
that begin and end at the boundary vertices of a region, but
are internally vertex disjoint from the region.  By applying
this point location method in the 
framework of Charalampopoulos et al.~\cite{CharalampopoulosGMW19},
we obtain a better time-space tradeoff, which is most noticeable
at the ``extremes'' when $\tilde{O}(n)$ space or $\tilde{O}(1)$ query time
is prioritized.
\begin{theorem}\label{thm:maintheorem}
Let $G$ be an $n$-vertex weighted planar digraph with no negative cycles,
and let $\kappa,m\geq 1$ be parameters.
A distance oracle occupying space 
$O(m\kappa n^{1+1/m+1/\kappa})$
can be constructed in $\tilde{O}(n^{3/2+1/m} + n^{1+1/m+1/\kappa})$
time that answers exact distance queries
in $O(2^{m}\kappa\log^{2}n\log\log n)$  time.
At the two extremes of the space-time tradeoff curve, 
we can construct oracles in $n^{3/2+o(1)}$ time with either
    \begin{itemize}
        \item $n^{1+o(1)}$ space and $\log^{2+o(1)}n$ query time, or
        \item $n\log^{2+o(1)}n$ space and $n^{o(1)}$ query time.
    \end{itemize}
\end{theorem}
Our new point-location routine suffices to get the query 
time down to $O(\log^3 n)$.  In order to reduce it further
to $O(\log^{2+o(1)} n)$, we develop a new dynamic tree data structure
based on Euler-Tour trees~\cite{HenzingerK99} with 
$O(\kappa n^{1/\kappa})$ update time and $O(\kappa)$ query time.
This allows us to generate \MSSP{} (multiple-source shortest paths) 
structures with a similar space-query tradeoff, specifically, 
$O(\kappa n^{1+1/\kappa})$ space and $O(\kappa\log\log n)$ query time.
Our \MSSP{} construction follows Klein~\cite{Klein05} (see also~\cite{GawrychowskiMWW18}), but uses our new 
dynamic tree in lieu of Sleator and Tarjan's Link-Cut trees~\cite{SleatorT83}, 
and uses persistent arrays~\cite{Dietz89} in lieu of~\cite{DriscollSST89} 
to make the data structure persistent.

\paragraph{Organization.}
In Section~\ref{sect:prelims} we review background on planar embeddings,
planar separators, multiple-source shortest paths, and weighted Voronoi diagrams.
In Section~\ref{sect:DistanceOracle} we introduce key parts of the data structure and describe the query algorithm, 
assuming a certain point location problem can be solved.  
Section~\ref{sect:NavigationOracle} introduces several more components
of the data structure, and shows how they can be applied to solve
this particular point location problem in near-logarithmic time.
The space and query-time claims of Theorem~\ref{thm:maintheorem} are proved
in Section~\ref{sect:analysis}.  The construction time claims of Theorem~\ref{thm:maintheorem} are proved in Appendix~\ref{sect:construction}.
Appendix~\ref{sect:Euler} gives the \MSSP{} construction based on 
Euler Tour trees.  Appendix~\ref{sect:MultipleHoles} explains how
to remove a simplifying assumption made throughout the paper, that
the boundary vertices of every region in the $\vec{r}$-decomposition 
lie on a \emph{single} hole, which is bounded by a \emph{simple} cycle.

\section{Preliminaries}\label{sect:prelims}

\subsection{The Graph and Its Embedding} 

A weighted planar graph $G=(V,E,\ell)$ is represented by an abstract embedding: for each $v\in V(G)$
we list the edges incident to $v$ according to a clockwise order around $v$.  We assume the graph has no negative weight cycles and further assume the following, without loss of generality.
\begin{itemize}
    \item All the edge-weights can be made non-negative ($\ell : E\rightarrow \mathbb{R}_{\ge 0}$)~\cite{Johnson77}.
    Furthermore, via randomized or deterministic perturbation~\cite{EricksonFL18}, 
    we can assume there are no zero weight edges, and that
    shortest paths are \emph{unique} in \emph{any} subgraph of $G$.
    \item The graph is connected and triangulated.  Assign all artificial edges weight $n\cdot \max_{e\in E(G)}\{\ell(e)\}$ so as not to affect any finite distances.
    \item If $(u,v)\in E(G)$ then $(v,u)\in E(G)$ as well.  (In the circular ordering around $v$, they are represented as a single element $\{u,v\}$.)
\end{itemize}

Suppose $P=(v_0,v_1,\ldots,v_k)$ is a path oriented from $v_0$ to $v_k$,
and $e=(v_i,u)$ is an edge not on $P$, $i\in [1,k-1]$.  
Then $e$ is to the right of $P$ 
if $e$ appears between $(v_i,v_{i+1})$ and $(v_{i-1},v_i)$ in the clockwise
order around $v_i$, and left of $P$ otherwise.

\subsection{Separators and Divisions}

Lipton and Tarjan~\cite{LiptonT80} proved that every planar graph contains
a \emph{separator} of $O(\sqrt{n})$ vertices that, once removed, breaks
the graph into components of at most 2/3 the size.  
Miller~\cite{Miller86} showed that every triangulated planar 
graph has a $O(\sqrt{n})$-size separator that consists of a simple cycle.
Frederickson~\cite{Frederickson87} defined a \emph{division} to 
be a set of edge-induced subgraphs whose union is $G$.  
A vertex in more than one region is a \emph{boundary} vertex;
the boundary of a region $R$ is denoted $\boundary R$.
Edges along the boundary between two regions appear in both regions.
The $r$-divisions of~\cite{Frederickson87} have $\Theta(n/r)$ regions 
each with $O(r)$ vertices and $O(\sqrt{r})$ 
boundary vertices.

We use a linear-time algorithm of 
Klein, Mozes, and Sommer~\cite{KleinMS13}
for computing a hierarchical $\vec{r}$-division, 
where $\vec{r}=(r_m,\ldots,r_1)$ and $n = r_m > \cdots > r_1 = \Omega(1)$.  Such an $\vec{r}$-division has the following properties:
\begin{itemize}
    \item (Division \& Hierarchy) For each $i$, $\mathcal{R}_i$ is the set of regions in an $r_i$-division of $G$, where $\mathcal{R}_m = \{G\}$ consists of the graph itself.
    For each $i<i'\leq m$ and $R_i \in \mathcal{R}_i$, there is a unique $R_{i'}\in \mathcal{R}_{i'}$ such that $E(R_i) \subseteq E(R_{i'})$.  
    The $\vec{r}$-division is therefore represented as a rooted tree of regions.
    \item (Boundaries and Holes) The $O(\sqrt{r_i})$ boundary vertices of any $R_i\in \mathcal{R}_{i}$ lie on a constant number of faces of $R_i$ called \emph{holes},
    each bounded by a cycle (not necessarily simple).
\end{itemize}
We supplement the $\vec{r}$-division with a zeroth level.  
The layer-0 $\mathcal{R}_0 = \{\{v\} \mid v\in V(G)\}$ consists of 
singleton sets, and each $\{v\}$ is attached as a (leaf) child of 
an arbitrary $R\in \mathcal{R}_1$ for which $v\in R$.  

Suppose $f$ is one of the $O(1)$ holes of region $R$ and $C_f$ the cycle around $f$.
The cycle $C_f$ partitions $E(G) - C_f$ into two parts.
Let $R^{f,\out}$ be the graph induced by the part disjoint from $R$, together with $C_f$,
i.e., $C_f$ appears in both $R$ and $R^{f,\out}$.
To keep the description of the algorithm as simple as possible,
\emph{we will assume that $\boundary R$ lies on a single simple cycle (hole)} $f_R$, 
and let $R^{\out}$ be short for $R^{f_R,\out}$.
The modifications necessary to deal with multiple
holes and non-simple boundary cycles are explained in Appendix~\ref{sect:MultipleHoles}.

\subsection{Multi-source Shortest Paths}

Suppose $H$ is a weighted planar graph with a distinguished face $f$ on vertices $S$.
Klein's \MSSP{} algorithm takes $O(|H|\log |H|)$ time and produces an $O(|H|\log |H|)$-size data structure such that given $s\in S$ and $v\in V(H)$, returns $\dist_H(s,v)$ in $O(\log |H|)$ time.  
Klein's algorithm can be viewed as continuously moving the source vertex 
around the boundary face $f$, recording all changes to the SSSP tree 
in a dynamic tree data structure~\cite{SleatorT83}.  
It is shown~\cite{Klein05} that each edge in $H$ 
enters and leaves the SSSP tree exactly once, meaning
the number of changes is $O(|H|)$.  Each change to the tree is effected in 
$O(\log |H|)$ time~\cite{SleatorT83}, and the generic persistence 
method of~\cite{DriscollSST89} allows for querying any state of the SSSP tree.
The important point is that the total space is linear in the number of updates
to the structure ($O(|H|)$) times the update time ($O(\log|H|)$).
As observed in~\cite{GawrychowskiMWW18},
this structure can also answer other useful queries in $O(\log |H|)$ time.
Lemma~\ref{lem:MSSP} is similar to~\cite{Klein05,GawrychowskiMWW18} 
except that we use a dynamic tree data structure based on 
Euler Tour trees~\cite{HenzingerK99}
rather thank Link-Cut trees~\cite{SleatorT83}, 
which allows for a more flexible tradeoff between
update and query time.  Because our data structure does not satisfy the
criteria of Driscoll et al.'s~\cite{DriscollSST89} 
persistence method for pointer-based data structures,
we use the folklore implementation of persistent 
arrays\footnote{Dietz~\cite{Dietz89}
credits this method to an oral presentation of 
Dietzfelbinger et al.~\cite{DietzfelbingerKMHRT88}, 
which highlighted it as an application of 
dynamic perfect hashing.}
to make any RAM data structure persistent, 
with doubly-logarithmic slowdown in the query time.
See Appendix~\ref{sect:Euler} for a proof of Lemma~\ref{lem:MSSP}.

\begin{lemma}\label{lem:MSSP}
(Cf.~Klein~\cite{Klein05}, Gawrychowski et al.~\cite{GawrychowskiMWW18})
Let $H$ be a planar graph, $S$ be the vertices on some distinguished face $f$,
and $\kappa \ge 1$ be a parameter.
An $O(\kappa |H|^{1+1/\kappa})$-space data structure can be computed in 
$O(\kappa|H|^{1+1/\kappa})$ time that answers the 
following queries in $O(\kappa\log\log|H|)$ time.
\begin{itemize}
    \item Given $s\in S, v\in V(H)$, return $\dist_H(s,v)$.
    \item Given $s\in S, u,v\in V(H)$, 
            return $(x,e_u,e_v)$, where $x$
            is the least common ancestor of $u$ and 
            $v$ in the SSSP tree rooted at $s$
            and $e_z$ is the edge on the path from 
            $x$ to $z$ (if $x\neq z$), $z\in \{u,v\}$.
\end{itemize}
\end{lemma}

The purpose of the second query is to tell whether $u$ lies on the shortest $s$-$v$ path ($x=u$)
or vice versa, or to tell which direction the $s$-$u$ path branches from the $s$-$v$ path.  Once we retrieve the LCA $x$ and edges $e_u,u_v$, we get the edge $e_x$
from $x$ to its parent.  The clockwise order of $e_x,e_u,e_v$ around $x$ tells us whether $s$-$u$ branches from $s$-$v$ to the left or right.  See Figure~\ref{fig:LCA}.

\begin{figure}
    \centering
    \scalebox{.4}{\includegraphics{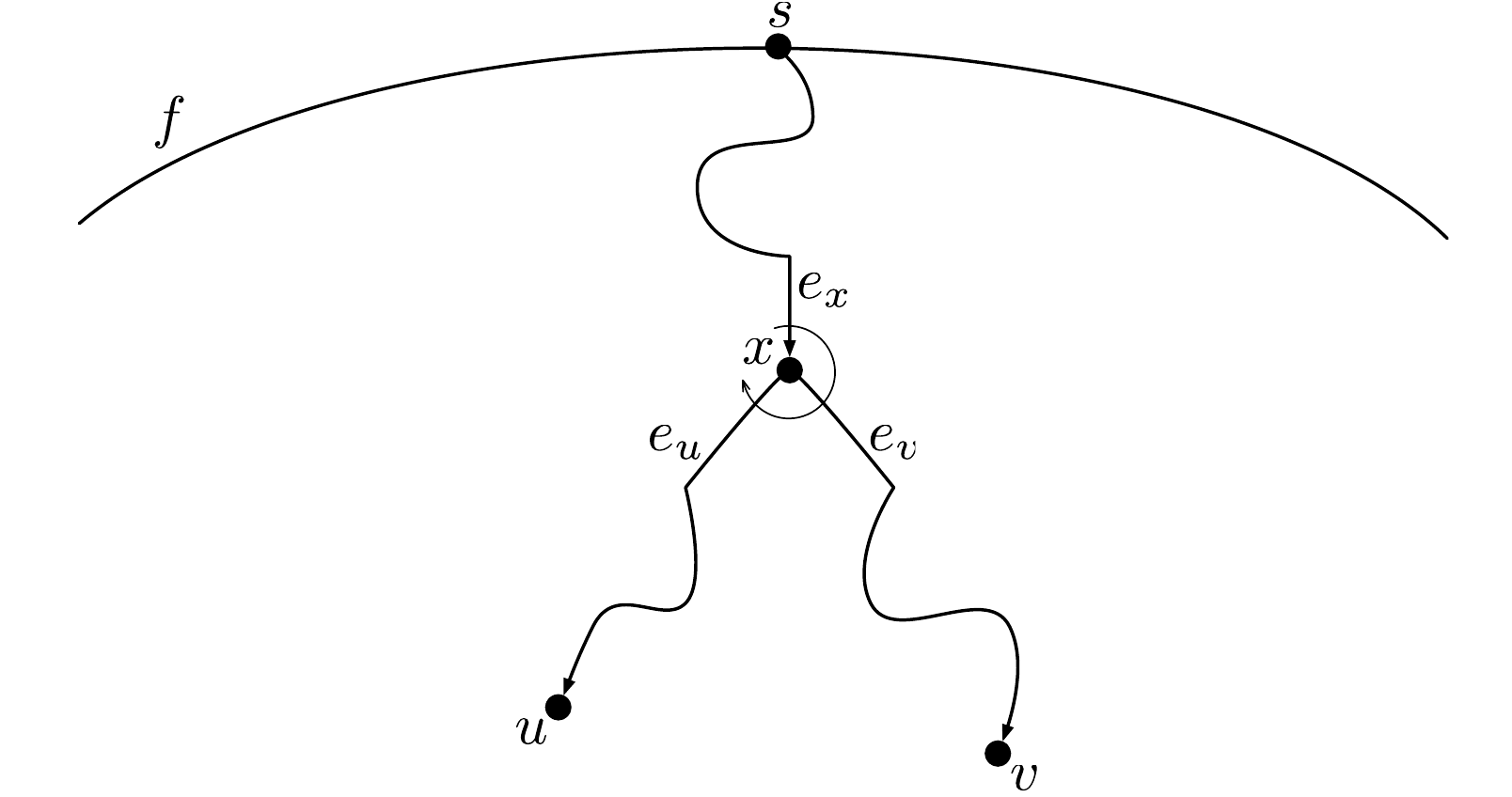}}
    \caption{The clockwise order of $e_x,e_u,e_v$ around $v$ tells us whether
    the shortest $s$-$u$ path branches from the shortest $s$-$v$ path to the right
    or left.}
    \label{fig:LCA}
\end{figure}

\subsection{Additively Weighted Voronoi Diagrams}

Let $H$ be a weighted planar graph, 
$f$ a distinguished face whose vertices $S$ are called \emph{sites},
and $\omega : S \to \mathbb{R}_{\ge 0}$ be
a weight function on sites.
We augment $H$ with large-weight edges so that it is triangulated,
except for $f$.
For $s\in S, v\in V(H)$, define
\[
d^\omega(s,v) \bydef \omega(s) + \dist_H(s,v).
\]
The \emph{Voronoi diagram} $\VD[H,S,\omega]$ is a partition
of $V(H)$ into \emph{Voronoi cells}, where for $s\in S$,
\[
\Vor(s) \bydef \{v \in V(H) \mid 
    \forall s'\neq s.\; (d^\omega(s,v), -\omega(s)) < (d^\omega(s',v), -\omega(s'))\}
\]
In other words, $\Vor(s)$ is the set of vertices that are closer to $s$
than any other site, breaking ties in favor of larger $\omega$-values.  
We usually work with the dual representation
of a Voronoi diagram.  It is constructed as follows.

\begin{itemize}
\item Define $\hat{S}$ to be the set of sites with nonempty Voronoi cells, 
i.e., $\hat{S} = \{s\in S \mid s\in \Vor(s)\}$.  
The case $|\hat{S}|=1$ is trivial, so assume $|\hat{S}|\ge 2$.
\item Add large-weight dummy edges to $H$ so that $\hat{S}$ appear 
on the boundary of a single face $\hat{f}$, but is otherwise triangulated.  Observe that this has no effect on the Voronoi cells.
\item An edge is \emph{bichromatic} if its endpoints are in different cells.  
In particular, the edges bounding $\hat{f}$ are entirely bichromatic.
Define $\VD_0^*$ to be the (undirected) subgraph of $H^*$ consisting
of the duals of bichromatic edges.
\item Obtain $\VD_1^*$ from $\VD_0^*$ by repeatedly contracting edges incident to a degree-2 vertex, terminating when there are no degree-2 vertices, or when it becomes a self-loop.\footnote{The latter case only occurs when $|\hat{S}|=2$.}  
Observe that in $\VD_1^*$, $\hat{f}^*$ has degree $|\hat{S}|$ and all other vertices have degree 3; moreover, the faces of $\VD_1^*$ are in one-to-one correspondence with the Voronoi cells.  
\item We obtain $\VD^* = \VD^*[H,S,\omega]$ by splitting $\hat{f}^*$ into $|\hat{S}|$ degree-1 vertices, each taking an edge formerly incident to $\hat{f}^*$.  
It was proved in~\cite[Lemma 4.1]{GawrychowskiMWW18} that $\VD^*$ is a single tree.\footnote{If we skipped the step of forming the face $\hat{f}$ on the site-set $\hat{S}$ and triangulating the rest, 
$\VD^*$ would still be acyclic, but perhaps disconnected.  See~\cite{GawrychowskiMWW18,CharalampopoulosGMW19}.}
\item We store with $\VD^*$ supplementary information useful for point location. Each degree-3 vertex $g^*$ in $\VD^*$ corresponds a \emph{trichromatic}
face $g$ whose three vertices, say $y_0,y_1,y_2$,
belong to different Voronoi cells.  We store in $\VD^*$ the sites $s_0,s_1,s_2\in S$ such that $y_i \in \Vor(s_i)$.  We also store a \emph{centroid decomposition} of $\VD^*$.  
A centroid of a tree $T$ is a vertex $c$ that partitions the edge set of $T$ into disjoint
subtrees $T_1,\ldots,T_{\deg(c)}$, each containing at most $(|E(T)|+1)/2$ edges, and each containing $c$ as a leaf.
The decomposition is a tree rooted at $c$, whose subtrees
are the centroid decompositions of $T_1,\ldots,T_{\deg(c)}$.
The recursion bottoms out when $T$ consists of a single edge,
which is represented as a single (leaf) node in the centroid decomposition.\footnote{I.e., internal nodes correspond to vertices of $T$; leaf nodes correspond to edges of $T$.}
\end{itemize}

\begin{figure}
    \centering
    \begin{tabular}{c@{\hspace{1.5cm}}c}
    \multicolumn{2}{c}{\scalebox{.30}{\includegraphics{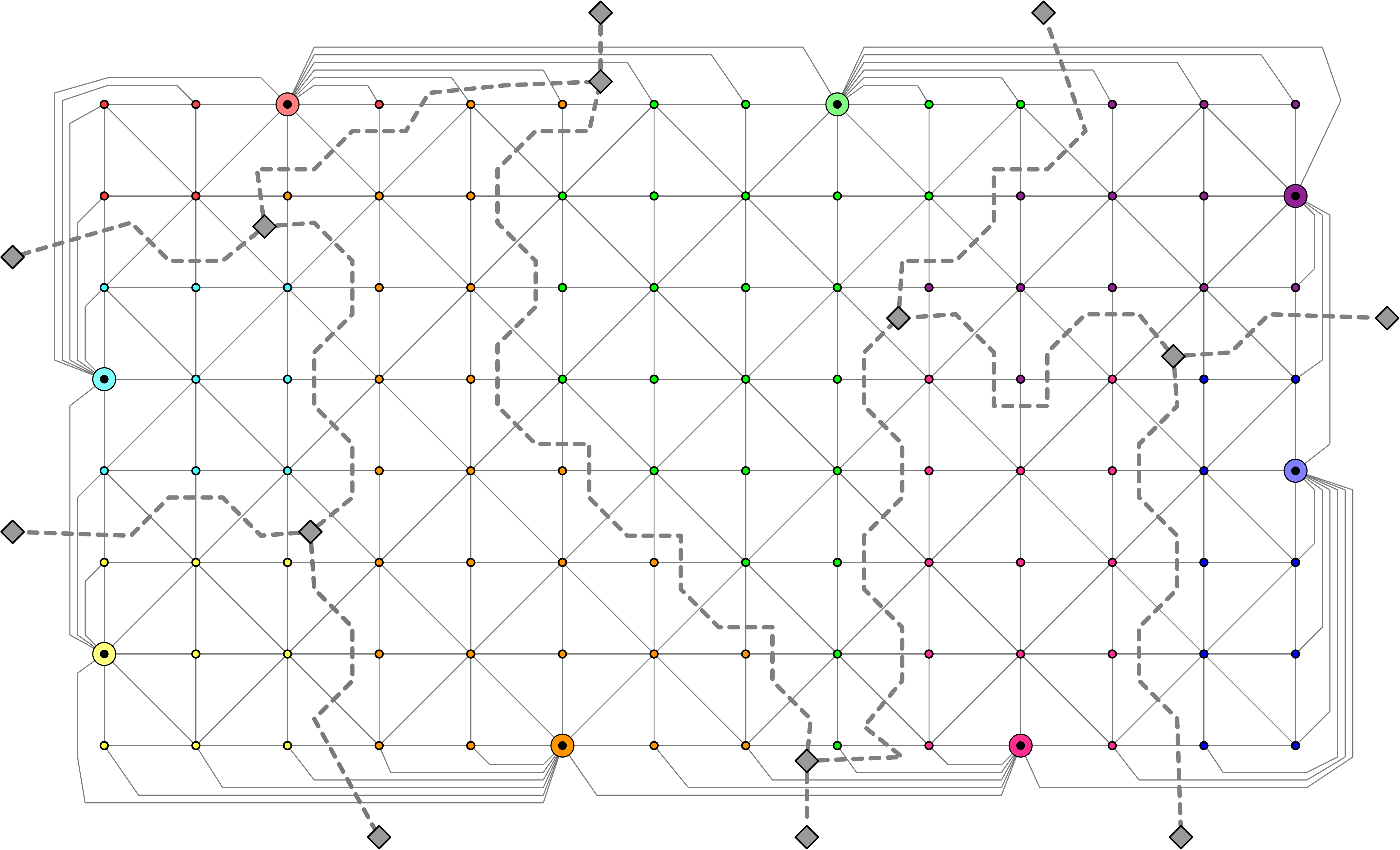}}}\\
    &\\
    \multicolumn{2}{c}{\bf (a)}\\
    \scalebox{.25}{\includegraphics{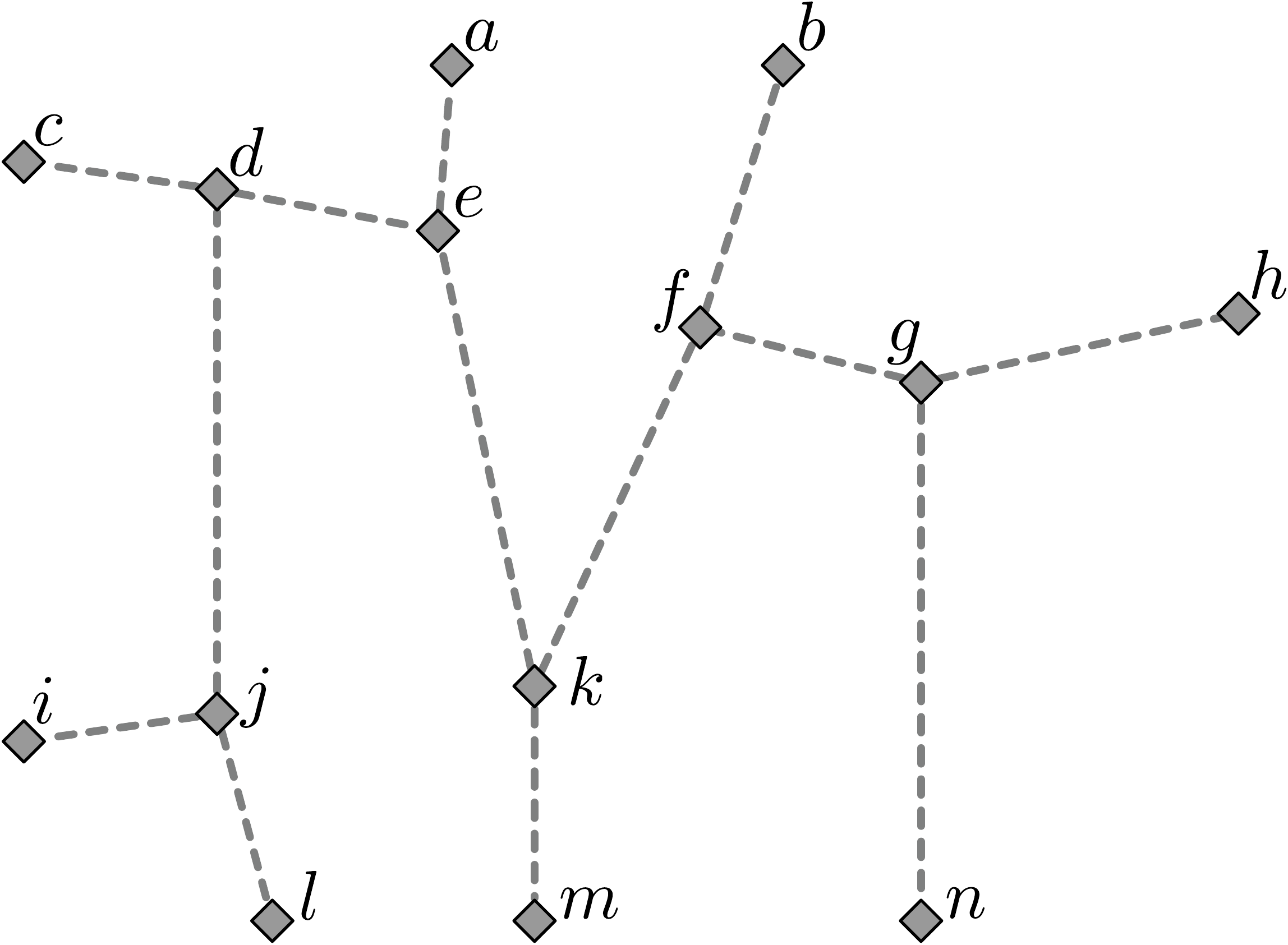}}
    &\scalebox{.3}{\includegraphics{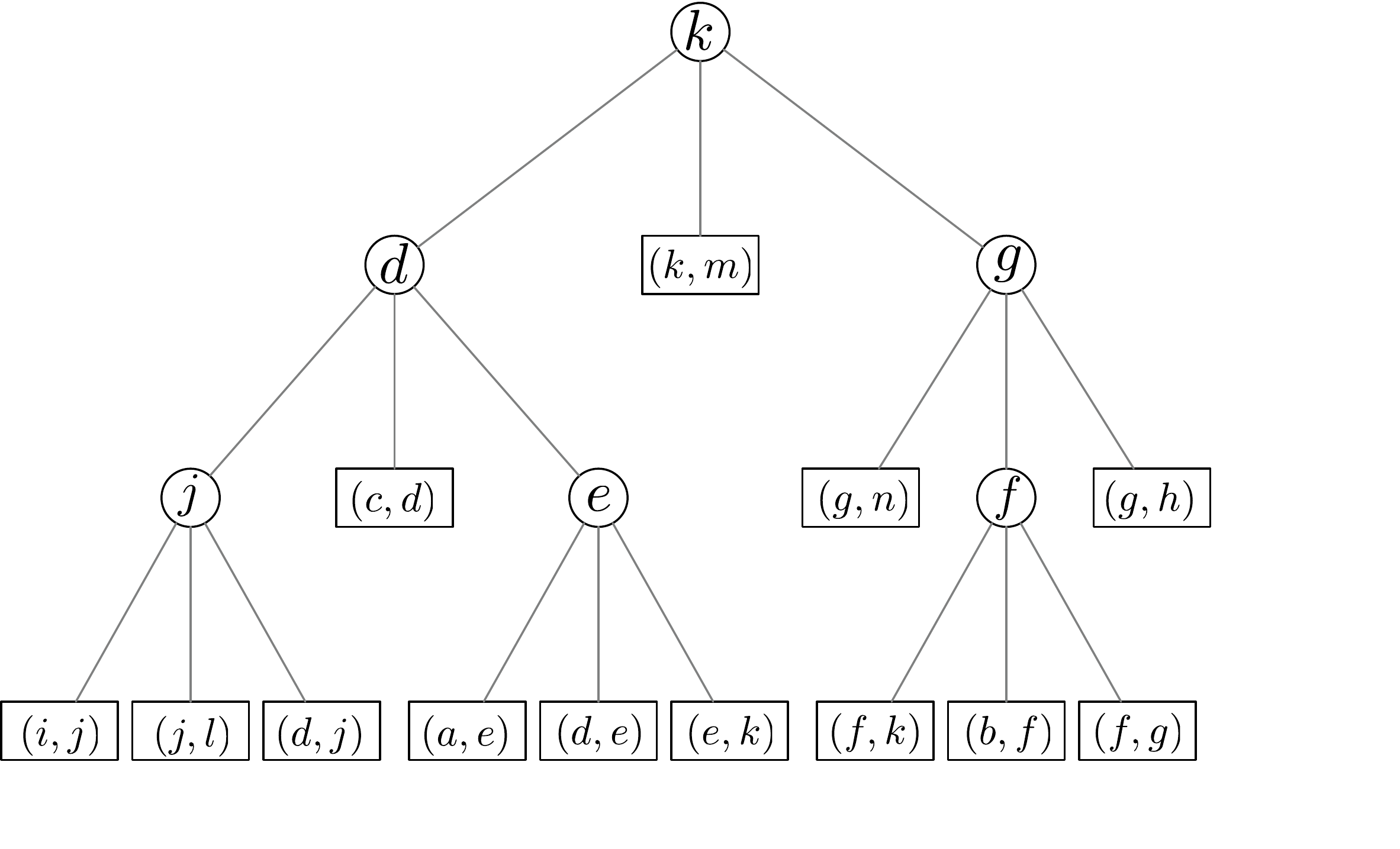}}\\
    &\\
    {\bf (b)} &{\bf (c)}
    \end{tabular}
    \caption{{\bf (a)} The original $H$ is a triangulated grid, 
    with $f$ being the exterior face. The boundary vertices $\hat{S}$
    with non-empty Voronoi cells are marked with colored halos. Edges
    are added so that $\hat{S}$ are on the exterior face $\hat{f}$.
    The vertices of $\VD^*$ are the duals of trichromatic
    faces, and those derived by splitting $\hat{f}^*$ into $|\hat{S}|$ vertices.  The edges of $\VD^*$
    correspond to paths of duals of bichromatic edges.
    {\bf (b)} The dual representation $\VD^*$.
    {\bf (c)} A centroid decomposition of $\VD^*$.}
    \label{fig:VD}
\end{figure}

The most important query on Voronoi diagrams is \emph{point location}.

\begin{lemma}\label{lem:PointLocate} 
(Gawrychowski et al.~\cite{GawrychowskiMWW18})
The $\PointLocate(\VD^*[H,S,\omega],v)$ function is given the dual
representation of a Voronoi diagram $\VD^*[H,S,\omega]$ and a vertex $v\in V(H)$
and reports the $s\in S$ for which $v\in \Vor(s)$.
Given access to an \MSSP{} data structure for $H$
with source-set $S$ and query time $\tau$,
we can answer $\PointLocate(\VD^*[H,S,\omega],v)$ 
queries in $O(\tau \cdot \log |H|)$ time.
\end{lemma}

The challenge in our data structure (as in~\cite{CharalampopoulosGMW19})
is to do point location when our space budget precludes storing
all the relevant \MSSP{} structures.  Nonetheless, we do make use
of $\PointLocate$ when the \MSSP{} data structures are available.

\section{The Distance Oracle}
\label{sect:DistanceOracle}
As in~\cite{CharalampopoulosGMW19}, 
the distance oracle is based on an $\vec{r}$-decomposition, 
$\vec{r}=(r_m,\ldots,r_1)$,
where $r_i = n^{i/m}$ and $m$ is a parameter.
Suppose we want to compute $\dist_G(u,v)$.  
Let $R_0 =\{u\}$ be the artificial level-0 region containing $u$
and $R_i \in \mathcal{R}_i$ be the level-$i$ ancestor of $R_0$.
(Throughout the paper, we will use ``$R_i$'' to refer specifically 
to the level-$i$ ancestor of $R_0=\{u\}$, 
as well as to a \emph{generic} region at level-$i$.  Surprisingly,
this will cause no confusion.)
Let $t$ be the smallest index for which $v\not\in R_{t}$ but $v\in R_{t+1}$.
Define $u_i$ to be the \emph{last} vertex on $\boundary R_i$ encountered on the shortest
path from $u$ to $v$.  The main task of the distance query algorithm is to compute
the sequence $(u=u_0,\ldots,u_t)$.  
Suppose that we know the identity of $u_{i}$ and $t > i$.
Finding $u_{i+1}$ now amounts to a point location problem in $\VD^*[R_{i+1}^{\out},\boundary R_{i+1},\omega]$,
where $\omega(s)$ is the distance from $u_i$ to $s\in \boundary R_{i+1}$.  
However, we cannot apply the fast
$\PointLocate$ routine because we cannot afford to store an \MSSP{} structure for 
every $(R_{i+1}^{\out},\boundary R_{i+1})$,
since $|R_{i+1}^{\out}|=\Omega(|G|)$.  Our point location routine narrows down the 
number of possibilities for $u_{i+1}$ to at most two candidates in $O(\kappa \log^{2+o(1)} n)$ time, 
then decides between them using two recursive distance queries, 
but starting at a higher level in the hierarchy.
There are about $2^m$ recursive calls in total, leading to a $O(2^m \kappa \log^{2+o(1)} n)$ 
query time.

The data structure is composed of several parts.
Parts (A) and (B) are explained below,
while parts (C)--(E) will be revealed in Section~\ref{sect:MoreDataStructures}.

\begin{enumerate}
    \item[(A)] {\bf (\MSSP{} Structures)}
    For each $i\in[0,m-1]$ and each region $R_i\in\mathcal{R}_i$ with parent $R_{i+1}\in\mathcal{R}_{i+1}$, 
    we store an \MSSP{} data structure (Lemma~\ref{lem:MSSP}) for 
    the graph $R_i^{\out}$, and source set $\boundary R_i$.
    However, the structure only answers queries 
    for $s\in\boundary R_i$ and $u,v \in R_i^{\out} \cap R_{i+1}$.
    Rather than represent the \emph{full} SSSP tree from each root on $s\in \boundary R_i$, 
    the \MSSP{} data structure only stores the tree induced by $R_i^{\out} \cap R_{i+1}$, i.e.,
    the parent of any vertex $v\in R_i^{\out} \cap R_{i+1}$ is its nearest ancestor $v'$
    in the SSSP tree such that $v' \in R_i^{\out} \cap R_{i+1}$.
    If $(v',v)$ is a ``shortcut'' edge corresponding to a path in $R_{i+1}^{\out}$, 
    it has weight $\dist_{R_i^{\out}}(v',v)$.
    
    We fix a $\kappa$ and let the update time in the dynamic tree data structure
    be $O(\kappa n^{1/\kappa})$ time.
    Thus, the space for this structure is
    $O((|R_i^{\out}\cap R_{i+1}| + |\boundary R_i|\cdot|\boundary R_{i+1}|)\cdot \kappa n^{1/\kappa}) 
    = O(r_{i+1}\cdot \kappa n^{1/\kappa})$
    since each edge in $R_i^{\out} \cap R_{i+1}$ is swapped into and out of
    the SSSP tree once~\cite{Klein05}, and the number of shortcut edges on 
    $\boundary R_{i+1}$ swapped into and out of the SSSP is at most 
    $|\boundary R_{i+1}|$ for each of the $|\boundary R_i|$ sources.  
    Over all $i$ and $\Theta(n/r_i)$ choices of $R_i$,
    the space is $O(m\kappa n^{1 + 1/m + 1/\kappa})$ since $r_{i+1}/r_i = n^{1/m}$.
    
    \item[(B)] {\bf (Voronoi Diagrams)}
    For each $i\in [0,m-1]$ and $R_{i}\in\mathcal{R}_{i}$ with parent $R_{i+1}\in \mathcal{R}_{i+1}$, 
    and each $q \in \boundary R_{i}$, define $\VDout(q,R_{i+1})$ to be
    $\VD^*[R_{i+1}^{\out},\boundary R_{i+1},\omega]$, with $\omega(s) = \dist_{G}(q,s)$.
    The space to store the dual diagram and its centroid decomposition is 
    $O(|\boundary R_{i+1}|)=O(\sqrt{r_{i+1}})$.
    Over all choices for $i,R_i,$ and $q$, the space is $O(mn^{1+1/(2m)})$
    since $\sqrt{r_{i+1}/r_i}=n^{1/(2m)}$.
\end{enumerate}

Due to our tie-breaking rule in the definition of $\Vor(\cdot)$, 
locating $u_{i+1}$ ($t\ge i+1$) is tantamount to 
performing a point location on a Voronoi diagram in part (B) of the data structure.
    
\begin{lemma}\label{lemma:lastsite}
Suppose that $q\in \boundary R_i$ and $v\not\in R_{i+1}$.
Consider the Voronoi diagram
associated with $\VDout(q,R_{i+1})$ with sites 
$\boundary R_{i+1}$ and additive weights 
defined by distances from $q$ in $G$.  
Then $v\in \Vor(s)$ if and only if $s$ is the \underline{last}
$\boundary R_{i+1}$-vertex on the shortest path from $q$ to $v$
in $G$, and $d^\omega(s,v) = \dist_G(q,v)$.
\end{lemma}

\begin{proof}
By definition, $d^\omega(s,v)$ is the length of the shortest path from $q$
to $v$ that passes through $s$ and whose $s$-$v$ suffix does not leave $R_{i+1}^{\out}$.
Thus, $d^\omega(s,v) \geq \dist_G(q,v)$ for every $s$,
and $d^\omega(s,v) = \dist_G(q,v)$ for some $s$.
Because of our assumption that all edges are strictly positive, 
and our tie-breaking rule for preferring larger $\omega$-values 
in the definition of $\Vor(\cdot)$, if $v\in \Vor(s)$ then
$s$ must be the \emph{last} $\boundary R_{i+1}$-vertex on the shortest
$q$-$v$ path.
\end{proof}

\subsection{The Query Algorithm}\label{sect:queryalg}

A distance query is given $u,v\in V(G)$. 
We begin by identifying the level-0 region 
$R_{0}=\{u\}\in \mathcal{R}_{0}$
and call the function $\Dist(u,v,R_0)$.
In general, the function $\Dist(u_i,v,R_i)$ takes as 
arguments a region $R_i$, a source vertex $u_i$ on the boundary
$\boundary R_i$, and a target vertex 
$v\not\in R_i$.  It returns a value $d$ such that
\begin{equation}\label{eqn:Spec}
\dist_{G}(u_i,v) \leq d \leq \dist_{R_{i}^{\out}}(u_i,v).
\end{equation}
Note that $R_{0}^{\out} = G$, so the initial call to this function correctly
computes $\dist_G(u,v)$.
When $v$ is ``close'' to $u_i$ ($v\in R_i^{\out}\cap R_{i+1}$)
it computes $\dist_{R_i^{\out}}(u_i,v)$ without recursion, using part (A) of the data structure.
When $v\in R_{i+1}^{\out}$ it performs point location using
the function $\CentroidSearch$, which culminates in up to 
two recursive calls to $\Dist$ on the level-$(i+1)$ region $R_{i+1}$.
Thus, the correctness of $\Dist$ hinges on 
whether $\CentroidSearch$ correctly
computes distances when $v\in R_{i+1}^{\out}$. 

\begin{algorithm}[H]
    \caption{$\Dist(u_{i},v,R_{i})$\label{alg:DistGlobal}}
    \begin{algorithmic}[1]
        \Require A region $R_{i}$, a source $u_{i}\in\boundary R_{i}$ and a destination $v\in R_{i}^{\out}$.
        \Ensure A value $d$ such that $\dist_G(u_i,v)\leq d\leq \dist_{R_i^{\out}}(u_i,v)$.
        \If{$v \in R_i^{\out}\cap R_{i+1}$} \Comment{I.e., $i=t$}
                \State \Return $d \gets \dist_{R_{i}^{\out}}(u_{i},v)$ \Comment{Part (A)}
        \EndIf  \Comment{$v\in R_{i+1}^{\out}$}
        \State $f^* \gets$ root of the centroid decomposition of $\VDout(u_{i},R_{i+1})$
        \State \Return $d \gets \CentroidSearch(\VDout(u_{i},R_{i+1}),v,f^*)$
    \end{algorithmic}
\end{algorithm}

The procedure $\CentroidSearch$ is given $u_i\in \boundary R_i$,
$v\in R_{i+1}^{\out}$, $\VDout = \VDout(u_i,R_{i+1})$ and a node $f^*$ on the centroid decomposition of $\VDout$.  
It ultimately computes $u_{i+1} \in \boundary R_{i+1}$
for which $v\in \Vor(u_{i+1})$ 
and returns
\begin{align*}
&\omega(u_{i+1}) + \Dist(u_{i+1},v,R_{i+1}) & \mbox{Line 5 or 9 of $\CentroidSearch$}\\
&\leq \dist_G(u_i,u_{i+1}) + \dist_{R_{i+1}^{\out}}(u_{i+1},v) & \mbox{Defn. of $\omega$; guarantee of $\Dist$ (Eqn.~(\ref{eqn:Spec}))}\\
&= \dist_{G}(u_i,v).                    & \mbox{Lemma~\ref{lemma:lastsite}}
\end{align*}
The algorithm is recursive, and bottoms out in one of two base cases (Line 5 or Line 9).
The first way the recursion can end is if we reach the bottom of the centroid decomposition.  
If $f^*$ is a leaf of the decomposition, it corresponds to an edge in $\VDout$ separating
the Voronoi cells of two sites, say $s_1$ and $s_2$.  At this
point we know that either $u_{i+1} = s_1$ or $u_{i+1} = s_2$, and determine
which case is true with two recursive calls to 
$\Dist(s_j,v,R_{i+1})$, $j\in\{1,2\}$ (Lines 2--5).
In general, $f^*$ is dual to a trichromatic face $f$ composed of 
three vertices $y_0,y_1,y_2$ in clockwise order, which are, respectively, 
in distinct Voronoi cells of $s_0,s_1,s_2$.
The three shortest $s_j$-$y_j$ paths and $f$ partition the vertices of $R_{i+1}^{\out}$ into six parts, 
namely the shortest $s_j$-$y_j$ paths themselves, and the interiors of the 
regions bounded by $\boundary R_{i+1}$, two of the $s_j$-$y_j$ paths and an edge of $f$.  
See Figure~\ref{fig:CentroidSearch}.
The $\Navigation$ function
returns a pair $(\flag,a^*)$ that identifies which part $v$ is in.
If $\flag=\operatorname{terminal}$ then $a^* \in \{s_0,s_1,s_2\}$ is interpreted as a site,
indicating that $v$ lies on the shortest path from $a^*$ to its $f$-vertex.
In this case we return $\omega(a^*) + \Dist(a^*,v,R_{i+1}) = \dist_G(u_i,v)$
with just one call to $\Dist$.
If $\flag=\operatorname{nonterminal}$ then $a^*$ is the correct child of $f^*$ in the centroid
decomposition.  In particular, $f^*$ is incident to three edges $e_0^*, e_1^*, e_2^*$ dual to 
$\{y_0,y_2\},\{y_1,y_0\},\{y_2,y_1\}$.
The children of $f^*$ in the centroid decomposition are $f_0^*,f_1^*,f_2^*$, with $f_j^*$ ancestral
to $e_j^*$.  We have $a^* = f_j^*$ if $v$ lies to the right of the chord 
$(s_j,\ldots,y_j,y_{j-1},\ldots,s_{j-1})$ in $R_{i+1}^{\out}$.
For example, in Figure~\ref{fig:CentroidSearch}, $v$ lies to the right of the 
$(s_0,\ldots,y_0,y_2,\ldots,s_2)$ path.
In this case we continue the search recursively from $a^* = f_0^*$.

\begin{figure}
    \centering
    \scalebox{.35}{\includegraphics{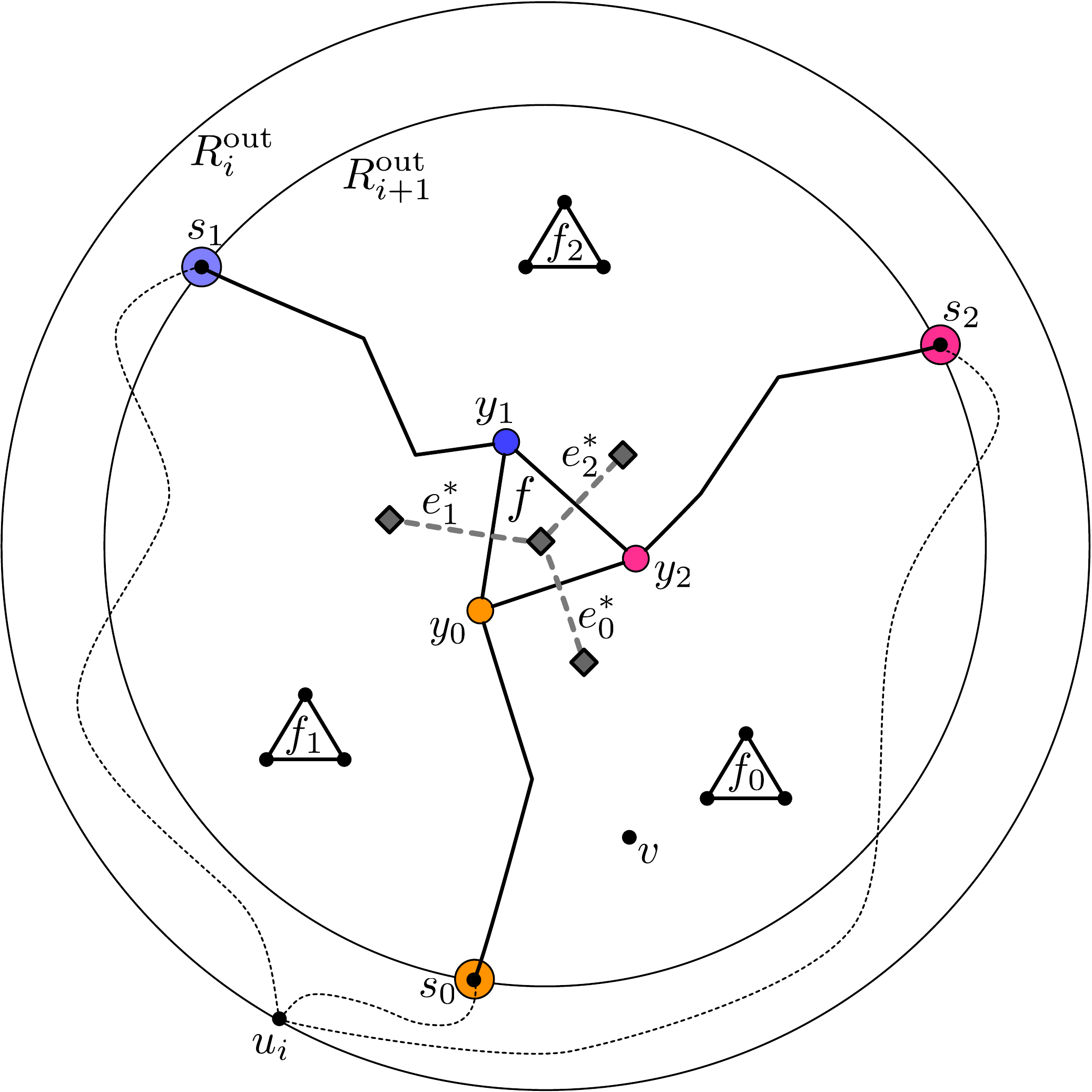}}
    \caption{Here $f^*$ is a degree-3 vertex in $\VDout(u_i,R_{i+1})$, 
            corresponding to a trichromatic face $f$ on vertices $y_0,y_1,y_2$,
            which are in the Voronoi cells of $s_0,s_1,s_2$ on the boundary
            $\boundary R_{i+1}^{\out}$.  The shortest $s_j$-$y_j$ paths
            partition $V(R_{i+1}^{\out})$ into six parts: the three shortest paths
            and the three regions bounded by them and $f$.  Let $e_0^*,e_1^*,e_2^*$ be the edges in $\VDout$ dual to $\{y_0,y_2\},\{y_1,y_0\},\{y_2,y_1\}$.  
            In the centroid
            decomposition $e_0^*,e_1^*,e_2^*$ are in separate subtrees of $f^*$.
            Let $f_j^*$ be the child of $f^*$ ancestral to $e_j^*$, which is either $e_j^*$ itself, or a trichromatic face to the right of the ``chord''
            $(s_j,\ldots, y_j,y_{j-1}, \ldots, s_{j-1})$.  $\CentroidSearch$
            locates the site whose Voronoi cell contains $v$ via recursion. 
            It calls $\Navigation$, a function that finds which 
            of the 6 parts contains $v$. 
            If $v$ lies on an $s_j$-$y_j$ path the $\CentroidSearch$ recursion terminates; otherwise it recurses 
            on the correct child $f_j^*$ of $f^*$.}
    \label{fig:CentroidSearch}
\end{figure}

\begin{algorithm}[H]
    \caption{$\CentroidSearch(\VDout(u_{i},R_{i+1}),v,f^{*})$}
    \label{alg:SpecificCentroidSearch}
    \begin{algorithmic}[1]
        \Require The dual representation $\VDout = \VDout(u_{i},R_{i+1})$ 
        of a Voronoi diagram with additive weights $\omega(s)=\dist_{G}(u_i,s)$, 
        a vertex $v\in R_{i+1}^{\out}$, 
        and a node $f^{*}$ in the centroid decomposition of $\VDout$.
        \Ensure The distance $\dist_G(u_i,v)$.
        \If {$f^{*}$ is a leaf in the centroid decomposition 
            (an edge in $\VDout$)}
            \State $s_{1},s_{2} \gets$ sites whose Voronoi cells are bounded by $f^{*}$ \Comment{Candidates for $u_{i+1}$}
            \State $d_{1}\gets \omega(s_{1}) + \Dist(s_{1},v,R_{i+1})$
            \State $d_{2}\gets \omega(s_{2}) + \Dist(s_{2},v,R_{i+1})$
            \State \Return $\min(d_{1},d_{2})$
        \EndIf
        \State $(\flag, a^*)\gets \Navigation(\VDout(u_{i},R_{i+1}), v, f^{*})$
        \If {$\flag = \operatorname{terminal}$}   \Comment{$a^*$ is interpreted as a site}
            \State \Return $\omega(a^*) + \Dist(a^*, v, R_{i+1})$   \Comment{$a^* = u_{i+1}$}
        \Else \; (i.e., $\flag = \operatorname{nonterminal}$)  \Comment{$a^*$ is the centroid child of $f^*$}
            \State \Return $\CentroidSearch(\VDout(u_{i},R_{i+1}),v,a^*)$
        \EndIf
    \end{algorithmic}
\end{algorithm}

\begin{lemma}
$\CentroidSearch$ correctly computes $\dist_{G}(u_i,v)$ 
\end{lemma}

\begin{proof}
Define $f,y_j,s_j,e_j^*,f_j^*$ as usual, and let $u_{i+1}$ be such that $v\in \Vor(u_{i+1})$.
The loop invariant is that in the subtree of the 
centroid decomposition rooted at $f^*$, 
there is \emph{some} leaf edge on the boundary of the cell $\Vor(u_{i+1})$.  
This is clearly true in the intial recursive call, 
when $f^*$ is the root of the centroid decomposition.  
Suppose that $\Navigation$ tells us that $v$ lies to the right of the oriented chord
$C^\star = (s_j,\ldots,y_j,y_{j-1},\ldots,s_{j-1})$.  
Observe that since the $s_j$-$y_j$ and $s_{j-1}$-$y_{j-1}$ 
shortest paths are monochromatic, 
all edges of the centroid decomposition correspond to paths in 
$G^*$ that lie strictly to the left or right of $C^\star$, with the exception of $e_j^*$.
Moreover, since $v\in \Vor(u_{i+1})$, $\Vor(u_{i+1})$ must be bounded by \emph{some} edge that is either 
$e_j^*$ or one entirely to the right of $C^\star$, from which it follows that $f_j^* = a^*$
is ancestral to at least one edge bounding $\Vor(u_{i+1})$.
When $f^*$ is a single edge on the boundary of $\Vor(s_1),\Vor(s_2)$
the loop invariant guarantees that either $u_{i+1}=s_1$ or $u_{i+1}=s_2$;
suppose that $u_{i+1}=s_1$.
It follows from the specification of $\Dist$ (Eqn.~(\ref{eqn:Spec})) and Lemma~\ref{lemma:lastsite} that
\[
d_1 = \omega(s_1) + \Dist(s_1,v,R_{i+1})
\leq \dist_G(u_i,s_1) + \dist_{R_{i+1}^{\out}}(s_1,v)
= \dist_G(u_i,v).
\]
Furthermore,
\[
d_2 = \omega(s_2) + \Dist(s_2,v,R_{i+1})
\geq \dist_G(u_i,s_2) + \dist_{G}(s_2,v) \geq \dist_G(u_i,v),
\]
so in this base case $\CentroidSearch$ correctly returns $d_1=\dist_G(u_i,v)$.
If $\Navigation$ ever reports that $v$ is on an $s_j$-$y_j$ path, 
then by definition $v\in \Vor(s_j)$. 
By the specification of $\Dist$ (Eqn.~(\ref{eqn:Spec})) 
and Lemma~\ref{lemma:lastsite} we have
\[
\omega(s_j) + \Dist(s_j,v,R_{i+1})
\leq \dist_G(u_i,s_j) + \dist_{R_{i+1}^{\out}}(s_j,v)
= \dist_G(u_i,v)
\]
and the base case on Lines 8--9 also works correctly.
\end{proof}

\medskip

Thus, the main challenge is to design an efficient $\Navigation$ function, 
i.e., to solve the restricted point location problem in $R_{i+1}^{\out}$ 
depicted in Figure~\ref{fig:CentroidSearch}.  Whereas
Charalampopoulos et al. \cite{CharalampopoulosGMW19} solve this problem
using several \emph{more} recursive calls to $\Dist$, 
we give a new method to do this point location directly, in 
$O(\kappa \log^{1+o(1)} n)$ time per call to $\Navigation$.

\section{The Navigation Oracle}\label{sect:NavigationOracle}

The input to $\Navigation$ is the same as $\CentroidSearch$, except that
$f^*$ is guaranteed to correspond to a trichromatic face $f$.
Define $y_j,s_j,e_j,f_j$, $j\in\{0,1,2\}$ as in the discussion of $\CentroidSearch$.
The $\Navigation$ function determines the location of $v$ relative
to $f$ and the shortest $s_j$-$y_j$ paths.  It delegates nearly all the actual computation
to two functions: $\SitePathIndicator$, which returns a boolean indicating whether $v$ is on 
the shortest $s_j$-$y_j$ path, and $\ChordIndicator$, 
which indicates whether $v$ lies strictly
to the right of the oriented chord $(s_j,\ldots,y_j,y_{j-1},\dots,s_{j-1})$.  
If so, we return the centroid child $f_j^*$ of $f^*$ in this region.  
Three calls each to $\SitePathIndicator$
and $\ChordIndicator$ suffice to determine the location of $v$.

\begin{algorithm}[H]
    \caption{$\Navigation(\VDout(u_{i},R_{i+1}),v,f^{*})$}
    \label{alg:Navigation}
    \begin{algorithmic}[1]
    \Require The dual representation $\VDout(u_{i},R_{i+1})$ of a Voronoi diagram,
    a vertex $v\in R_{i+1}^{\out}$, and a centroid $f^{*}$ in the centroid decomposition.
    The face $f$ is on $y_0,y_1,y_2$, which are in the Voronoi cells of $s_0,s_1,s_2$,
    and $f_j^*$ is the child of $f^*$ containing the edge dual to $\{y_j,y_{j-1}\}$.
    \Ensure $(\operatorname{terminal}, s_{j})$ if $v$ is on the shortest $s_j$-$y_j$
    path, or $(\operatorname{nonterminal}, f_j^*)$ where $f_j^*$ is the 
    child of $f^*$ ancestral to an edge bounding $v$'s Voronoi cell.
    \State $s_{0},s_{1},s_{2}\gets$ sites corresponding to $f^{*}$
    \For {$j = 0,1,2$}
        \If {$\SitePathIndicator(\VDout(u_{i},R_{i+1}),v,f^{*},j)$ returns \textbf{True}}
            \State \Return $(\operatorname{terminal}, s_{j})$
        \EndIf
    \EndFor
    \For {$j = 0,1,2$}
        \If {$\ChordIndicator(\VDout(u_{i},R_{i+1}),v,f^{*},j)$ returns \textbf{True}}
            \State \Return $(\operatorname{nonterminal}, f_{j}^{*})$
        \EndIf
    \EndFor
    \end{algorithmic}
\end{algorithm}

In Section~\ref{sect:chords} we formally introduce the notion of \emph{chords}
used informally above, as well as some related concepts like \emph{laminar} sets of chords
and \emph{maximal} chords. 
In Section~\ref{sect:MoreDataStructures} we introduce parts (C)-(E) of the 
data structure used to support $\Navigation$.
The functions $\SitePathIndicator$ and $\ChordIndicator$
are presented in Sections~\ref{sect:SitePathIndicator} and \ref{sect:ChordIndicator}.

\subsection{Chords and Pieces}\label{sect:chords}

We begin by defining the key concepts of our point location method: 
\emph{chords}, \emph{laminar chord sets}, \emph{pieces},
and the \emph{occludes} relation.  
\begin{definition}\label{def:chord}
{\bf (Chords)}
Fix an $R$ in the $\vec{r}$-decomposition and two vertices $c_0,c_1\in \boundary R$.
An oriented simple path $\chord{c_0c_1}$ is a \emph{chord} of $R^{\out}$
if it is contained in $R^{\out}$ and is internally vertex-disjoint
from $\boundary R$.
When the orientation is irrelevant we write it as $\overline{c_0c_1}$.
\end{definition}

\begin{definition}\label{def:laminar}
{\bf (Laminar Chord Sets)}
A set of chords $\mathcal{C}$ for $R^{\out}$ is \emph{laminar} (non-crossing)
if for any two such chords $C=\chord{c_0c_1},C'=\chord{c_2c_3}$,
if there exists a $v\in (C\cap C')-\boundary R$ then
the subpaths from $c_0$ to $v$ and from $c_2$ to $v$ are identical;
in particular $c_0=c_2$.
\end{definition}

The orientation of chords does not always coincide with a natural orientation
of paths defined by the algorithm.  For example, in Figure~\ref{fig:CentroidSearch},
the oriented chord $\chord{s_0s_2} = (s_0,\ldots,y_0,y_2,\ldots,s_2)$ is
composed of three parts: a shortest $s_0$-$y_0$ path (whose natural orientation
coincides with that of $\chord{s_0s_2}$),
the edge $\{y_0,y_2\}$ (which has no natural orientation in this context),
and the shortest $s_2$-$y_2$ path (whose natural orientation is the
reverse of its orientation in $\chord{s_0s_2}$).
The orientation serves two purposes.  In Definition~\ref{def:chord}
we can speak unambiguously about the parts of $R^{\out}$ to the \emph{right}
and \emph{left} of $\chord{s_0s_2}$.  In Definition~\ref{def:laminar} the
role of the orientation is to ensure that the partition of 
$R^{\out}$ into \emph{pieces} induced by $\mathcal{C}$ can be represented by a \emph{tree},
as we show in Lemma~\ref{lem:piecetree}.

\begin{definition}\label{def:pieces}
{\bf (Pieces)}
A laminar chord set $\mathcal{C}$ for $R^{\out}$
partitions the faces of $R^{\out}$ into pieces, excluding the face on $\boundary R$. 
Two faces $f,g$ are in the same piece iff
$f^*$ and $g^*$ are connected by a path in $(R^{\out})^*$ that avoids
to duals of edges in $\mathcal{C}$ and edges along the boundary cycle on $\boundary R$.
A piece is regarded as the subgraph induced by its faces, 
i.e., it includes their constituent vertices and edges.
Two pieces $P_1,P_2$ are \emph{adjacent} if there is an edge
$e$ on the boundary of $P_1$ and $P_2$ and $e$ is in a \underline{\emph{unique}}
chord of $\mathcal{C}$.  See Figure~\ref{fig:PieceTree}.
\end{definition}

\begin{figure}
    \centering
    \scalebox{.4}{\includegraphics{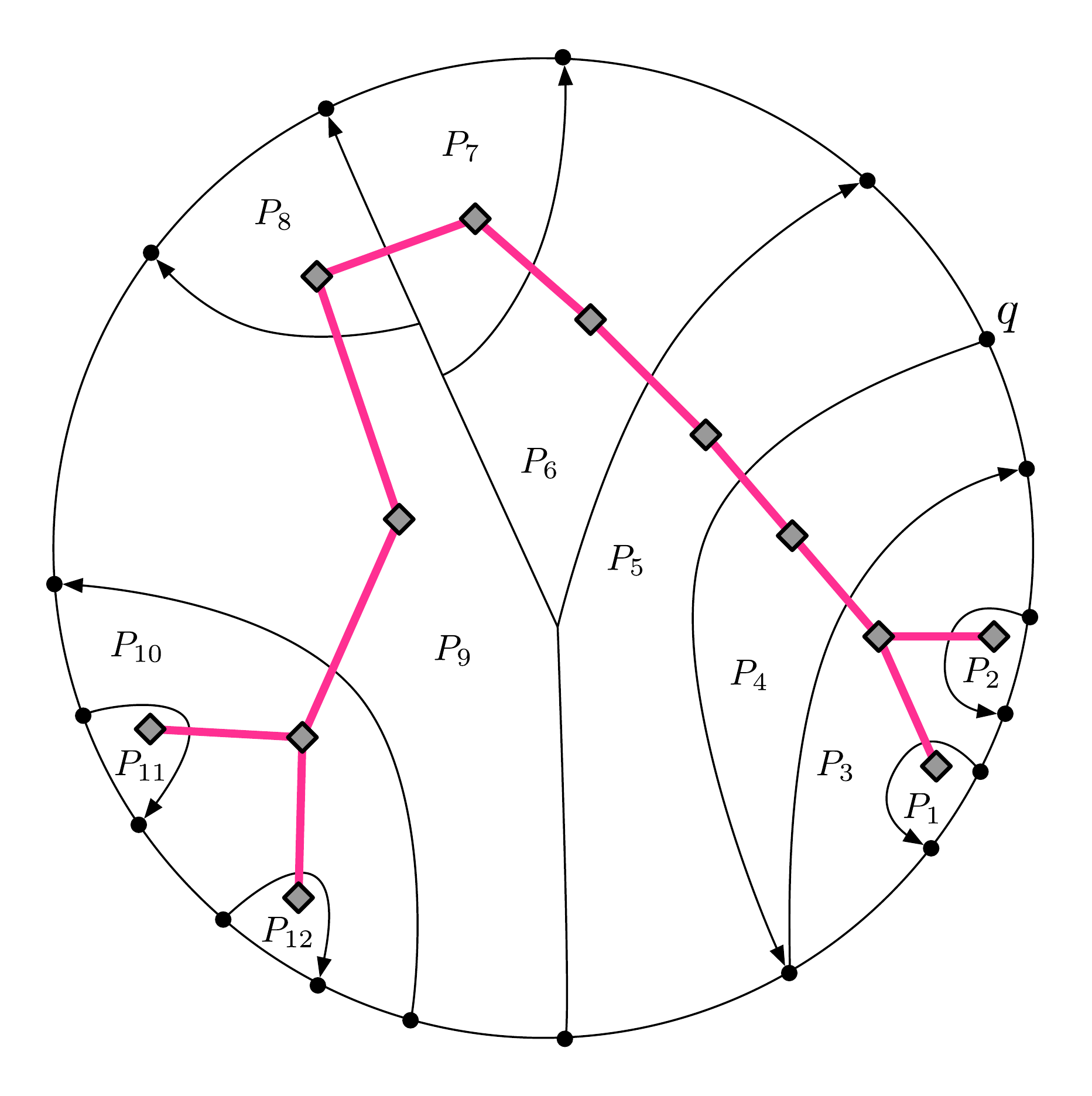}}
    \caption{A laminar set of chords partition $R^{\out}$ into pieces.
    Observe that the chords separating pieces $P_5$--$P_9$ overlap
    in certain prefixes.  The piece tree is indicated by diamond
    verties and pink edges.  Note two pieces (e.g. $P_5$ and $P_9$) 
    may share a boundary, but not be \emph{adjacent}.}
    \label{fig:PieceTree}
\end{figure}

\begin{lemma}\label{lem:piecetree}
Suppose $\mathcal{C}$ is a laminar chord set for $R^{\out}$,
$\mathcal{P}=\mathcal{P}(\mathcal{C})$ is the corresponding piece set and 
$\mathcal{E}$ are pairs of adjacent pieces.  Then $\mathcal{T}=(\mathcal{P},\mathcal{E})$
is a tree, called the \emph{piece tree} induced by $\mathcal{C}$.
\end{lemma}

\begin{proof}
The claim is clearly true when $\mathcal{C}$ contains zero or one chords,
so we will try to reduce the general case to this case via a peeling argument.
We will find a piece $P$ with degree 1 in $\mathcal{T}$, 
remove it and the chord bounding it,
and conclude by induction that the truncated instance is a tree.
Reattaching $P$ implies $\mathcal{T}$ is a tree.

Let $C=\chord{c_0c_1}\in\mathcal{C}$ be a chord such that no edge
of any other chord appears strictly to one side of $C$, say 
to the right of $C$.  Let $P$ be the piece to the right of $C$.
(In Figure~\ref{fig:PieceTree} the chords bounding $P_1,P_2,P_{11},P_{12}$ would be eligible to be $C$.)
Let $C=(c_0=v_0,v_1,v_2,\ldots,v_k=c_1)$ and $v_{j^\star}$
be such that the edges of the suffix $(v_{j^\star},\ldots,v_k)$ are on no other chord,
meaning the vertices $\{v_{j^\star+1},\ldots,v_{k-1}\}$ are on no other chord.
Let $g_j$ be the face to the left of $(v_j,v_{j+1})$.  It follows
that there is a path from $g_{j^\star}^*$ to $g_{k-1}^*$ in $(R^{\out})^*$ 
that avoids the duals of all edges in $\mathcal{C}$ and along $\boundary R$.
All pieces adjacent to $P$ contain some face among 
$\{g_{j^\star},\ldots,g_{k-1}\}$, but these are in a single piece, 
hence $P$ corresponds to a degree-1 vertex in $\mathcal{T}$.
Let $P$ be bounded by $C$ and an interval $B$ of the boundary cycle on $\boundary R$.
Obtain the ``new'' $R^{\out}$ by cutting along $C$ and removing $P$,
the new $\boundary R$ by substituting $C$ for $B$, and the new
chord-set $\mathcal{C}$ by removing $C$ and trimming any chords
that shared a non-empty prefix with $C$.  By induction the resulting
piece-adjacency graph is a tree; reattaching $P$ as a degree-1 vertex
shows $\mathcal{T}$ is a tree.
\end{proof}

\begin{definition}\label{def:occludes} 
{\bf (Occludes Relation)}
Fix $R^{\out}$, chord $C$, and two faces $f,g$, neither of which 
is the hole defined by $\boundary R$.  
If $f$ and $g$ are on opposite sides of $C$, 
we say that from vantage $f$, $C$ \emph{occludes} $g$.
Let $\mathcal{C}$ be a set of chords.
We say $C\in \mathcal{C}$ is \emph{maximal} in $\mathcal{C}$ 
with respect to a vantage $f$
if there is no $C' \in \mathcal{C}$ such that $C'$ occludes a \emph{strict}
superset of the faces that $C$ occludes.  (Note that the
orientation of chords is irrelevant to the occludes relation.)
\end{definition}

It follows from Definition~\ref{def:occludes} that if $\mathcal{C}$ 
is laminar, the set of maximal chords with respect to $f$ are exactly 
those chords intersecting the boundary of $f$'s piece in $\mathcal{P}(\mathcal{C})$.

We can also speak unambiguously about a chord $C$ occluding a \emph{vertex} 
or \emph{edge} not on $C$,
from a certain vantage.  Specifically, we can say that from some vantage, $C$
occludes an \emph{interval} of the boundary cycle on $\boundary R$,
say according to a clockwise traversal around the hole on $\boundary R$ in $R^{\out}$.\footnote{This is one 
place where we use the assumption that all boundary holes are simple cycles.}
This will be used in the $\ChordIndicator$ procedure of Section~\ref{sect:ChordIndictor-procedure}.

\subsection{Data Structures for Navigation}\label{sect:MoreDataStructures}

Parts (C)--(E) of the data structure are used to implement the $\SitePathIndicator$ and $\ChordIndicator$ functions.
\begin{enumerate}
    \item[(C)] \textbf{(More Voronoi Diagrams)} 
    For each $i$, each $R_{i}\in \mathcal{R}_{i}$, 
    and each $q\in\boundary R_{i}$, 
    we store $\VDout(q,R_i)$, which is 
    $\VD^*[R_i^{\out},\boundary R_i,\omega]$, 
    where $\omega(s)=\dist_G(q,s)$.
    The total space for these diagrams is 
    $\tilde{O}(n)$ and dominated by part (B).

    \item[(D)] \textbf{(Chord Trees; Piece Trees)} 
    For each $i$, each $R_i\in\mathcal{R}_i$, and source $q\in \boundary R_i$, we store the SSSP tree 
    from $q$ induced by $\boundary R_i$
    as a \emph{chord tree} $T_{q}^{R_i}$.  
    In particular, the parent of $x\in\boundary R_i$ in $T_{q}^{R_i}$ is the nearest ancestor
    in the SSSP tree from $q$ that lies on $\boundary R_i$.  
    Every edge of $T_{q}^{R_i}$ is designated a \emph{chord} if the corresponding path is contained
    in $R_{i}^{\out}$ but not in $R_i$, or a \emph{non-chord} otherwise.
    Define $\mathcal{C}_{q}^{R_i}$ to be the set of all chords in $T_{q}^{R_i}$, 
    oriented away from $q$; 
    this is clearly a laminar set since shortest paths are unique and all 
    prefixes of shortest paths are shortest paths. 
    Define $\mathcal{P}_{q}^{R_i}$ to be the corresponding partition of $R_{i}^{\out}$ into pieces,
    and $\mathcal{T}_{q}^{R_i}$ the corresponding piece tree.
    Define $T_{q}^{R_i}[x]$ to be the path from $q$ to $x$ in $T_{q}^{R_i}$, 
    $\mathcal{C}_{q}^{R_i}[x]$ the corresponding chord-set,
    and $\mathcal{P}_{q}^{R_i}[x]$ the corresponding piece-set.
    
    The data structure answers the following queries
    \begin{description}
    \item[$\MaximalChord(R_i,q,x,P,P')$:] We are given $R_i$, $q,x\in\boundary R_i$, a piece $P\in \mathcal{P}_{q}^{R_i}$, 
    and possibly another piece $P'\in \mathcal{P}_{q}^{R_i}$
    (which may be \textbf{Null}).
    If $P'$ is \textbf{Null}, return any maximal chord in $\mathcal{C}_{q}^{R_i}[x]$ from vantage $P$.
    If $P'$ is not \textbf{Null}, return the maximal chord $\mathcal{C}_{q}^{R_i}[x]$ (if any) 
    that occludes $P'$ from vantage $P$.
    \item[$\AdjacentPiece(R_i,q,e)$:] Here $e$ is an edge on the boundary cycle on $\boundary R_i$.
    Return the unique piece in $\mathcal{P}_{q}^{R_i}$ with $e$ on its boundary.\footnote{This is another place where we use the assumption that holes are bounded by simple cycles.}
    \end{description}
    
    \item[(E)] \textbf{(Site Tables; Side Tables)}
    Fix an $i$ and a diagram $\VDout = \VDout(u',R_i)$ from part (B) or (C).
    Let $f^*$ be any node in the centroid decomposition of $\VDout$,
    with $y_j,s_j$, $j\in\{0,1,2\}$ defined as usual,
    and let $R_{i'}\in \mathcal{R}_{i'}$ be the ancestor of $R_i$, $i'\ge i$.  
    Fix $j\in\{0,1,2\}$ and $i'>i$.  Define $q$ and $x$ to be the \emph{first}
    and \emph{last} vertices on the shortest $s_j$-$y_j$ path that lie on $\boundary R_{i'}$.
    We store $(q,x)$ and $\dist_{G}(u',x)$.
    
    We also store whether $R_{i'}^{\out}$ lies to the left or right 
    of the site-centroid-site chord $\chord{s_jy_jy_{j-1}s_{j-1}}$ 
    in $R_i^{\out}$, or {\bf Null} if the relationship cannot be determined, i.e., 
    if the chord crosses $\boundary R_{i'}$.
    These tables increase the space of (B) and (C) by a negligible $O(m)$ factor.
\end{enumerate}

Part (D) of the data structure is the only one that is non-trivial to store compactly.  
Our strategy is as follows.
We fix $R_i$ and $q\in \boundary R_i$ and build a dynamic data structure for these 
operations relative to a dynamic subset 
$\hat{\mathcal{C}} \subseteq \mathcal{C}_{q}^{R_i}$ subject to the insertion and deletion of chords 
in $O(\log|\boundary R_i|)$ time.
By inserting/deleting $O(|\boundary R_i|)$ chords in the correct order, we can arrange that
$\hat{\mathcal{C}} = \mathcal{C}_{q}^{R_i}[x]$ at some point in time, for every $x\in \boundary R_i$.  
Using the generic persistence technique for RAM data structures (see~\cite{Dietz89})
we can answer $\MaximalChord$ queries relative to $\mathcal{C}_{q}^{R_i}[x]$ in 
$O(\log|\boundary R_i|\log\log |\boundary R_i|)$ time.\footnote{Our data structure works
in the pointer machine model, but it has unbounded in-degrees so the theorem
of Driscoll et al.~\cite{DriscollSST89,SarnakT86} cannot be applied directly.
It is probably possible to improve the bound to $O(\log|\boundary R_i|)$ but this
is not a bottleneck in our algorithm.}

\begin{lemma}\label{lem:partD}
    Part (D) of the data structure can be stored in $O(mn\log n)$ total space and
    answer $\MaximalChord$ queries in $O(\log n\log\log n)$ time and $\AdjacentPiece$ queries 
    in $O(1)$ time.
\end{lemma}

\begin{proof}
We first address $\MaximalChord$.
Let $\mathcal{T} = \mathcal{T}_{q}^{R_i}$ be the piece tree.  
The edges of $\mathcal{T}$ are in 1-1 correspondence with the chords of $\mathcal{C}=\mathcal{C}_{q}^{R_i}$
and if $P,P'\in \mathcal{P}=\mathcal{P}_{q}^{R_i}$ are two pieces, the path from $P$ to $P'$ in $\mathcal{T}$
crosses exactly those chords that occlude $P'$ from vantage $P$ (and vice versa).
We will argue that to implement $\MaximalChord$ it suffices to design an efficient 
dynamic data structure for the following
problem; initially all edges are \emph{unmarked}.
\begin{description}
\item[Mark$(e)$] Mark an edge $e\in E(\mathcal{T})$.
\item[Unmark$(e)$] Unmark $e$.  
\item[LastMarked$(P',P)$] Return the \emph{last} marked edge on the path from $P'$ to $P$, or \textbf{Null} if all are unmarked.
\end{description}
By doing a depth-first traversal of the chord tree $T_{q}^{R_i}$, marking/unmarking chords
as they are encountered, the set $\{e\in E(\mathcal{T}) \mid e \mbox{ is marked}\}$ will
be equal to $\mathcal{C}_{q}^{R_i}[x]$ precisely 
when $x$ is first encountered in DFS.
To answer a $\MaximalChord(R_i,q,x,P,P')$ query we interact 
with the state of the data structure
when the marked set is $\mathcal{C}_{q}^{R_i}[x]$.
If $P'$ is not null we return {\bf LastMarked}$(P',P)$.  
Otherwise we pick an arbitrary (marked)
chord $C\in \mathcal{C}_{q}^{R_i}[x]$, get the adjacent pieces 
$P_1',P_2'$ on either side of $C$, 
then query {\bf LastMarked}$(P_1',P)$ and {\bf LastMarked}$(P_2',P)$.  
At least one of these queries
will return a chord
and that chord is maximal w.r.t.~vantage $P$.
(Note that $C$ must separate $P$ from either
$P_1'$ or $P_2'$.)

We now argue how all three operations can be implemented in $O(\log n)$ worst case time.
The ideas are standard, so we do not go into great detail.  Root $\mathcal{T}$ arbitrarily
and subdivide every edge; the resulting tree is also called $\mathcal{T}$.
Every node in $\mathcal{T}$ knows its depth.
The \emph{vertices} corresponding to subdivided edges may carry marks.
In order to answer {\bf LastMarked} queries it suffices to be able to find least common 
ancestors, and, given nodes $P_d,P_a$, where $P_a$ is an ancestor of $P_d$, to find
the \emph{first \underline{and} last} marked node on the path from $P_d$ to $P_a$.
Decompose the vertices of $\mathcal{T}$ using a heavy path decomposition.
Each vertex points to the path that it is in.
Each path in the decomposition is a data structure that maintains an ordered
set of its marked nodes, 
a pointer to the most ancestral marked node in the path,
and a pointer to the parent, in $\mathcal{T}$, of the root of the path.  It is straightforward to 
find LCAs in this structure in $O(\log n)$ time.\footnote{Of course, $O(1)$ time is 
also possible~\cite{HarelT84,BenderF00} but this is not the bottleneck in the algorithm.}
Suppose we want the first and last marked node on the path from $P_d$ to $P_a$, an ancestor of $P_d$.
Let $Z_0,Z_1,...,Z_{\ell}$, $\ell=O(\log n)$ be the heavy paths ancestral to $P_d$
such that $P_d\in Z_0, P_a\in Z_{\ell}$.  Let $v_j\in Z_j$ be the nearest ancestor to $P_d$.
We can find $j^\star$ such that $Z_{j^\star}$ contains the first marked node on the $P_d$-$P_a$ path 
in $\ell=O(\log n)$ time
by comparing $v_j$ against the most ancestral marked node in $Z_j$, $j=0,1,\ldots$.
We can then find the first marked node by 
finding the marked predecessor of $v_{j^\star}$ in $Z_{j^\star}$,
in $O(\log n)$ time.  Finding the last marked node on the path from $P_d$ to $P_a$ is similar.
{\bf Mark} and {\bf Unmark} are implemented by keeping a balanced binary search tree 
over the marked nodes in each heavy path.

For fixed $R_i,q\in \boundary R_i$ there are $O(|\boundary R_i|)$ \textbf{Mark} and \textbf{Unmark}
operations, each taking $O(\log n)$ time.  Over all choices of $i, R_i$, and $q$ the total update
time is $O(mn\log n)$.  After applying generic persistence for RAM data structures (see~\cite{Dietz89})
the space becomes $O(mn\log n)$ and the query time for \textbf{LastMarked} becomes $O(\log n\log\log n)$. 

Turning to $\AdjacentPiece(R_i,q,e)$, there are $|\boundary R_i|^2$ choices of $(q,e)$.
Hence all answers can be precomputed in a lookup table in $O(mn)$ space.

\end{proof}

\subsection{The $\SitePathIndicator$ Function}\label{sect:SitePathIndicator}

The $\SitePathIndicator$ function is relatively simple. 
We are given $\VDout(u_i,R_{i+1})$, $v\in R_{i+1}^{\out}$,
a centroid $f^* \in R_{i+1}^{\out}$, $f$ being a trichromatic 
face on $y_0,y_1,y_2$, which are, respectively, in the Voronoi cells of $s_0,s_1,s_2\in \boundary R_{i+1}$,
and an index $j\in\{0,1,2\}$.  We would like to know if $v$ 
is on the shortest $s_j$-to-$y_j$ path.
Recall that $t$ is such that $v\not\in R_t$ but $v\in R_{t+1}$.

\begin{figure}
    \centering
    \begin{tabular}{c@{\hspace*{1cm}}c}
    \scalebox{.30}{\includegraphics{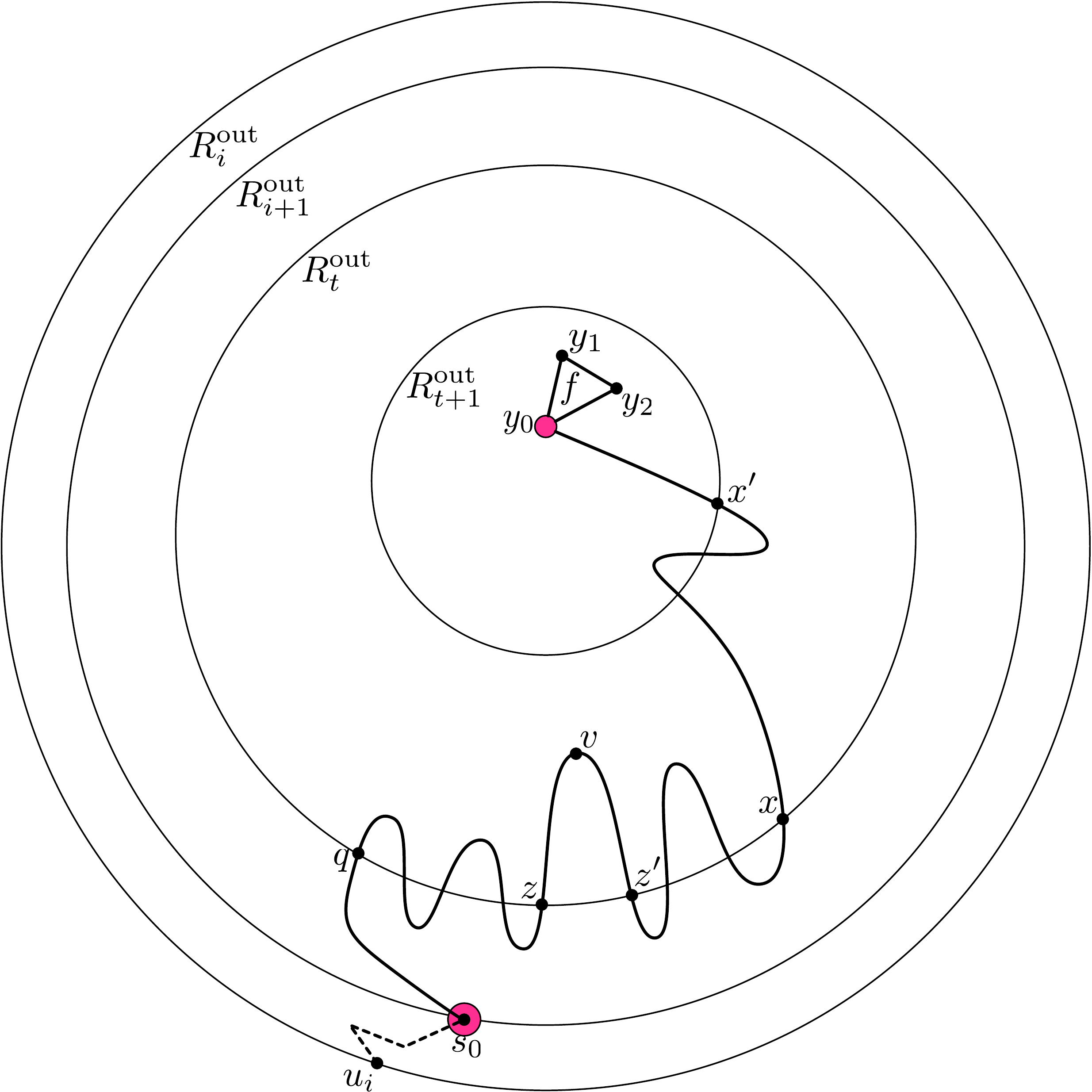}}
&
    \scalebox{.30}{\includegraphics{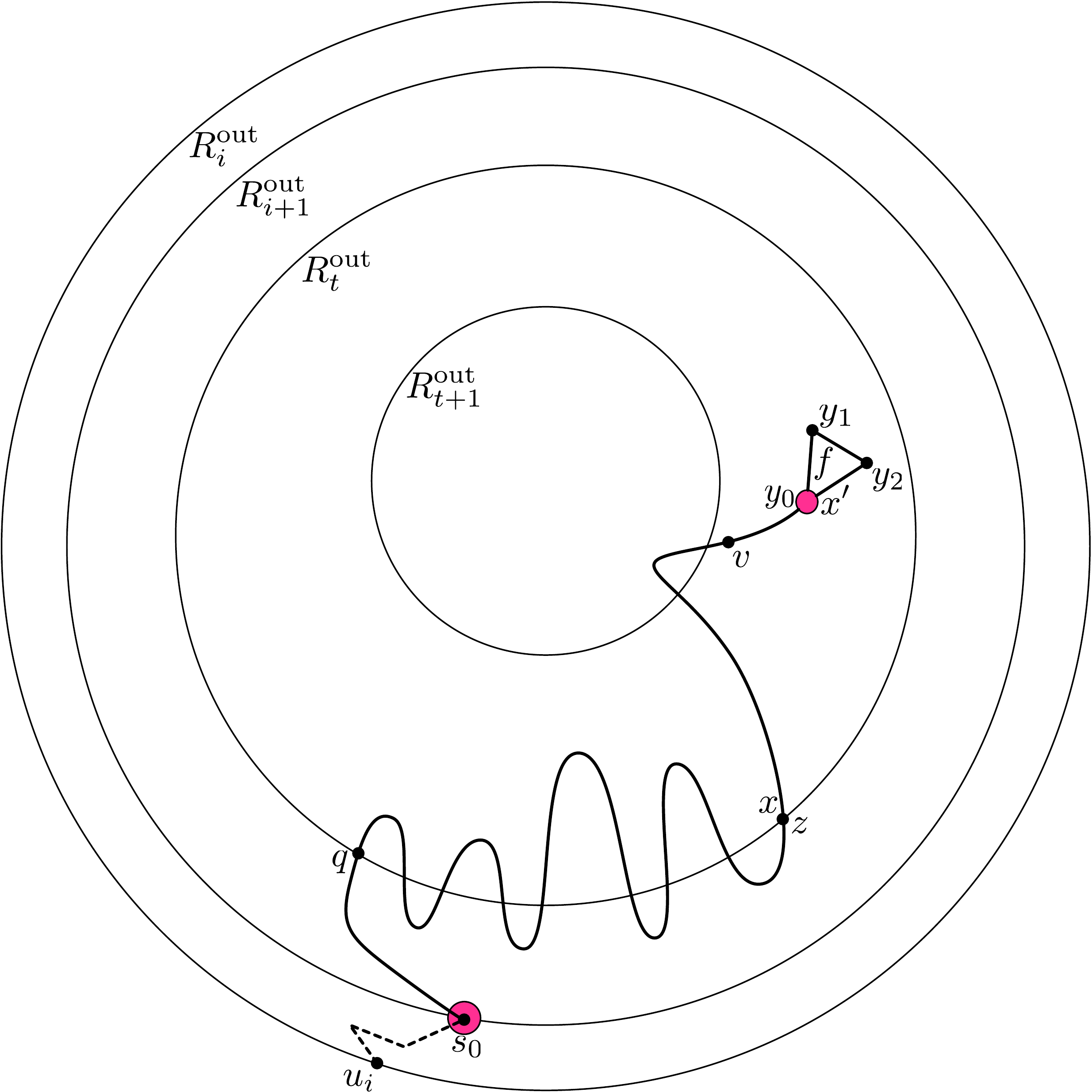}}\\ & \\
    {\bf (a)} & {\bf (b)}
    \end{tabular}
    \caption{{\bf (a)} If $z=x$ and $y_j$ is not in $R_{t+1}$, $x'$ is the last boundary vertex of $\boundary R_{t+1}$ on the $s_j$-$y_j$ path. {\bf (b)} If $z=x$ and $y_j$ is in $R_t^{\out}\cap R_{t+1}$ then $x'=y_j$.
    (Not depicted: if $y_j\in R_t$ then $x'=x$.)  We test whether $v$ is on the shortest $x$-$x'$ path.
    If $z\neq x$ then $z'$ is well defined and the position of $y_j$ is immaterial; 
    we test whether $v$ is on the shortest $z$-$z'$ path (depicted in {\bf (a)}).}
    \label{fig:SitePathIndicator}
\end{figure}

Using the lookup tables in part (E) of the data structure, we find
the first and last vertices ($q$ and $x$) of $\boundary R_t$ on the $s_j$-$y_j$ path.
If $q,x$ do not exist then $v$ is certainly not on the $s_j$-$y_j$ path (Line 4).
Using parts (A,C,E) of the data structure, we invoke $\PointLocate$ to 
find the last point $z$ of $\boundary R_t$ on the shortest path (in $G$) from 
$q$ to $v$.  (See Lemma~\ref{lemma:lastsite}.)
If $z$ is not on the path from $q$ to $x$ in $G$ 
(which corresponds to it not being on the path
from $q$ to $x$ in $T_q^{R_t}$, stored in Part (D)), 
then once again $v$ is certainly not on the $s_j$-$y_j$ path
(Line 8).  So we may assume $z$ lies on the $q$-$x$ path.  
If $z=x$ then there are three cases to consider, depending on whether the destination $y_j$
of the path is in $R_t^{\out}\cap R_{t+1}$, or in $R_{t+1}^{\out}$, or in $R_t$.  
If $y_j\in R_t^{\out}\cap R_{t+1}$ we let $x' = y_j$;
if $y_j\in R_{t+1}^{\out}$ we let $x'$ be the last vertex of $\boundary R_{t+1}$
encountered on the shortest $s_j$-$y_j$ path (part (E));
and if $y_j\in R_t$ we let $x'=x$.  In all cases, $x'$ is the last vertex
of the shortest $s_j$-$y_j$ path that is contained in relevant 
subgraph $R_t^{\out}\cap R_{t+1}$.
(Figure~\ref{fig:SitePathIndicator}(a,b) illustrates the first two possibilities for $x'$.)
Now $v$ is on the $s_j$-$y_j$ path
iff it is on the $x$-$x'$ shortest path, which can be answered using part (A)
of the data structure (Lines 19, 21).  (Figure~\ref{fig:SitePathIndicator}(b) illustrates
one way for $v$ to appear on the $x$-$x'$ path.)
In the remaining case $z$ is on the shortest $q$-$x$
path but is not $x$, meaning the child $z'$ of $z$ on $T_q^{R_t}[x]$ is well defined.
If $\chord{zz'}$ is a chord (corresponding to a path in $R_t^{\out}$)
then $v$ is on the shortest $s_j$-$y_j$ path iff it is on the shortest
$z$-$z'$ path in $R_t^{\out}$, which, once again, can be answered with part (A) of the data structure
(Lines 26, 28).  See Figure~\ref{fig:SitePathIndicator}(a) for an illustration of this case.

\begin{remark}\label{remark:PointLocate}
Strictly speaking we cannot apply Lemma~\ref{lem:PointLocate} (Gawrychowski et al.~\cite{GawrychowskiMWW18}) 
since we do not have an \MSSP{} structure for all of $R_{t}^{\out}$.  Part (A) only handles
distance/LCA queries when the query vertices are in $R_t^{\out}\cap R_{t+1}$.
It is easy to make Gawrychowski et al.'s algorithm work using parts (A) and (E) of 
the data structure. See the discussion at the end of Section~\ref{sect:sidequeries}.
\end{remark}

\begin{algorithm}[H]
    \caption{$\SitePathIndicator(\VDout(u_{i},R_{i+1}),v,f^{*},j)$}
    \label{alg:SiteIndicator}
    \begin{algorithmic}[1]
    \Require The dual representation $\VDout(u_{i},R_{i+1})$ of a Voronoi diagram, a vertex $v\in R_{i+1}^{\out}$, and an $s_{j}$-to-$y_{j}$ site-centroid shortest path ($s_{j},y_{j}$ are with respect to $f^{*}$) in $\VD^{*}$.
    \Ensure \textbf{True} if $v$ is on $s_{j}$-to-$y_{j}$ shortest path, or \textbf{False} otherwise.
    \State $R_{t}\gets$ the ancestor of $R_{i}$ s.t. $v\notin R_{t},v\in R_{t+1}$.
    \State $(q,x) \gets$ first and last $\boundary R_t$ vertices on the shortest $s_j$-$y_j$ path.
     \Comment{Part (E) of the data structure}
    \If {$q,x$ are \textbf{Null}}
        \State \Return \textbf{False}
    \EndIf
    \State $z\gets \PointLocate(\VDout(q, R_{t}),v)$ \Comment{Uses parts (A,C,E) of the data structure}
    \If {$z$ is not on $T_{q}^{R_{t}}[x]$}
        \State \Return \textbf{False}
    \EndIf
    \If {$z=x$}
        \If {$y_j$ is in $R_t^{\out} \cap R_{t+1}$}
            \State $x' \gets y_j$
        \ElsIf{$y_j \not\in R_{t+1}$}
            \State $x' \gets $ last $\boundary R_{t+1}$ vertex on $s_j$-$y_j$ path. \Comment{Part (E)} 
        \Else  
            \State $x' \gets x$ \Comment{I.e., $y_j \not\in R_t^{\out}$}
        \EndIf 
        \If {$v$ is on the shortest $x$-$x'$ path} \Comment{Part (A)}
                \State \Return \textbf{True} \Else \State \Return \textbf{False}
        \EndIf
    \EndIf
    \State $z'\gets$ the child of $z$ on $T_{q}^{R_{t}}[x]$
    \Comment{Part (D)}
    \If {$\chord{zz'}$ is a chord in $\mathcal{C}_q^{R_t}[x]$ and $v$ is on the shortest $z$-$z'$ path in $R_t^{\out}$} \Comment{Part (A)}
        \State \Return \textbf{True}
    \EndIf
    \State \Return \textbf{False}
    \end{algorithmic}
\end{algorithm}

\subsection{The $\ChordIndicator$ Function}\label{sect:ChordIndicator}

The $\ChordIndicator$ function is given $\VDout(u_i,R_{i+1})$, $v\in R_{i+1}^{\out}$, 
a centroid $f^*$, with $\{y_j,s_j\}$ defined as usual, and an index $j\in\{0,1,2\}$.
The goal is to report whether $v$ lies to right of the oriented \emph{site-centroid-site} chord
\[
C^\star = \chord{s_jy_jy_{j-1}s_{j-1}},
\]
which is composed of the shortest $s_j$-$y_j$ and 
$s_{j-1}$-$y_{j-1}$ paths, and the single edge $\{y_j,y_{j-1}\}$.
It is guaranteed that $v$ does not lie on $C^\star$, as this
case is already handled by the $\SitePathIndicator$ function.

Figure~\ref{fig:SiteCentroidSite} illustrates why this point location problem is
so difficult.  Since we know $v\in R_{t+1}$ but not in $R_t$, we can narrow
our attention to $R_t^{\out}\cap R_{t+1}$.  
However the projection of $C^\star$ 
onto $R_t^{\out}$ can touch the boundary 
$\boundary R_t$ an arbitrary number of times. 
Define $\mathcal{C}$ to be the set of oriented chords of 
$R_t^{\out}$ obtained by projecting $C^\star$ onto $R_t^{\out}$.

\begin{figure}
    \centering
    \begin{tabular}{ccc}
\multicolumn{3}{c}{\scalebox{.4}{\includegraphics{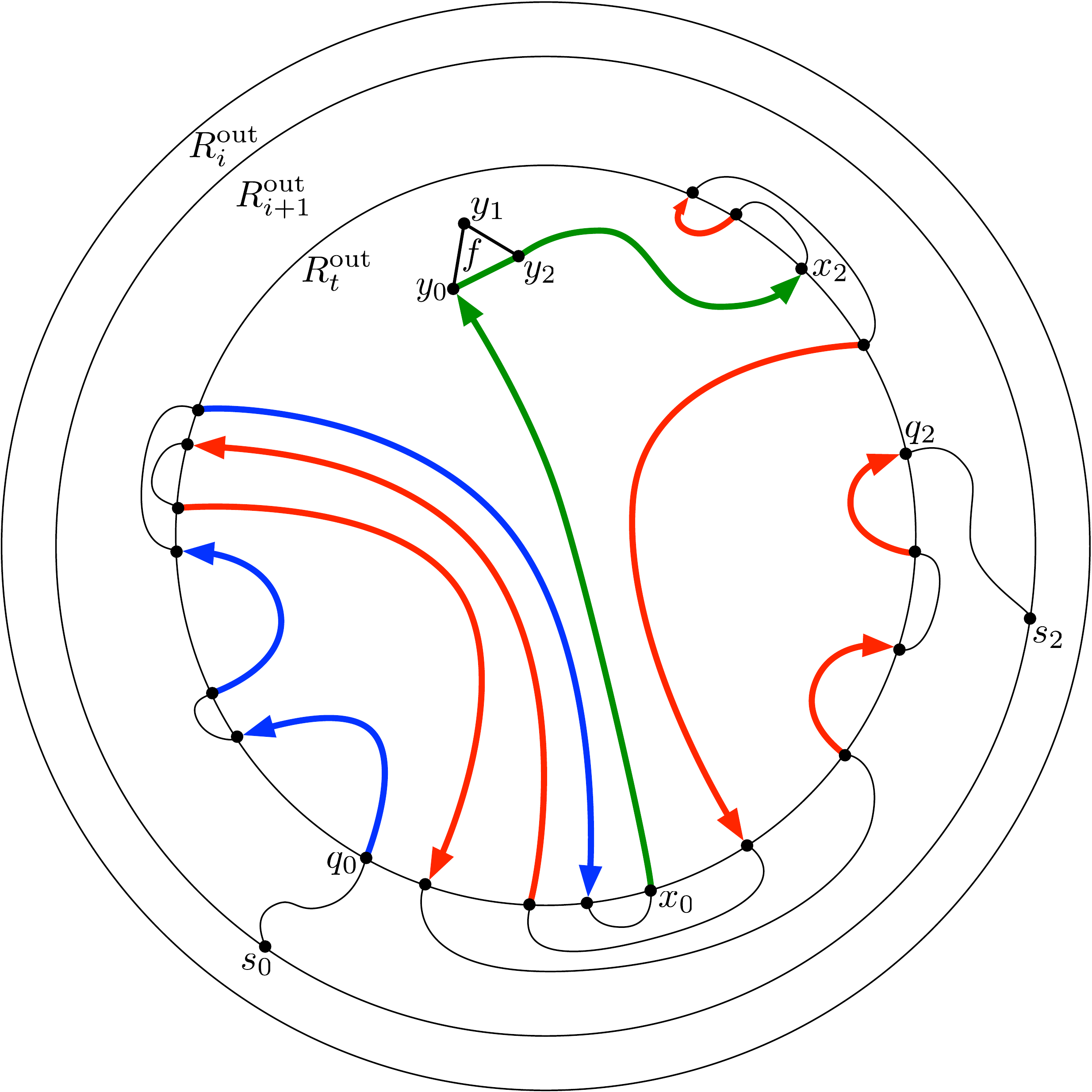}}}\\
&&\\
\multicolumn{3}{c}{{\bf (a)}}\\
&&\\
&&\\
\scalebox{.3}{\includegraphics{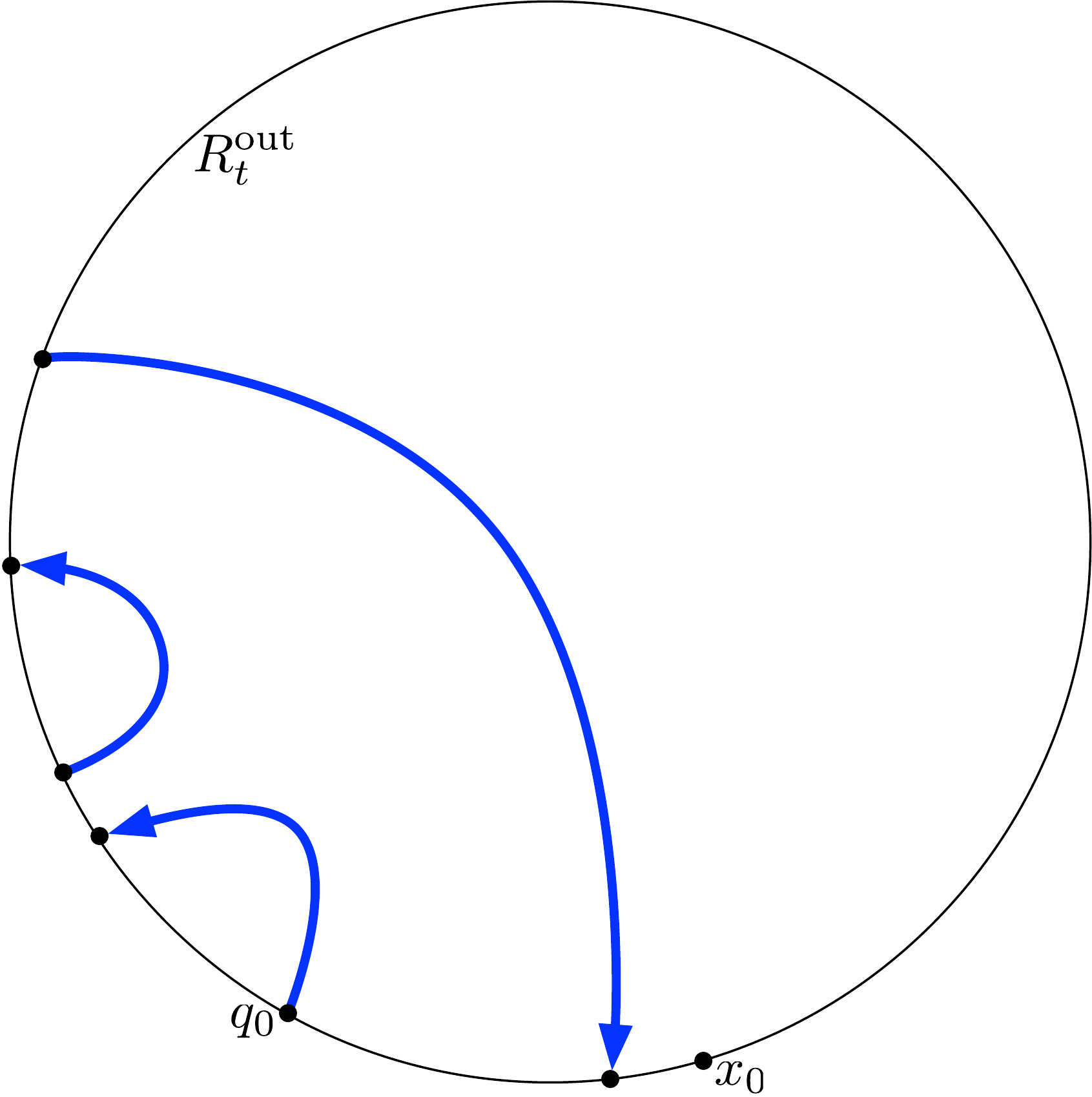}}
&
\scalebox{.3}{\includegraphics{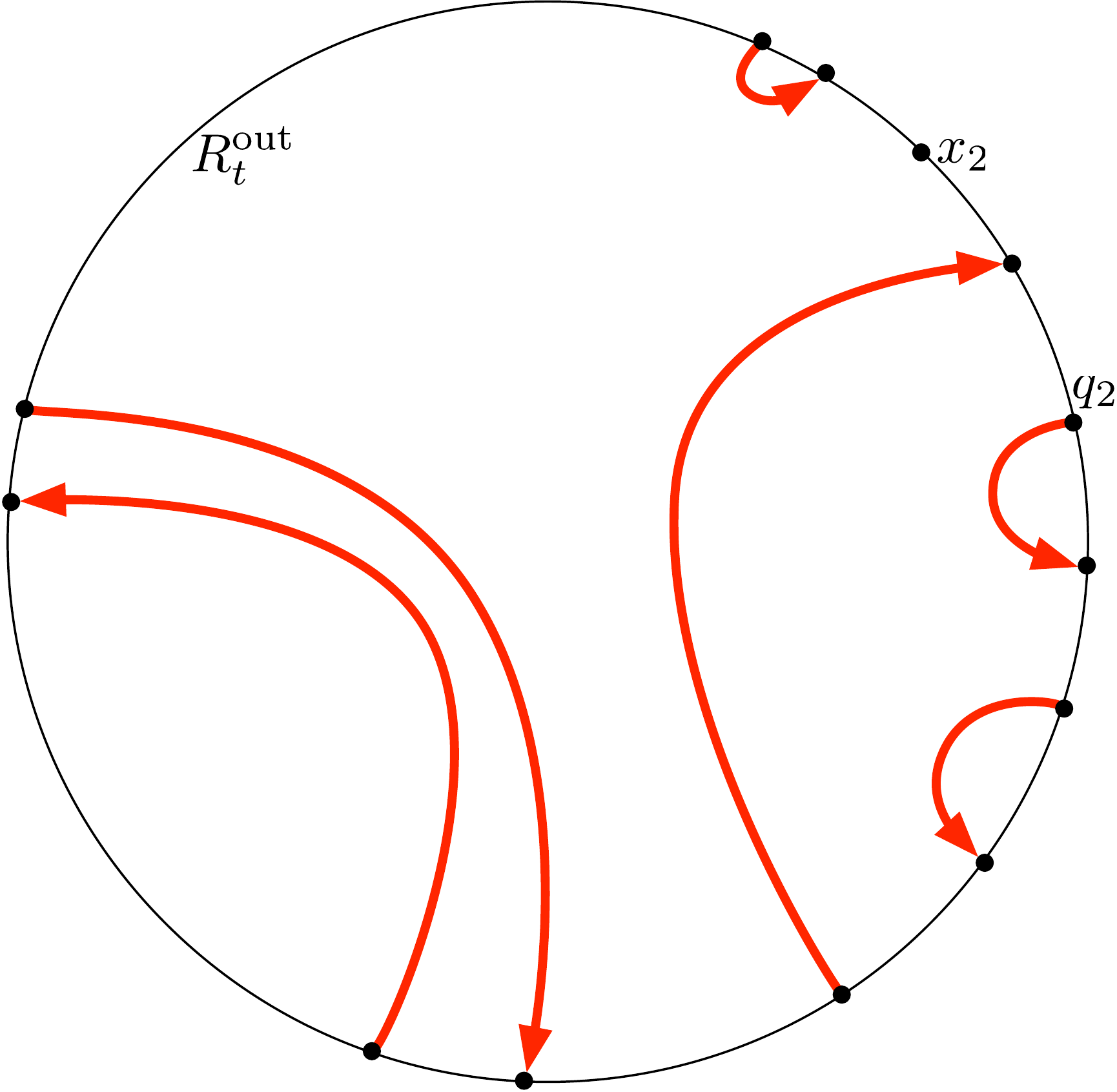}}
&
\scalebox{.3}{\includegraphics{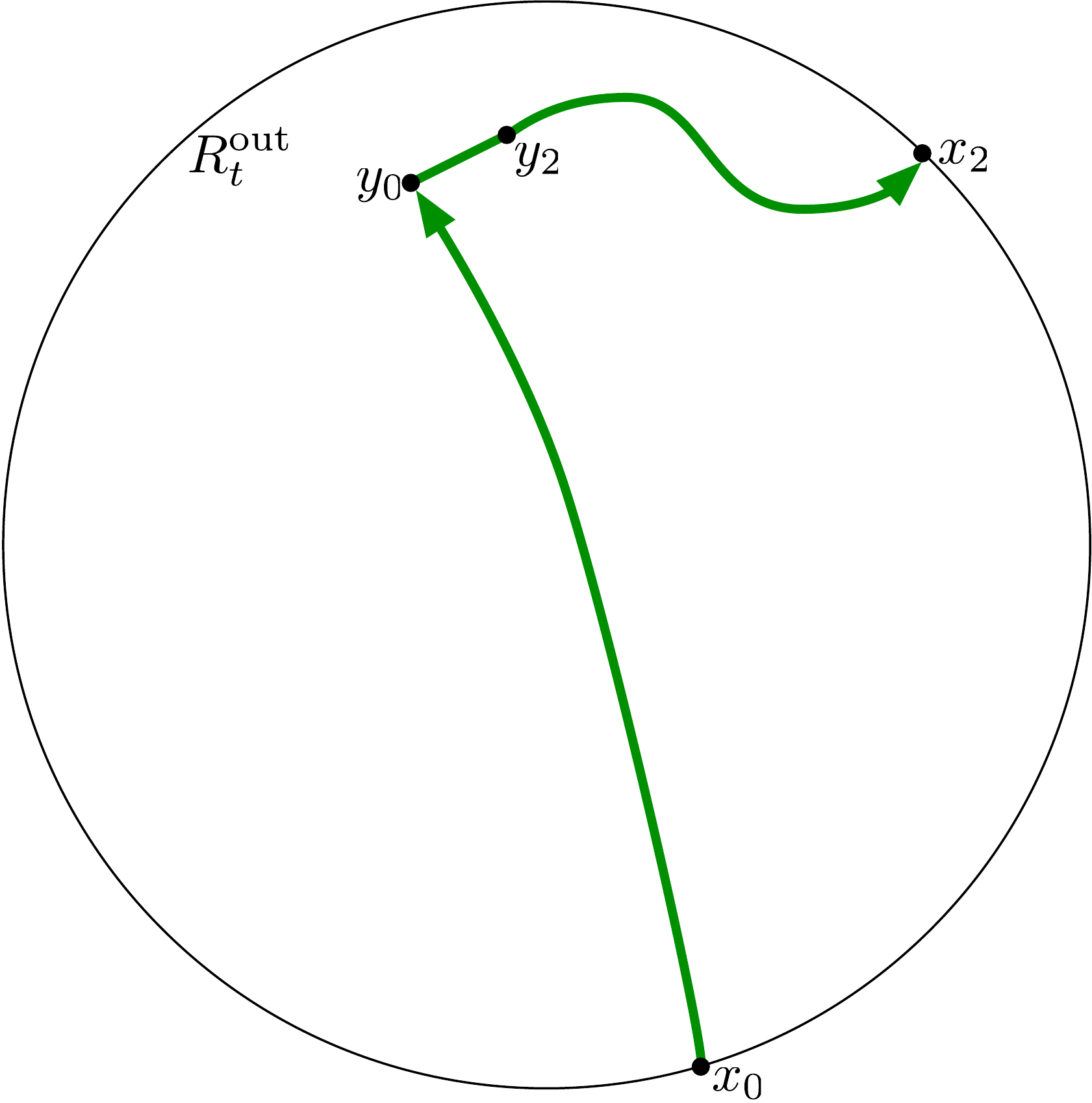}}\\
&&\\
{\bf (b)} & {\bf (c)} & {\bf (d)}
\end{tabular}
    \caption{{\bf (a)} The projection of a site-centroid-site chord 
    $C^\star = \protect\overrightarrow{s_jy_jy_{j-1}s_{j-1}}$ of $R_{i+1}^{\out}$
    onto $R_t^{\out}$ yields a set $\mathcal{C}$ of chords of $R_t^{\out}$,
    partitioned into three classes.  Let $q_j,x_j$ and $q_{j-1},x_{j-1}$
    be the first and last $\boundary R_t$-vertices on 
    the $s_j$-$y_j$ and $s_{j-1}$-$y_{j-1}$ paths.
    {\bf (b)} $\mathcal{C}_1$: all chords in $T_{q_j}^{R_t}[x_j]$.
    {\bf (c)} $\mathcal{C}_2$: all chords in $T_{q_{j-1}}^{R_t}[x_{j-1}]$. Their orientation
    is the reverse of their counterparts in $C^\star$.
    {\bf (d)} $\mathcal{C}_3$: the single chord $\protect\overrightarrow{x_jy_jy_{j-1}x_{j-1}}$.
    }
    \label{fig:SiteCentroidSite}
\end{figure}

Luckily $\mathcal{C}$ has some structure.  
Let $(q_j,x_j)$ and $(q_{j-1},x_{j-1})$
be the first and last $\boundary R_t$ vertices on the 
shortest $s_j$-$y_j$ and $s_{j-1}$-$y_{j-1}$ paths, respectively.
(One or both of these pairs may not exist.)
The chords of $\mathcal{C}$ are in one-to-one correspondence with the chords of
$\mathcal{C}_1 \cup \mathcal{C}_2 \cup \mathcal{C}_3$, defined below, but as we will see,
sometimes with their orientation reversed.
\begin{enumerate}
    \item[$\mathcal{C}_1$:] By definition $\mathcal{C}_1 = C_{q_j}^{R_t}[x_j]$
    contains all the chords on the path from $q_j$ to $x_j$, stored in part (D) of the data structure.  
    Moreover, the orientation of $\mathcal{C}_1$ agrees with the orientation of $C^\star$.
    The blue chords of Figure~\ref{fig:SiteCentroidSite}(a) are isolated as $\mathcal{C}_1$
    in Figure~\ref{fig:SiteCentroidSite}(b).
    \item[$\mathcal{C}_2:$] By definition $\mathcal{C}_2 = C_{q_{j-1}}^{R_t}[x_{j-1}]$
    contains all the chords on the path from $q_{j-1}$ to $x_{j-1}$.
    The red chords of $\mathcal{C}$ in Figure~\ref{fig:SiteCentroidSite}(a) 
    are \emph{represented} by chords $\mathcal{C}_2$, but with reversed orientation.
    Figure~\ref{fig:SiteCentroidSite}(c) depicts $\mathcal{C}_2$.
    \item[$\mathcal{C}_3:$] This is the singleton set containing the oriented
    chord $\chord{x_jx_{j-1}}$ consisting of the shortest $x_j$-$y_j$ and $x_{j-1}$-$y_{j-1}$ 
    paths and the edge $\{y_j,y_{j-1}\}$.
\end{enumerate}

The chord-set $\mathcal{C}$ partitions $R_t^{\out}$ into a piece-set $\mathcal{P}$,
with one such piece $P\in \mathcal{P}$ containing $v$.  (Remember that $v$ is not on $C^\star$.)
We can also consider the piece-sets 
$\mathcal{P}_1,\mathcal{P}_2,\mathcal{P}_3$
generated by 
$\mathcal{C}_1,\mathcal{C}_2,\mathcal{C}_3$.  
Let $P_1\in\mathcal{P}_1, P_2\in \mathcal{P}_2, P_3\in\mathcal{P}_3$ be the pieces
containing $v$.  Since, ignoring orientation, $\mathcal{C}=\mathcal{C}_1\cup\mathcal{C}_2\cup\mathcal{C}_3$,
it must be that $P=P_1\cap P_2 \cap P_3$.   In order to determine whether $v$ is to the right
of $C^\star$, it suffices to find some chord $C\in\mathcal{C}$ bounding $P$ and ask whether $v$
is to the right of $C$.  Thus, $C$ must also be on the boundary of one of $P_1,P_2,$ or $P_3$.

The high-level strategy of $\ChordIndicator$ is as follows.
First, we will find some piece $P_1' \in \mathcal{P}_{q_j}^{R_t}$
that is contained in $P_1$ using the procedure $\PieceSearch$ described below, in Section~\ref{sect:PieceSearchprocedure}.
The chords of $\mathcal{C}_1$ bounding $P_1$ are precisely 
the \emph{maximal} chords in $\mathcal{C}_1$ 
from vantage $P_1'$.
Using $\MaximalChord$ (part (D)) 
we will find a candidate chord $C_1\in\mathcal{C}_1$, 
and one edge $e$ on the boundary cycle of $\boundary R_t$ occluded by $C_1$ from vantage $P_1'$.
Turning to $\mathcal{C}_2$, we use $\AdjacentPiece$ to find the piece 
$P_e \in \mathcal{P}_{q_{j-1}}^{R_t}$ adjacent to $e$.
Then, using $\PieceSearch$ and $\MaximalChord$ again,
we find a $P_2'\in\mathcal{P}_{q_{j-1}}^{R_t}$ contained in $P_2$ and the maximal
chord $C_2$ occluding $P_e$ from vantage $P_2'$.  Let $C_3$ be the singleton chord in $\mathcal{C}_3$.
We determine the ``best'' chord $C_\ell \in \{C_1,C_2,C_3\}$, decide whether $v$ lies to 
the right of $C_{\ell}$, and return this answer if $\ell\in\{1,3\}$ or reverse it if $\ell=2$.
Recall that chords in $\mathcal{C}_2$ have the opposite orientation as their counterparts in $\mathcal{C}$.

$\PieceSearch$ is presented in Section~\ref{sect:PieceSearchprocedure}
and $\ChordIndicator$ in Section~\ref{sect:ChordIndictor-procedure}.

\subsubsection{$\PieceSearch$}\label{sect:PieceSearchprocedure}

We are given a region $R_t$, a vertex $v\in R_t^{\out}\cap R_{t+1}$,
and two vertices $q,x\in \boundary R_t$.  We 
must locate \emph{any} piece $P' \in\mathcal{P}_q^{R_t}$ that is contained in the 
unique piece $P \in \mathcal{P}_q^{R_t}[x]$ containing $v$.  
The first thing we do is find the \emph{last}
$\boundary R_t$ vertex $z$ on the shortest path from $q$ to $v$,
which can be found with a call to $\PointLocate$ on $\VDout(q,R_t)$.
(This uses parts (A,C,E) of the data structure.)
The shortest path from $z$ to $v$ cannot cross any chord
in $\mathcal{C}_q^{R_t}[x]$ (since they are part of a shortest path),
but it can coincide with a prefix of some chord in $\mathcal{C}_q^{R_t}[x]$.
Thus, if no chord of $\mathcal{C}_q^{R_t}[x]$ is incident to $z$,
then we are free to return \emph{any} piece containing $z$.
(There may be multiple options if $z$ is an endpoint of a chord in $\mathcal{C}_q^{R_t}$.  This case is depicted in Figure~\ref{fig:PieceSearch}.
When $z = z_0$, we know that $v\in P_5\cup \cdots \cup P_9$ and return
any piece containing $z$.)
In general $z$ may be incident to up to
two chords $C_1,C_2 \in \mathcal{C}_q^{R_t}[x]$.
(This occurs when the shortest $q$-$x$ path
touches $\boundary R_t$ at $z$ without leaving
$R_t^{\out}$.)  In this case we determine
which side of $C_1$ and $C_2$ $v$ is on 
(using parts (A) and (E) of the data structure; see Lemma~\ref{lem:chordside} 
in Section~\ref{sect:sidequeries} for details) and return the appropriate 
piece adjacent to $C_1$ or $C_2$.
This case is depicted in Figure~\ref{fig:PieceSearch}
with $z=z_1$; the three possible answers coincide with
$v\in\{v_1,v_2,v_3\}$.

\begin{algorithm}
    \caption{$\PieceSearch(R_{t},q,x,v)$}
    \label{alg:PieceSearch}
    \begin{algorithmic}[1]
    \Require A region $R_{t}$, two vertices $q,x \in\boundary R_{t}$, 
             and a vertex $v$ not on the $q$-to-$x$ shortest path in $G$.
    \Ensure A piece $P'\in\mathcal{P}_{q}^{R_{t}}$, 
            which is a subpiece of the unique piece 
            $P\in\mathcal{P}_{q}^{R_{t}}[x]$ containing $v$.
    \State $z \gets \PointLocate(\VDout(q, R_{t}),v)$ 
                        \Comment{Uses parts (A,C,E) of the data structure}
    \If {$z$ is not the endpoint of any chord in $\mathcal{C}_{q}^{R_{t}}[x]$}
        \State \Return any piece in $\mathcal{P}_{q}^{R_{t}}$ containing $z$.
    \EndIf
    \State $C_{1},C_{2}\gets$ two chords in $\mathcal{C}_{q}^{R_{t}}[x]$ adjacent to $z$ ($C_2$ may be {\bf Null})
    \State Determine whether $v$ is to the left or right of $C_1$ and $C_2$. 
                        \Comment{Part (A); see Lemma~\ref{lem:chordside}}
    \State \Return a piece adjacent to $C_1$ or $C_2$ that respects the queries of Line 6.
    \end{algorithmic}
\end{algorithm}

\begin{figure}
    \centering
    \scalebox{.45}{\includegraphics{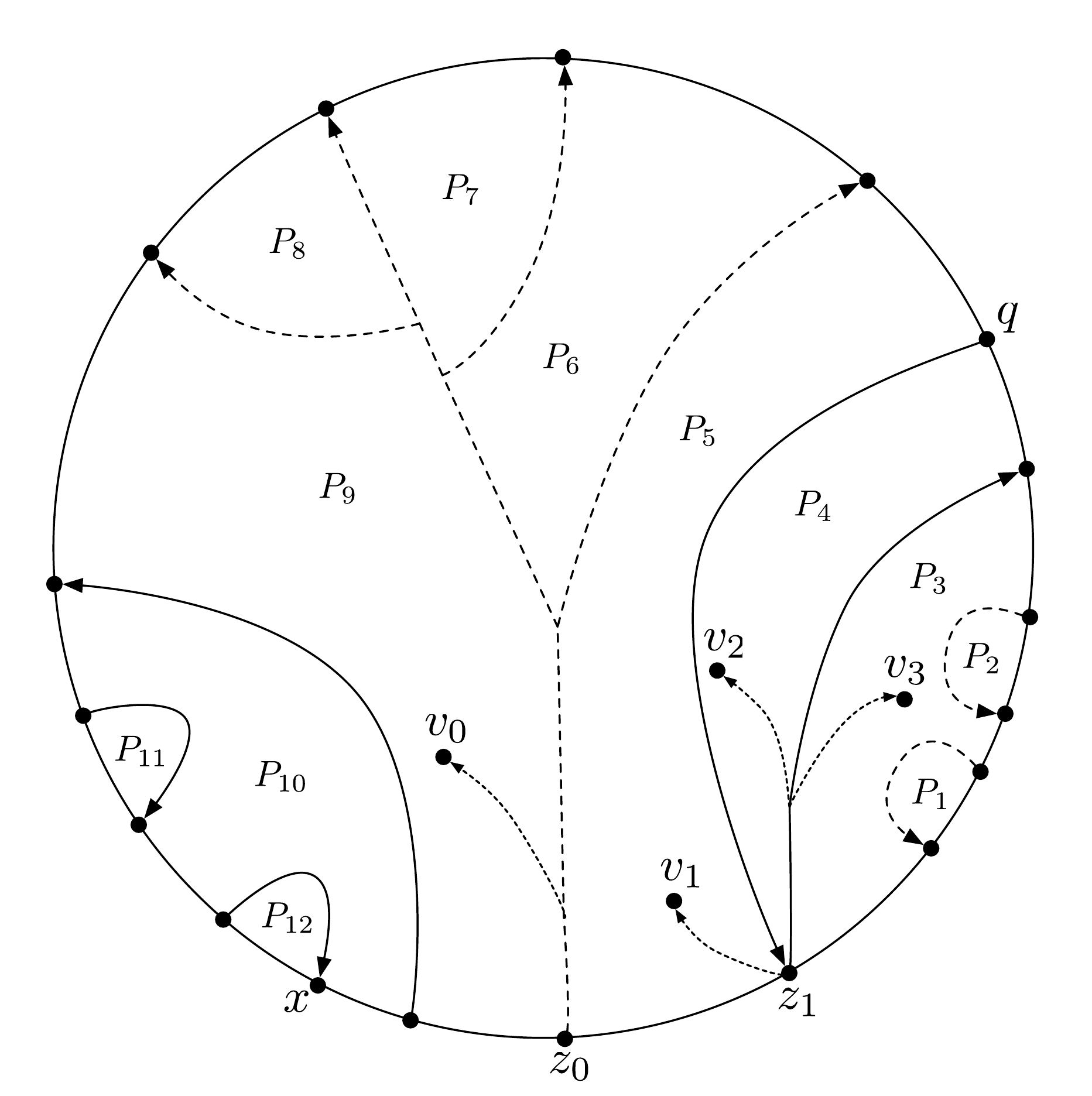}}
    \caption{Solid chords are in $\mathcal{C}_q^{R_t}[x]$.
    Dashed chords are in $\mathcal{C}_q^{R_t}$ but not $\mathcal{C}_q^{R_t}[x]$.
    When $z = z_0, v = v_0$, the piece in $\mathcal{P}_q^{R_t}[x]$ 
    containing $v$ is the union of $P_5$--$P_9$.  
    $\PieceSearch$ reports any piece
    containing $z_0$.  When $z=z_1, v\in \{v_1,v_2,v_3\}$, $z$ is incident
    to two chords $C_1,C_2$.  $\PieceSearch$ decides which side of $C_1,C_2$ 
    $v$ is on (see Lemma~\ref{lem:chordside}), 
    and returns the appropriate piece adjacent to $C_1$ or $C_2$.}
    \label{fig:PieceSearch}
\end{figure}

\subsubsection{$\ChordIndicator$}\label{sect:ChordIndictor-procedure}

Let us walk through the $\ChordIndicator$ function.  
If $C^\star = \chord{s_jy_jy_{j-1}s_{j-1}}$
does not touch the interior of $R_t^{\out}$ then the left-right relationship
between $C^\star$ and $v\not\in R_t$ is known, and stored in part (E) of the data structure.
If this is the case the answer is returned immediately, at Line 3.
A relatively simple case is when $\mathcal{C}_1$ and $\mathcal{C}_2$ are empty, 
and $\mathcal{C}=\mathcal{C}_3$ consists of just one chord 
$C_3 = \chord{x_jy_jy_{j-1}x_{j-1}}$.  
We determine whether $v$ is to the right or left
of $C_3$ and return this answer (Line 8).  
(Lemma~\ref{lem:chordside} in Section~\ref{sect:sidequeries} explains how to test whether
$v$ is to one side of a chord.)
Thus, without loss of generality we can assume $\mathcal{C}_1\neq \emptyset$ and $\mathcal{C}_2$ may or may not be empty.

Recall that $P_1$ is $v$'s piece in $\mathcal{P}_{q_j}^{R_t}[x_j]$.  
Using $\PieceSearch$ we find a piece $P_1'\subseteq P_1$ in the more 
refined partition $\mathcal{P}_{q_j}^{R_t}$ and 
find a $\MaximalChord$ $C_1\in\mathcal{C}_1$
from vantage $P_1'$, and hence from vantage $v$
as well.
We regard $\boundary R_t$ as circularly ordered 
according to a clockwise walk around 
the hole on $\boundary R_t$ in $R_t^{\out}$.  
The chord $C_1$ occludes an interval $I_1$ of $\boundary R_t$ from vantage $v$.
If $C_1$ is \emph{not} one of the chords bounding $P$, then $C_3$ 
or some $C_2\in\mathcal{C}_2$ must occlude a superset of $I_2$, so we will
attempt to find such a $C_2$, as follows.

Let $e$ be the first edge on the boundary cycle occluded by $C_1$, 
i.e., $e$ joins the first two elements of $I_1$.
Using $\AdjacentPiece$ we find the unique 
piece $P_e\in \mathcal{P}_{q_{j-1}}^{R_t}$
with $e$ on its boundary.  Using $\PieceSearch$ again we
find $P_2'\in\mathcal{P}_{q_{j-1}}^{R_t}$ contained in $P_2$,
and using $\MaximalChord$ again, 
we find the maximal chord $C_2 \in \mathcal{C}_2$ that occludes
$P_e$ from vantage $P_2'$, and hence from vantage $v$ as well.  
Observe that since all chords
in $\mathcal{C}_2$  are vertex-disjoint from $C_1$, 
if $C_2\neq $ \textbf{Null} then $C_2$ must occlude
a strictly larger interval $I_2 \supset I_1$ of $\boundary R_t$.
(If $C_2$ is \textbf{Null} then $I_2=\emptyset$.)
It may be that $C_1$ and $C_2$ are both not on the boundary of $P$,
but the only way that could occur is if $C_3 \in \mathcal{C}_3$
occludes a superset of $I_1$ and $I_2$ on the boundary $\boundary R_t$.
We check whether $v$ lies to the right or left of $C_3$ and
let $I_3$ be the interval of $\boundary R_t$ occluded by $C_3$ from vantage $v$.
If $I_3$ does not cover $e$, then we cannot conclude that $C_3$ is superior
than $C_1/C_2$.  Thus, we find the chord $C_{\ell} \in \{C_1,C_2,C_3\}$
that covers $e$ and maximizes $|I_\ell|$.  $C_\ell$ must be on the boundary of $P$,
so the left-right relationship between $v$ and $C^\star$ is exactly the
same as the left-right relationship between $v$ and $C_\ell$, if $\ell\in\{1,3\}$,
and the reverse of this relationshp if $\ell=2$ since chords in $\mathcal{C}_2$
have the opposite orientation as their subpath counterparts in $C^\star$.
Figure~\ref{fig:ChordIndicator} illustrates how $\ell$ could take on all three
values.

\begin{figure}
    \begin{tabular}{ccc}
    \hspace{-.3cm}\scalebox{.26}{\includegraphics{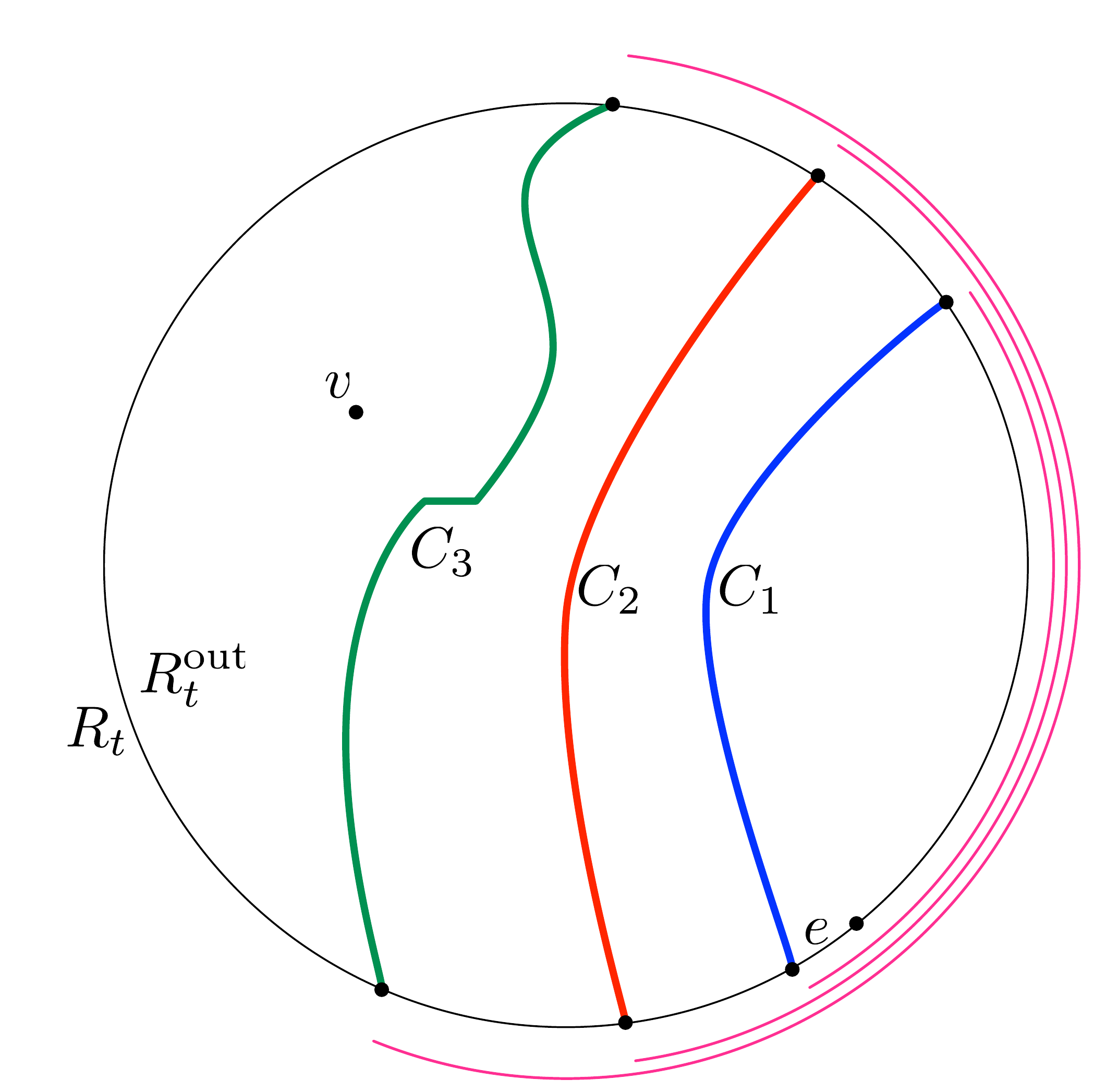}}
    &\hspace{-.35cm}\scalebox{.26}{\includegraphics{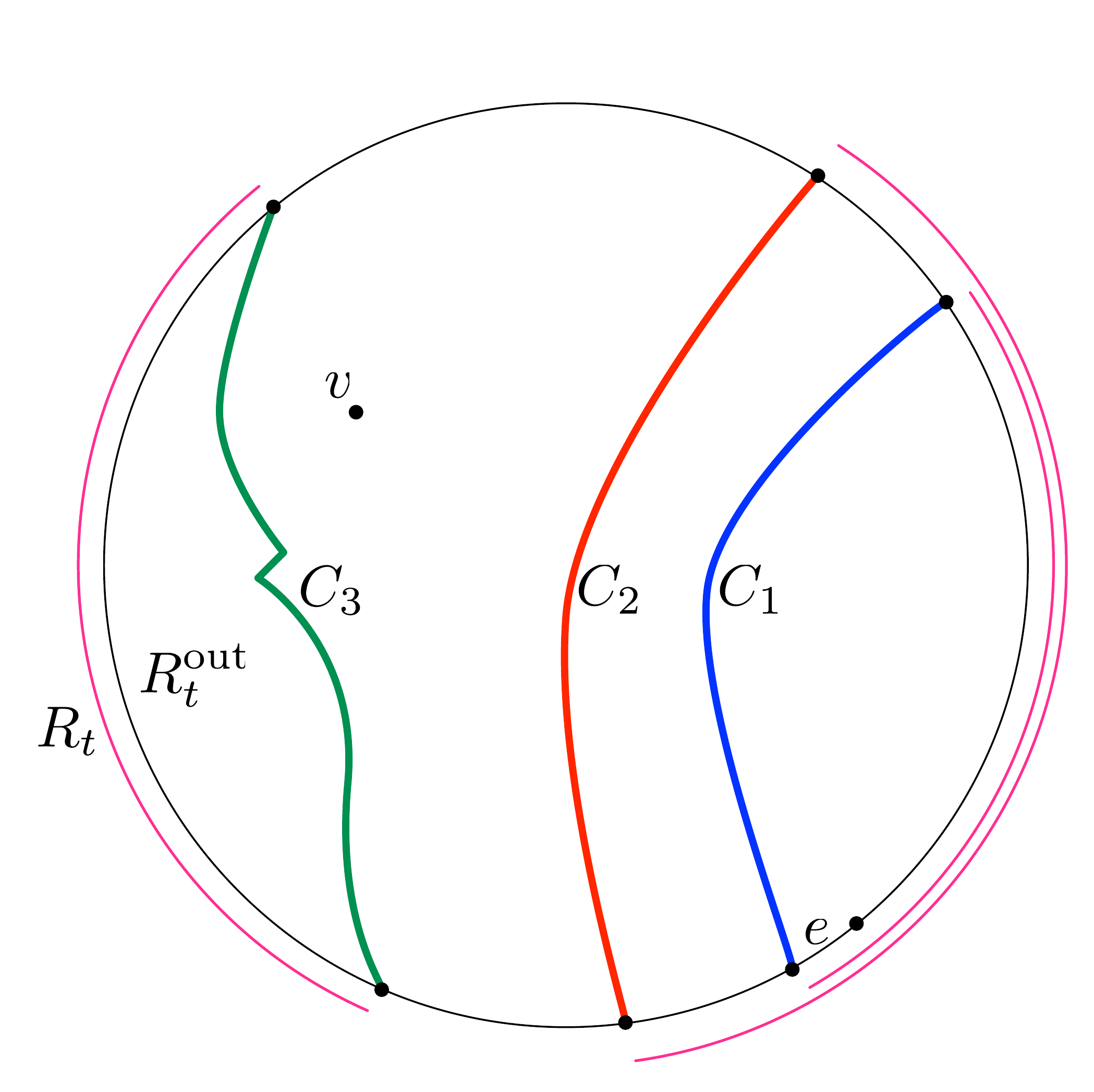}}
    &\hspace{-.4cm}\scalebox{.26}{\includegraphics{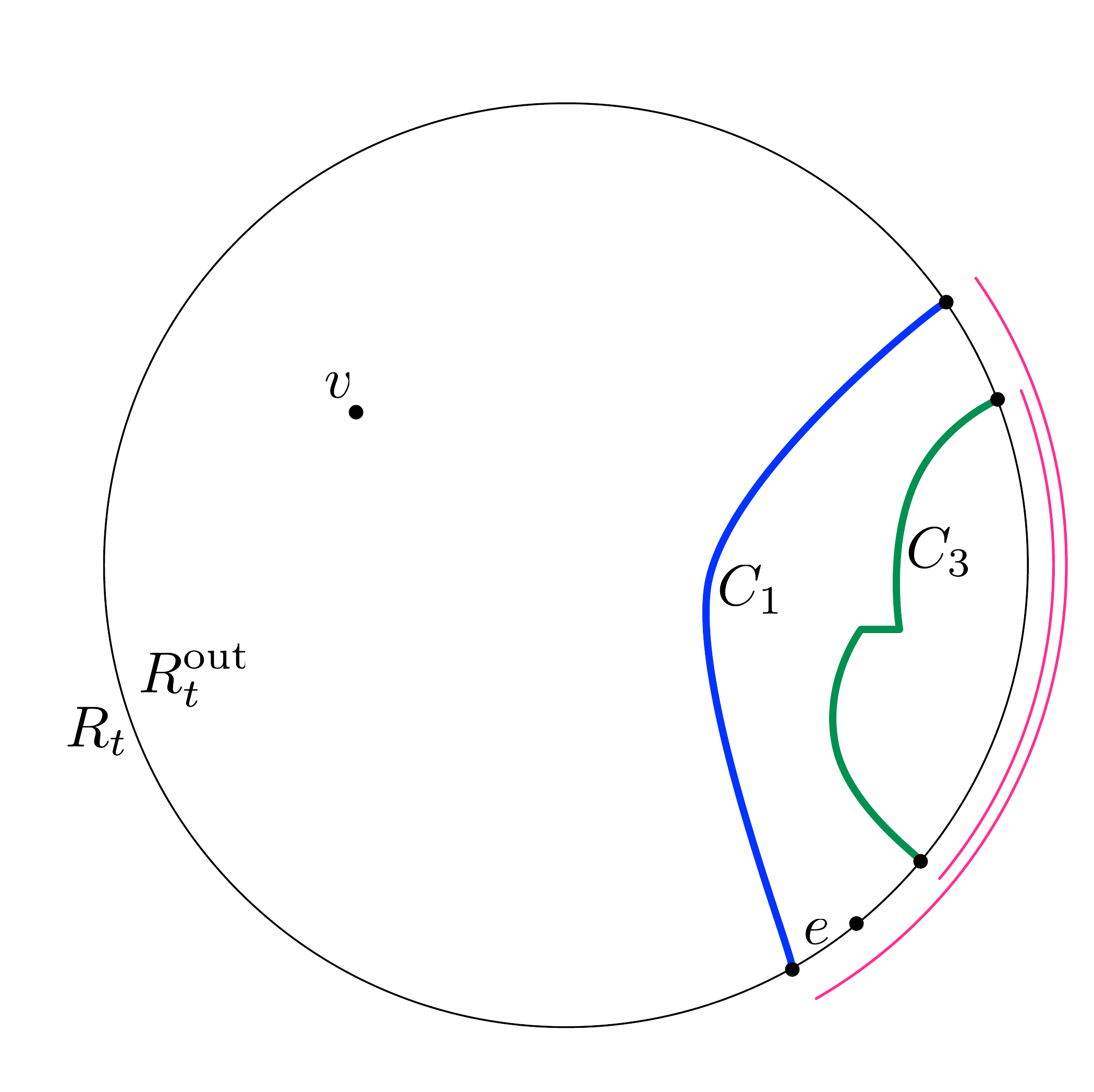}}\\
    (a) & (b) & (c)
    \end{tabular}
    \caption{The intervals $I_1,I_2,I_3$ are represented as pink circular arcs.
In (a) $C_2$ exists and $C_3$ is better than 
$C_1,C_2$ since $I_3 \supset I_2 \supset I_1$.  In 
(b) $C_2$ exists, but $C_3$ occludes an interval $I_3$ that does not contain $e$,
so $C_2$ is the best chord.  
In (c) $C_2$ is \textbf{Null}, and $C_3$ does not occlude $e$ 
from $v$, so $C_1$ is the only eligible chord. 
(In the figure $I_3 \subset I_1$ but it could also be 
as in (b), with $I_3$ disjoint from $I_1$.)}
    \label{fig:ChordIndicator}
\end{figure}

\begin{algorithm}
    \caption{$\ChordIndicator(\VDout(u_{i},R_{i+1}),v,f^{*},j)$}
    \label{alg:CentroidIndicator}
    \begin{algorithmic}[1]
    \Require The dual representation $\VDout = \VDout(u_{i},R_{i+1})$ of a Voronoi diagram,
    a centroid $f^*$ in $\VDout$ with face $f$ on vertices $y_0,y_1,y_2$, which are in the 
    Voronoi cells of $s_0,s_1,s_2$, 
    an index $j\in\{0,1,2\}$, and a vertex $v \in R_{i+1}^{\out}$ that does not lie on the
    site-centroid-site chord $C^\star = \overrightarrow{s_{j}y_{j}y_{j-1}s_{j-1}}$.
    \Ensure \textbf{True} if $v$ lies to the right of $C^\star$, and $\textbf{False}$ otherwise.
    \State $R_{t}\gets$ the ancestor of $R_{i}$ s.t. $v\notin R_{t}, v\in R_{t+1}$.  $\mathcal{C}$ is the projection of $C^\star$ onto $R_t^{\out}$.
    \If {the left/right relationship between $R_t^{\out}$ and $C^\star = \chord{s_{j}y_{j}y_{j-1}s_{j-1}}$ is known} 
    \State \Return stored \textbf{True}/\textbf{False} answer.  \Comment{Part (E)}
    \EndIf  \Comment{(It follows that $C^\star$ crosses $\boundary R_t$ and that $\mathcal{C}\neq\emptyset$)}
    \State $(q_{j},x_{j}) \gets $ first and last $\boundary R_t$-vertices on shortest $s_{j}$-$y_{j}$ path. \Comment{Part (E)}
    \State $(q_{j-1},x_{j-1}) \gets$ first and last $\boundary R_t$-vertices on shortest $s_{j-1}$-$y_{j-1}$ path. \Comment{Part (E)}
    \If{$\mathcal{C}_1 = \mathcal{C}_2 = \emptyset$}
        \State \Return \textbf{True} if $v$ is to the right of the 
        $\mathcal{C}_3$-chord $\chord{x_jy_jy_{j-1}x_{j-1}}$, or \textbf{False} otherwise.
    \EndIf \Comment{\emph{W.l.o.g., continue under the assumption that $\mathcal{C}_1 \neq \emptyset$.}}
    \State $P_1' \gets \PieceSearch(R_{t},q_{j},x_{j},v)$          \Comment{Uses parts (A,C)}
    \State $C_1 \gets \MaximalChord(R_t,q_j,x_j,P_1',\perp)$        \Comment{Part (D)}
    \State $I_1 \gets $ the clockwise interval of hole $\boundary R_t$ occluded by $C_1$ from vantage $v$.
    \State $e \gets $ edge joining first two elements of $I_1$.
    \State $P_e \gets \AdjacentPiece(R_t,q_{j-1},e)$        \Comment{Part (D)}
    \State $P_2' \gets \PieceSearch(R_t,q_{j-1},x_{j-1},v)$         \Comment{Uses parts (A,C)}
    \State $C_2 \gets \MaximalChord(R_t,q_{j-1},x_{j-1},P_2',P_e)$  \Comment{Part (D); may return {\bf Null}}
    \State $I_2 \gets $ the clockwise interval of hole $\boundary R_t$ occluded by $C_2$ from vantage $v$. \Comment{$\emptyset$ if $C_2=$ {\bf Null}}
    \State $C_3 \gets $ single chord in $\mathcal{C}_3$, if any.     \Comment{May be {\bf Null}}
    \State $I_3 \gets $ the clockwise interval of hole $\boundary R_t$ occluded by $C_3$ from vantage $v$. \Comment{$\emptyset$ if $C_3=$ {\bf Null}}
    \State $\ell \gets $ index such that $I_\ell$ covers $e$, and $|I_\ell|$ is maximum.
    \If{$v$ is to the right of $C_\ell$ and $\ell\in\{1,3\}$ or $v$ is to the left of $C_\ell$ and $\ell=2$}
         \State \Return \textbf{True} \EndIf
    \State \Return \textbf{False}
    \end{algorithmic}
\end{algorithm}

\subsubsection{Side Queries}\label{sect:sidequeries}

Lemma~\ref{lem:chordside}
explains how we test whether $v$ is to the 
right or left of a chord, which is used in both $\PieceSearch$ and $\ChordIndicator$.

\begin{lemma}\label{lem:chordside}
For any $C\in\mathcal{C}_1\cup \mathcal{C}_2\cup \mathcal{C}_3$ and $v$ not on $C$,
we can test whether $v$ lies to the right or left
of $C$ in $O(\kappa\log\log n)$ time, 
using parts (A) and (E) of the data structure.
\end{lemma}

\begin{proof}
There are several cases.
\paragraph{Case 1.}
Suppose that 
$C = \chord{c_0c_1} \in \mathcal{C}_1\cup\mathcal{C}_2$ 
corresponds to the shortest
path from $c_0$ to $c_1$ in $R_t^{\out}$, 
$c_0,c_1\in \boundary R_t$.
Let $c_0',c_1'$ be pendant vertices attached 
to $c_0,c_1$ embedded inside the face of $R_t^{\out}$ 
bounded by $\boundary R_t$.
The shortest $c_0'$-$v$ paths and $c_0'$-$c_1'$ paths branch at some point.
We ask the \MSSP{} structure (part (A)) for the least common ancestor, $w$,
of $v$ and $c_1'$ in the shortcutted SSSP tree rooted at $c_0'$.
This query also returns the two tree 
edges $e_{v},e_{c_1'}$ leading to $v$ and $c_1'$, respectively.
Let $e_w$ be the edge connecting $w$ to its parent.\footnote{The purpose of adding $c_0',c_1'$ 
is to make sure all three edges $e_w,e_v,e_{c_1'}$ exist.  The vertices $c_0',c_1'$ 
are not represented in the \MSSP{} structure.  The edges $(c_0',c_0)$ and $(c_1,c_1')$
can be simulated by inserting them between the two boundary edges on $\boundary R_t$
adjacent to $c_0$ and $c_1$, respectively.}
If the clockwise order around $w$ is $e_w,e_{c_1'},e_v$ 
then $v$ lies to the right of $\chord{c_0c_1}$; 
otherwise it lies to the left.  Note that if the shortest
$c_0'$-$c_1'$ and $c_0'$-$v$ paths in $G$ 
branch at a point in $R_{t+1}^{\out}$,
then $w$ will be the nearest ancestor of the branchpoint on $\boundary R_{t+1}$
and one or both of $e_v,e_{c_1'}$ may be ``shortcut'' edges 
in the \MSSP{} structure.  See Figure~\ref{fig:leftright}(a) for a depiction
of this case.

\begin{figure}[h]
    \centering
    \begin{tabular}{cc}
    \multicolumn{2}{c}{\scalebox{.33}{\includegraphics{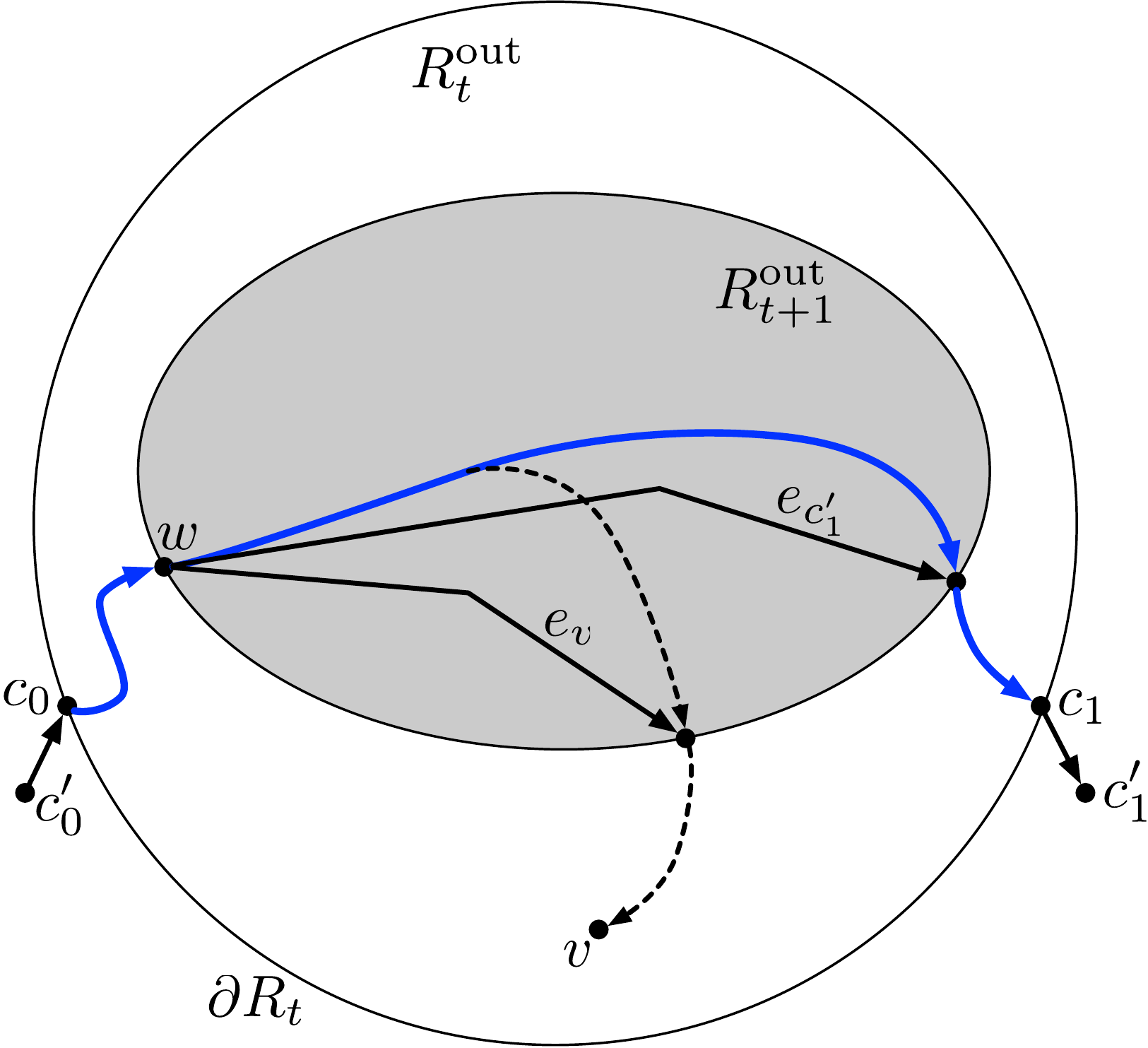}}}\\
    \multicolumn{2}{c}{ }\\
    \multicolumn{2}{c}{(a)}\\
    \scalebox{.32}{\includegraphics{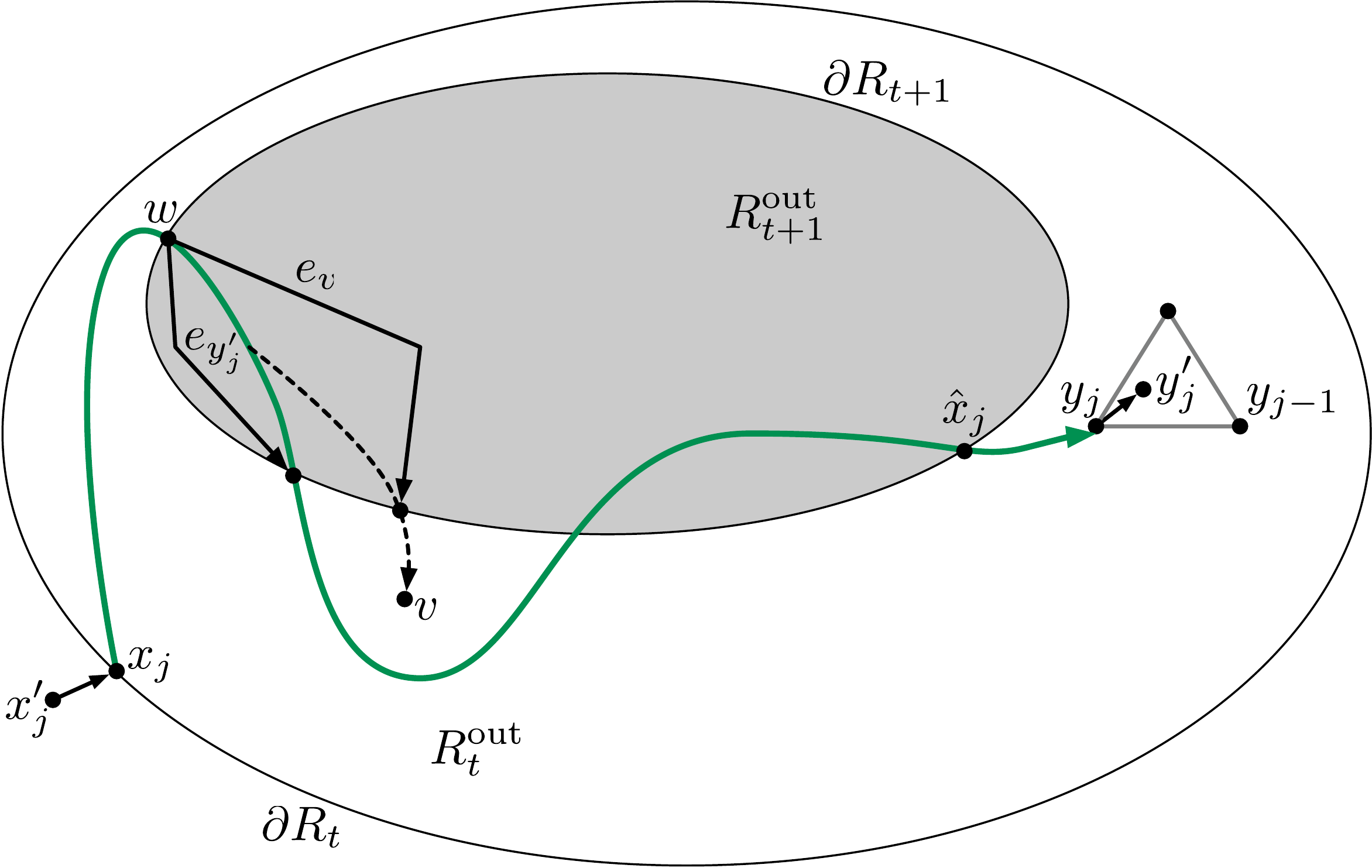}}
    &\scalebox{.32}{\includegraphics{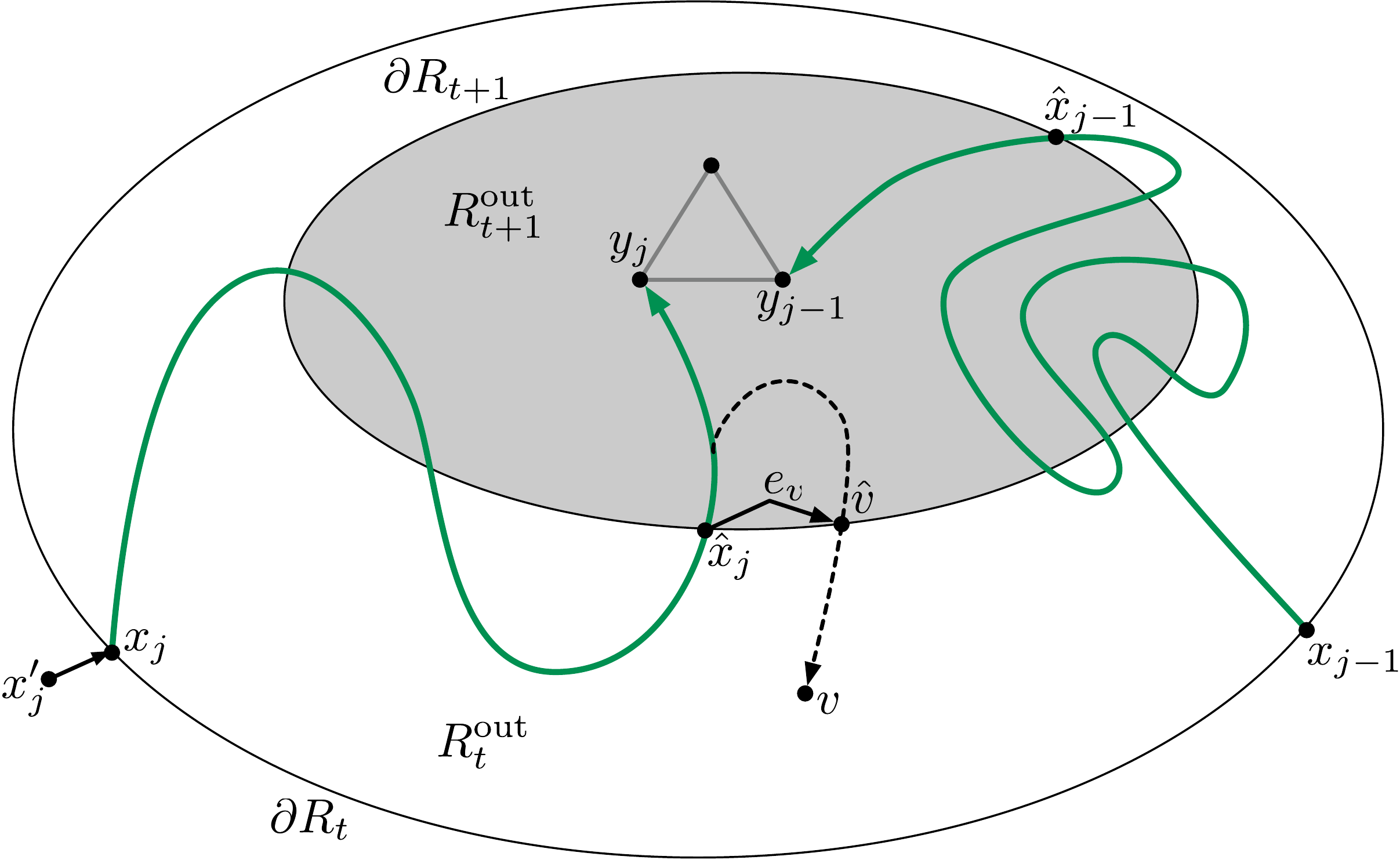}}\\
    &\\
    (b) & (c)
    \end{tabular}
    \caption{(a) The chord $C\in \mathcal{C}_1\cup\mathcal{C}_2$ corresponds to a shortest path,
    which may pass through $R_{t+1}^{\out}$, in which case it is represented in the \MSSP{}
    structure with shortcut edges (solid, angular edges).  
    (b) The chord $C = \protect\overrightarrow{x_jy_jy_{j-1}x_{j-1}}$ is in $\mathcal{C}_3$, and
    $f$ lies in $R_{t}^{\out}\cap R_{t+1}$.   This is handled similarly to (a).
    (c) Here $f$ lies in $R_{t+1}^{\out}$, $\hat{x}_j,\hat{x}_{j-1}$ are the 
    last $\boundary R_{t+1}$ vertices on
    the $s_j$-$y_j$ and $s_{j-1}$-$y_{j-1}$ paths.  
    If the shortest $x_j'$-$\hat{x}_j$ 
    and $x_j'$-$v$ paths branch, we can answer the query as in (b).  
    If $x_j'$-$\hat{x}_j$ is a prefix of $x_j'$-$v$, $e_v = (\hat{x}_j,\hat{v})$,
    and $\hat{v}\in \boundary R_{t+1}$, then we can use the clockwise
    order of $\hat{x}_j,\hat{v},\hat{x}_{j-1}$ around the hole on $\boundary R_{t+1}$
    to determine whether $v$ lies to the right of $C$.
    (Not depicted: the case when $\hat{v}\not\in \boundary R_{t+1}$.)
    }
    \label{fig:leftright}
\end{figure}
\paragraph{Case 2.}
Now suppose $C = \chord{x_jy_jy_{j-1}x_{j-1}}$ 
is the one chord in $\mathcal{C}_3$.
Consider the following distance function $\hat{d}$ for 
vertices in $z\in R_t^{\out}$:
\[
\hat{d}(z) = \min\Big\{
\dist_{G}(u_i,x_j) + \dist_{G}(x_j,z),\; \;
                        \dist_{G}(u_i,x_{j-1}) + \dist_{G}(x_{j-1},z)\Big\}.
\]
Observe that the terms involving $u_i$ are stored in part (E) 
and, if $z\in R_t^{\out}\cap R_{t+1}$, the other terms can be queried 
in $O(\kappa\log\log n)$ time using part (A).
It follows that the shortest path forest w.r.t.~$\hat{d}$ 
has two trees, rooted at $x_j$ and $x_{j-1}$.
Using part (A) of the data structure
we compute $\hat{d}(v)$, 
which reveals the $j^\star \in \{j,j-1\}$ 
such that $v$ is in $x_{j^\star}$'s tree.
At this point we break into two cases, depending
on whether $f$ is in $R_{t}^{\out}\cap R_{t+1}$, or in $R_{t+1}^{\out}$.
We assume $j^\star = j$ without loss of generality and depict only
this case in Figure~\ref{fig:leftright}(b,c).

\paragraph{Case 2a.}
Suppose that $f$ is in $R_t^{\out}\cap R_{t+1}$.
Let $y'_{j}$ be a pendant vertex attached to $y_{j}$ embedded 
inside $f$ and let $x_j'$ be a pendant attached to $x_j$ embedded 
in the face on $\boundary R_t$.
The shortest $x_{j}'$-$y_{j}'$ and $x_{j}'$-$v$ paths
diverge at some point.
We query the \MSSP{} structure (part (A)) 
to get the least common ancestor $w$ of $y_j'$ and $v$
and the three edges $e_{y_j'},e_v,e_w$ around $w$,
then determine the left/right relationship as in Case 1.
(If $j^\star = j-1$ then we would reverse the answer
due to the reversed orientation of the $x_{j-1}$-$y_{j-1}$ 
subpath w.r.t.~$C$.)
Once again, some of $e_{y_j'},e_v,e_w$ may be shortcut edges 
between $\boundary R_{t+1}$-vertices or artificial pendant edges.
See Figure~\ref{fig:leftright}(b)

\paragraph{Case 2b.}
Now suppose $f$ lies in $R_{t+1}^{\out}$.
We get from part (E) the last vertices 
$\hat{x}_j, \hat{x}_{j-1} \in \boundary R_{t+1}$ 
that lie on the $s_j$-$y_j$ and $s_{j-1}$-$y_{j-1}$ shortest paths.
We ask the \MSSP{} structure of part (A) for the least common ancestor
$w$ of $\hat{x}_j$ and $v$ in the shortcutted SSSP tree rooted at $x_j'$,
and also get the three incident edges $e_{\hat{x}_j},e_v,e_w$.
The edges $e_v$ and $e_w$ exist and are different, 
but $e_{\hat{x}_j}$ may not exist if $w=\hat{x}_j$, i.e., 
if $v$ is a descendant of $\hat{x}_j$.  
If all three edges $\{e_{\hat{x}_j},e_v,e_w\}$ exist 
we can determine whether $v$ lies to the right of $C$ as in Case 1 or 2a.

\paragraph{Case 2b(i).}
Suppose $w = \hat{x}_j$ and $e_{\hat{x}_j}$ does not exist.
Let $e_v = (\hat{x}_j,\hat{v})$.
If $\hat{v} \in \boundary R_{t+1}$ then $e_v$ represents a path that is completely
contained in $R_{t+1}^{\out}$.  Thus, if we walk clockwise around
the hole of $R_{t+1}^{\out}$ on $\boundary R_{t+1}$ 
and encounter $\hat{x}_j,\hat{v},\hat{x}_{j-1}$ in that order then 
$v$ lies to the right of $C$, and if we encounter them in the reverse
order then $v$ lies to the left of $C$.  See Figure~\ref{fig:leftright}(c).

\paragraph{Case 2b(ii).}
Finally, suppose 
$\hat{v} \not\in \boundary R_{t+1}$ 
and $e_v = (\hat{x}_j,\hat{v})$ is a normal edge in $G$.
Redefine $e_{\hat{x}_j}$ to be the first edge on the 
path from $\hat{x}_j$ to $y_j$.\footnote{We could store
$e_{\hat{x}_j}$ in part (E) of the data structure but that is not necessary.
If $e_0,e_1$ are the edges adjacent to $\hat{x}_j$ on the boundary 
cycle of $\boundary R_{t+1}$, then we can
use any member of $\{e_0,e_1\}\backslash \{e_w\}$ as a proxy
for $e_{\hat{x}_j}$.}
Now we can determine if $v$ is to the right of $C$
by looking at the clockwise order of $e_{w},e_{v},e_{\hat{x}_j}$ 
around $\hat{x}_j$.
\end{proof}

\medskip

As pointed out in Remark~\ref{remark:PointLocate}, 
Lemma~\ref{lem:MSSP} does not immediately imply
that Line 6 of $\SitePathIndicator$ and Line 1 of $\PieceSearch$
can be implemented efficiently.  Gawrychowski et al.'s~\cite{GawrychowskiMWW18}
implementation of $\PointLocate$ requires \MSSP{} access to $R_t^{\out}$,
whereas part (A) only lets us query vertices in $R_t^{\out}\cap R_{t+1}$.
Gawrychowski et al.'s algorithm is identical to $\CentroidSearch$, 
except that $\Navigation$ is done directly with \MSSP{} structures.
Suppose we are currently at $f^*$ in the centroid decomposition, 
with $y_j,s_j$ defined as usual.  Gawrychowski's algorithm finds
$j$ minimizing $\omega(s_j)+\dist_{R_t^{\out}}(s_j,v)$ using three
distance queries to the \MSSP{} structure, then decides whether
the $s_j'$-$v$ shortest path is a prefix of the $s_j'$-$y_j'$ shortest path,
and if not, which direction it branches in.\footnote{$s_j',y_j'$ being pendant
vertices attached to $s_j,y_j$, as in Lemma~\ref{lem:chordside}.}
If $f$ is in $R_{t}^{\out}\cap R_{t+1}$ we can proceed exactly 
as in Gawrychowski et al.~\cite{GawrychowskiMWW18}.
If not, we retrieve from part (E) the last vertex $\hat{x}$ 
of $\boundary R_{t+1}$ on the $s_j$-$y_j$ shortest path,
use $\hat{x}$ in lieu of $y_j'$ for the LCA queries, and tell whether
the $s_j'$-$v$ path branches to the right exactly as in 
Lemma~\ref{lem:chordside}, Case 2b.

\section{Analysis}\label{sect:analysis}

This section constitutes a proof of the claims of Theorem~\ref{thm:maintheorem}
concerning space complexity and query time; refer to Appendix~\ref{sect:construction}
for an efficient construction algorithm.

Combining Lemmas~\ref{lem:MSSP} and \ref{lem:PointLocate} (see Section~\ref{sect:sidequeries}),
$\PointLocate$ runs in $O(\kappa\log n\log\log n)$ time.
Together with Lemma~\ref{lem:chordside} it follows that
$\PieceSearch$ also takes $O(\kappa\log n\log\log n)$ time.
$\SitePathIndicator$ uses $\PointLocate$, the $\MSSP$ structure,
and $O(1)$-time tree operations on $T_q^{R_i}$ and the $\vec{r}$-hierarchy 
like least common ancestors and level ancestors~\cite{HarelT84,BenderF00,BenderF04,Hagerup20}.
Thus $\SitePathIndicator$ also takes $O(\kappa\log n\log\log n)$ time.
The calls to $\MaximalChord$ and $\AdjacentPiece$ in $\ChordIndicator$
take $O(\log n\log\log n)$ time by Lemma~\ref{lem:piecetree}, and testing
which side of a chord $v$ lies on takes $O(\kappa\log\log n)$ 
time by Lemma~\ref{lem:chordside}.  The bottleneck in $\ChordIndicator$
is still $\PieceSearch$, which takes $O(\kappa\log n\log\log n)$ time.
The only non-trivial parts of $\Navigation$ are calls to
$\SitePathIndicator$ and $\ChordIndicator$, so it, too, 
takes $O(\kappa\log n\log\log n)$ time.

An initial call to $\CentroidSearch$ (Line 5 of $\Dist$) 
generates at most $\log n$
recursive calls to $\CentroidSearch$, culminating in the last recursive call
making 1 or 2 calls to $\Dist$ with the ``$i$'' parameter incremented.
Excluding the cost of recursive calls to $\Dist$, 
the cost of $\CentroidSearch$ is dominated by calls to $\Navigation$,
i.e., an initial call to $\CentroidSearch$ costs 
$\log n \cdot O(\kappa\log n\log\log n) = O(\kappa\log^2 n\log\log n)$ time.
Let $T(i)$ be the cost of a call to $\Dist(u_i,v,R_i)$.  We have
\begin{align*}
    T(m-1) &= O(\kappa\log\log n)                       & \mbox{$\Dist$ returns at Line 2 with one \MSSP{} query}\\
    T(i)    &= 2T(i+1) + O(\kappa\log^2 n\log\log n)
\end{align*}
It follows that the time to answer a distance query is 
$T(0) = O(2^m\cdot \kappa\log^2 n\log\log n)$.

\medskip

The space complexity of each part of the data structure is as follows.
(A) is $O(\kappa m n^{1+1/m+1/\kappa})$
by Lemma~\ref{lem:MSSP} and the fact that $r_{i+1}/r_i = n^{1/m}$.
(B) is $O(mn^{1+1/(2m)})$ since $\sqrt{r_{i+1}/r_i} = n^{1/(2m)}$.
(C) is $O(mn)$ since $\sum_{i} n/r_i\cdot (\sqrt{r_i})^2 = O(mn)$.
(D) is $O(mn\log n)$ by Lemma~\ref{lem:partD},
and
(E) is $O(m)$ times the space cost of (B) and (C), namely
$O(m^2 n^{1+1/(2m)})$.  The bottleneck is (A).

\medskip

We now explain how $m,\kappa$ can be selected to achieve the 
extreme space and query complexities claimed Theorem~\ref{thm:maintheorem}.
To optimize for query time, pick $\kappa = m$ to be any function 
of $n$ that is $\omega(1)$ and $o(\log\log n)$.  Then
the query time is
\[
O(2^m\kappa \log^2 n\log\log n) = \log^{2+o(1)} n
\]
and the space is
\[
O(m\kappa n^{1+1/m+1/\kappa})=n^{1+o(1)}.
\]
To optimize for space, choose $\kappa = \log n$ 
and $m$ to be a function that is 
$\omega(\log n/\log\log n)$ and $o(\log n)$.
Then the space is
\[
O\left(m\kappa n^{1+1/m+1/\kappa}\right)
=
o\left(n^{1+1/m}\log^{2}n\right)
=
n\cdot 2^{o(\log\log n)}\cdot \log^{2} n
=
n\log^{2+o(1)}n,
\]
and the query time
\[
O(2^{m}\kappa\log^{2}n \log\log n)
=
2^{o(\log n)}\log^{3}n\log\log n
=
n^{o(1)}.
\]

\subsection{Speeding Up the Query Time}

Observe that the space of (B) is asymptotically smaller than the space of (A).
Replace (B) with (B')
\begin{enumerate}
    \item[(B')] {\bf (Voronoi Diagrams)}
    Fix $i$, a region $R_{i}\in\mathcal{R}_{i}$ with ancestors 
    $R_{i+1}\in \mathcal{R}_{i+1}$ and $R_{i+4} \in \mathcal{R}_{i+4}$.
    For each $q \in \boundary R_{i}$ store 
    \begin{align*}
    \VDout(q,R_{i+1}) &= \VD^*[R_{i+1}^{\out},\boundary R_{i+1},\omega]\\
    \VDfarout(q,R_{i+4}) &= \VD^*[R_{i+4}^{\out},\boundary R_{i+4},\omega] & \mbox{only if $i<m-4$}
    \end{align*}
    with $\omega(s)=\dist_G(q,s)$ in both cases.
    Over all regions $R_i$, the space for storing all $\VDout$s
    is $\tilde{O}(n^{1+1/(2m)})$ since $\sqrt{r_{i+1}/r_i} = n^{1/(2m)}$
    and the space for $\VDfarout$s is $\tilde{O}(n^{1+2/m})$ since 
    $\sqrt{r_{i+4}/r_i}=n^{2/m}$.
\end{enumerate}

Now the space for (A) is 
$\tilde{O}(n^{1+1/m+1/\kappa})=\tilde{O}(n^{1+2/m})$ is balanced
with (B').  
In the $\Dist$ function we now consider three possibilities.
If $v\in R_{i+1}$ we use part (A) to solve the problem without recursion.
If $v\not\in R_{i+1}$ but $v\in R_{i+4}$ we proceed as usual,
calling $\CentroidSearch(\VDout(u_i,R_{i+1}),v,\cdot)$,
and if $v\not\in R_{i+4}$ we call
$\CentroidSearch(\VDfarout(u_i,R_{i+4}),v,\cdot)$.
Observe that the depth of the $\Dist$-recursion 
is now at most $t/4+O(1) < m/4+O(1)$, giving us
a query time of $O(m2^{m/4}\log^2 n\log\log n)$
with space $\tilde{O}(n^{1+2/m})$.

\section{Conclusion}

In this paper we have proven that it is possible to simultaneously 
achieve optimal space \underline{\emph{or}} query time, up to a $\log^{2+o(1)} n$ factor, 
and near-optimality in the other complexity measure, 
up to an $n^{o(1)}$ factor.  The main open question in this area
is whether there exists an exact distance oracle with 
$\tilde{O}(n)$ space and $\tilde{O}(1)$ query time.
This will likely require new insights into the structure
of shortest paths, which could lead, for example, 
to storing correlated versions of Voronoi diagrams more efficiently, 
or avoiding the binary branching recursion in our query algorithm.

\medskip
\paragraph{Acknowledgements.} We thank Danny Sleator and Bob Tarjan for discussing 
                                update/query time tradeoffs for dynamic trees.

\bibliographystyle{plain}
\bibliography{references}

\appendix

\section{\MSSP{} via Euler Tour Trees (Proof of Lemma~\ref{lem:MSSP})}\label{sect:Euler}

Let us recall the setup.  We have a planar graph $H$ with a distinguished face $f$, and wish to answer
$\dist_H(s,v)$ queries w.r.t.~any $s$ on $f$ and $v\in V(H)$, and LCA queries w.r.t. any $s$ on $f$
and $u,v\in V(H)$.  Klein~\cite{Klein05} proved that if we move the source vertex $s$ around $f$ and 
record all the changes to the SSSP tree, every edge in $E(H)$ can be swapped into and out of the SSSP 
at most once, i.e., there are $O(|H|)$ updates in total.  
Thus, if we maintain the SSSP tree as the source travels around $f$ 
in a dynamic data structure with update time $U$ and query time $T$ (for distance and LCA queries), 
the universal persistence method for RAM data structures (see~\cite{Dietz89}) yields
an \MSSP{} data structure with space $O(|H|U)$ and query time $O(T\log\log|H|)$.
Thus, to establish Lemma~\ref{lem:MSSP} it suffices to design a dynamic data structure for the following:

\begin{description}
\item[InitTree$(s^\star,T)$:] Initialize a directed spanning tree $T$ from root $s^\star$.  Edges have real-valued lengths.
\item[Swap$(v,p,l)$:] Let $p'$ be the parent of $v$; $p$ is not a descandant of $v$.  
        Update $T \gets T-\{(p',v)\}\cup \{(p,v)\}$, where $(p',v)$ has length $l$.
\item[Dist$(v)$:] Return $\dist_T(s^\star,v)$.
\item[LCA$(u,v)$:] Return the LCA $y$ of $u$ and $v$ and the first edges $e_u,e_v$
                 on the paths from $y$ to $u$ and from $y$ to $v$, respectively.
\end{description}

Here $s^\star$ will be a fixed root vertex embedded in $f$ with a single, weight-zero, 
out-edge to the current root on $f$.  Changes to the SSSP tree are effected with $O(|H|)$
\textbf{Swap} operations.
Klein~\cite{Klein05} and Gawrychowski~\cite{GawrychowskiMWW18} use 
Sleator and Tarjan's Link-Cut trees~\cite{SleatorT83},
which support \textbf{Swap}, \textbf{Dist}, and \textbf{LCA} (among other operations) in $O(\log|T|)$ time.
We will use a souped-up version of Henzinger and King's~\cite{HenzingerK99} Euler Tour trees.
Let $\ET(T)$ be an Euler tour of $T$ 
starting and ending at $s^\star$.  
The elements of $\ET(T)$ are edges,
and each edge of $T$ appears twice in $\ET(T)$, once in each direction.  Each edge in $T$ points to its two occurrences in $\ET(T)$.

\begin{figure}
    \centering
    \begin{tabular}{c@{\hspace{2cm}}c}
    \scalebox{.4}{\includegraphics{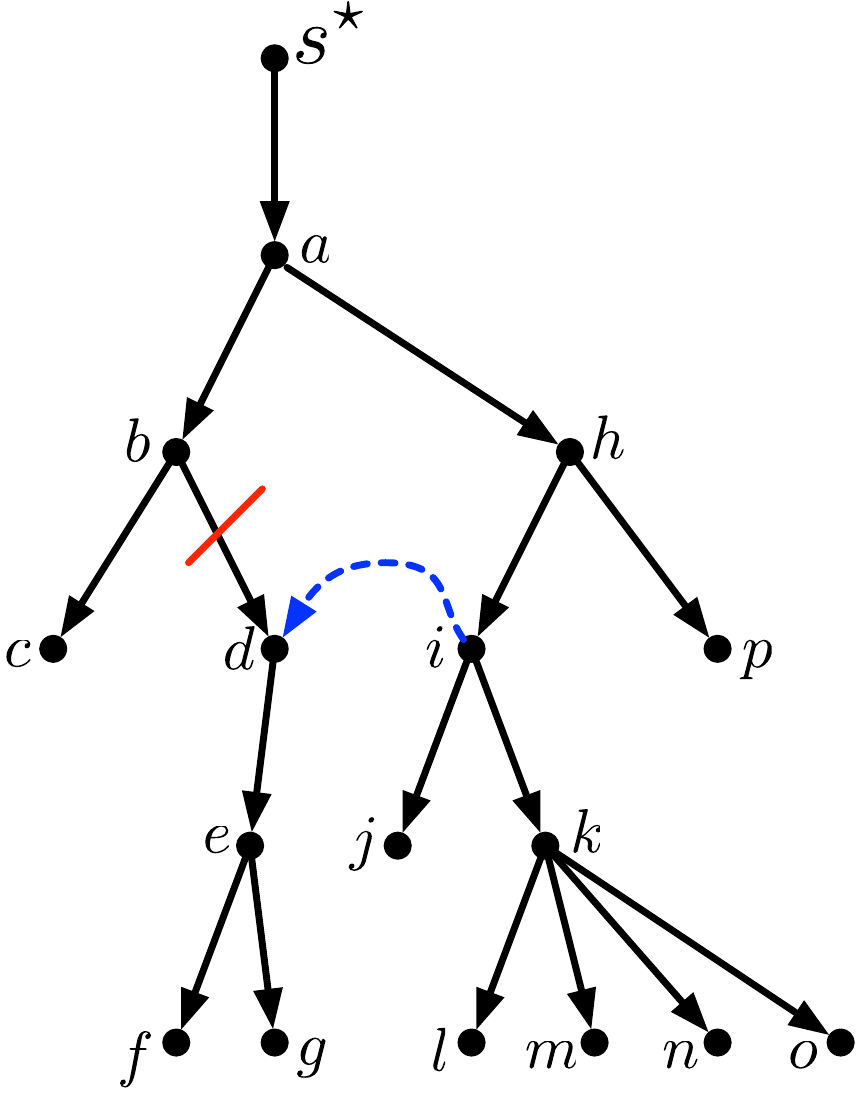}}
    &\scalebox{.4}{\includegraphics{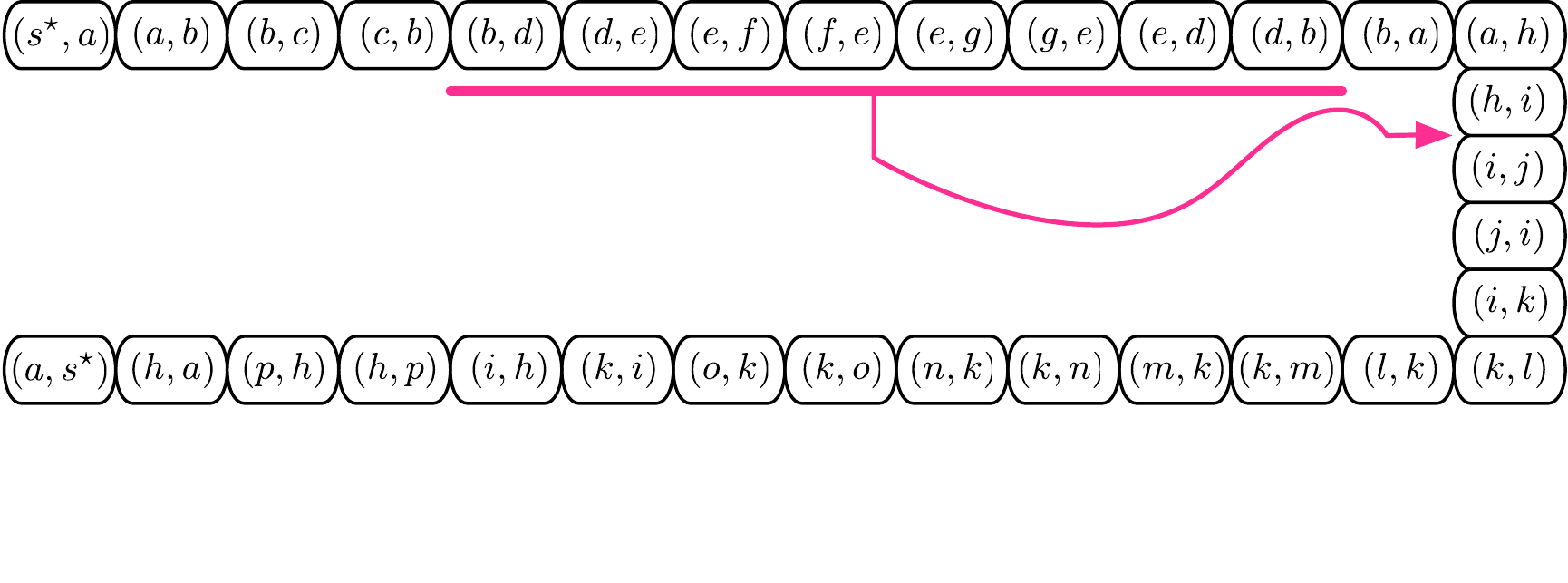}}
    \end{tabular}
    \caption{The effect of \textbf{Swap}$(d,i,\cdot)$ on the Euler Tour.  The interval $((b,d),(d,e),\cdots,(e,d),(d,b))$
    is spliced out and inserted between
    $(h,i)$ and $(i,j)$, and the elements $(b,d),(d,b)$ are renamed $(i,d),(d,i)$.}
    \label{fig:ET}
\end{figure}
Suppose $T_{\operatorname{ante}}$ is the tree before a \textbf{Swap} operation and $T_{\operatorname{post}}$ the tree afterward.
It is easy to see that $\ET(T_{\operatorname{post}})$ can be derived from $\ET(T_{\operatorname{ante}})$ by $O(1)$ splits
and concatenates, and renaming the two elements corresponding to the swapped edge. 
See Figure~\ref{fig:ET}.
We will argue that
the dynamic tree operations \textbf{Swap}, \textbf{Dist}, \textbf{LCA} can be implemented
using the following list operations.

\begin{description}
\item[InitList$(L)$:] Initialize a list $L$ of weighted elements.
\item[Split$(e_0)$:] Element $e_0$ appears in some list $L$.  Split $L$ immediately after element $e_0$, resulting in two lists.
\item[Concatenate$(L_0,L_1)$:] Concatenate $L_0$ and $L_1$, resulting in one list.
\item[Add$(e_0,e_1,\delta)$:] Here $e_0,e_1$ are elements of the same list $L$.
                Add $\delta\in \mathbb{R}$ to the weight of all elements in $L$ between $e_0$ and $e_1$ inclusive.
\item[Weight$(e_0)$:] Return the weight of $e_0$.
\item[RangeMin$(e_0,e_1)$:] Return the minimum-weight element between $e_0$ and $e_1$ inclusive. 
                            If there are multiple minima, return the first one.
\end{description}

To implement \textbf{Dist} and \textbf{LCA} we will 
actually use the list data structure
with different weight functions.  For \textbf{Dist}, 
the weight of an edge $(x,y)$ in $\ET(T)$ is 
$\dist_T(s^\star,y)$.  Thus, \textbf{Dist}
is answered with a call to \textbf{Weight}.  
Each \textbf{Swap}$(v,p,l)$ is effected with $O(1)$ \textbf{Split}
and \textbf{Concatenate} operations, renaming the elements 
of the swapped edge, 
as well as one \textbf{Add}$(e_0,e_1,\delta)$ operation.
Here $(e_0,\ldots,e_1)$ is the sub-list corresponding 
to the subtree rooted at $v$, 
and $\delta = \dist_{T_{\operatorname{post}}}(s^\star,v) - \dist_{T_{\operatorname{ante}}}(s^\star,v)$ is the change
in distance to $v$, and hence all descendants of $v$.  

To handle \textbf{LCA} queries, we use
the list data structure where the weight of $(x,y)$ is the depth of $y$ in $T$, i.e., the distance from $s^\star$ to $y$ under the unit length function.  
Once again, a \textbf{Swap}
is implemented with $O(1)$ \textbf{Split} and
\textbf{Concatenate} operations, and one \textbf{Add} operation.  Consider an \textbf{LCA}$(u,v)$ query.  
Let $e_0=(p_u,u),e_1=(p_v,v)$ be the edges into $u$ and $v$ from their respective parents,
and suppose that $e_0$ appears before $e_1$ in $\ET(T)$.\footnote{As 
we will see, it is easy to determine which comes first.}
A call to \textbf{RangeMin}$(e_0,e_1)$ returns the \emph{first} edge
$\hat{e} = (x,y)$ in the interval $(e_0,\ldots,e_1)$ minimizing the depth of $y$.
It follows that $y$ is the LCA of $u$ and $v$.  Furthermore, by the tiebreaking rule,
if $\hat{e} \neq e_0$ then $\hat{e}=e_u$ is the (reversal of the) edge leading from $y$ towards $u$.
If $\hat{e} = e_0$ then $v$ is a descendant of $u$ and $e_u$ does not exist.
To find $e_v$, we retrieve the edge $\tilde{e}=(y,p_y)$ in $\ET(T)$ from $y$ to its parent
and let $\tilde{e}'$ be its predecessor in $\ET(T)$.
(Note that since $s^\star$ has degree 1, $\tilde{e},\tilde{e}'$ always exist.)
We call \textbf{RangeMin}$(e_1,\tilde{e}')$.
Once again, by the tiebreaking rule it returns the 
\emph{first} edge $e_v = (x',y)$ incident to $y$ in $(e_1,\ldots,\tilde{e}')$, 
which is the (reversal of the) first edge on the path from $y$ to $v$.
See Figure~\ref{fig:ET-LCA}.

\begin{figure}
    \centering
    \scalebox{.4}{\includegraphics{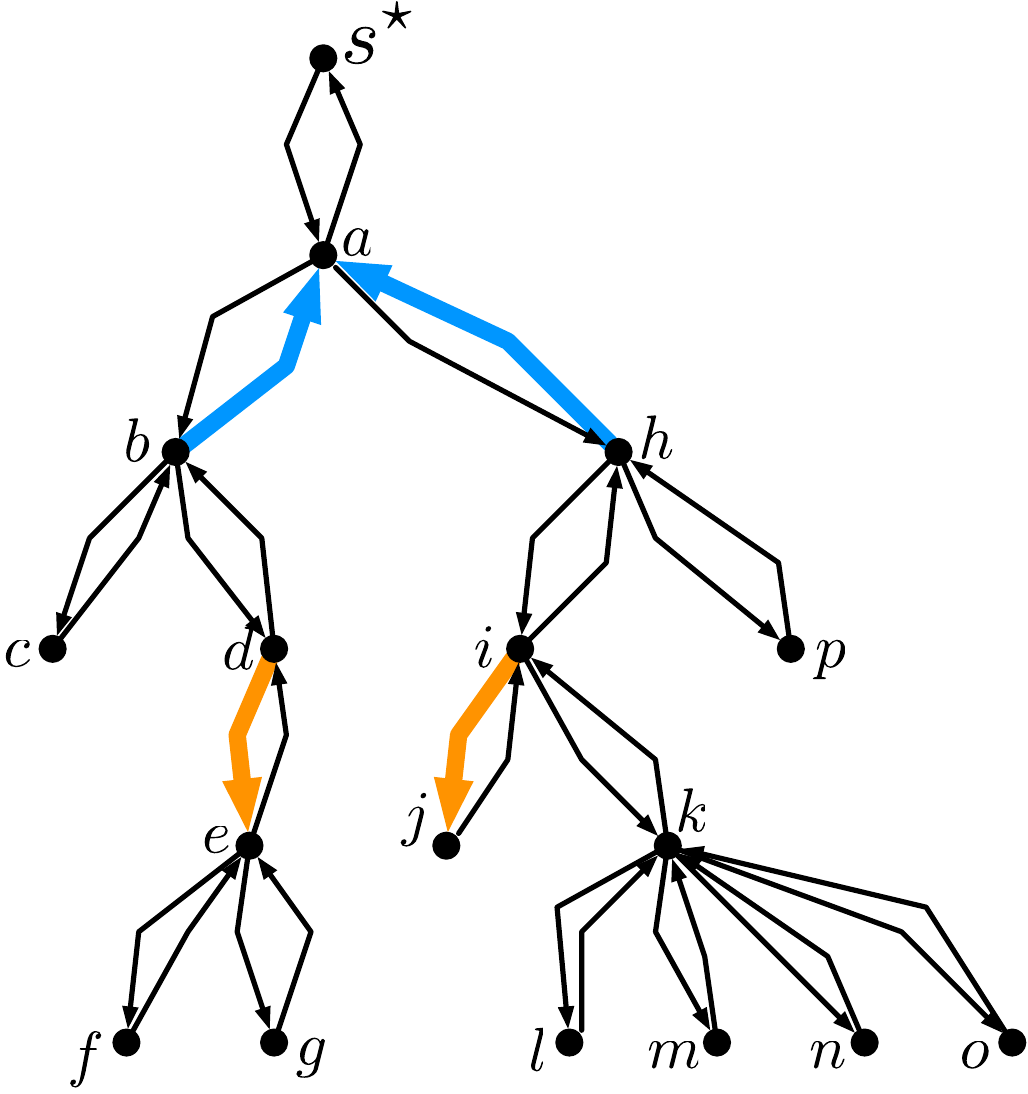}}
    \caption{An illustration of an \textbf{LCA}$(e,j)$
    query.  We do a \textbf{RangeMin} query on
    the interval $e_0=(d,e),\ldots,(i,j)=e_1$ and retrieve the edge $\hat{e}=e_e=(b,g)$ with weight $\operatorname{depth}_T(g)$.
    We then find $\tilde{e}=(g,s^\star)$ and 
    its predecessor $\tilde{e}'=(h,g)$.
    Another \textbf{RangeMin} query on the interval
    $(i,j),\ldots,(h,g)$ returns $e_j = (h,g)$.}
    \label{fig:ET-LCA}
\end{figure}

\medskip

We have reduced our dynamic tree problem to a dynamic
weighted list problem.  We now explain how the dynamic
list problem can be solved with balanced trees.

Fix a parameter $\kappa\ge 1$ and let $n$ be the total number of elements in all lists.
We now argue that \textbf{Split}, \textbf{Concatenate}, and \textbf{Add} can be implemented in
$O(\kappa n^{1/\kappa})$ time and \textbf{Weight} and \textbf{RangeMin} take $O(\kappa)$ time.
We store the elements of each list $L$ at the leaves of a rooted tree $\mathcal{T}(L)$.  
It satisfies the following invariants.
\begin{enumerate}
    \item[I.] Each node $\gamma$ of $\mathcal{T}(L)$ stores a weight 
    offset $w(\gamma)$, a min-weight value $\min(\gamma)$ and 
    a pointer $\operatorname{ptr}(\gamma)$.  The weight of (leaf) 
    $e\in L$ is the sum of the $w(\cdot)$-values of its ancestors, including $e$.
    The sum of $\min(\gamma)$ and the $w(\cdot)$-values of all strict ancestors of $\gamma$
    is exactly the weight of the minimum weight descendant of $\gamma$, and $\operatorname{ptr}(\gamma)$
    points to this element.
    \item[II.] Non-root internal nodes have between $n^{1/\kappa}$ and $3n^{1/\kappa}$ children. 
    In particular, the tree has height at most $\kappa$.
    \item[III.] Each internal node $\gamma$ maintains an $O(1)$-time 
    range minimum structure~\cite{BenderF00}
    over the vector of $\min(\cdot)$-values of its children.
\end{enumerate}

It is easy to show that \textbf{Split} and \textbf{Concatenate} can be implemented to satisfy
Invariant II by destroying/rebuilding $O(1)$ nodes at each level of $\mathcal{T}$.  Each costs
$O(n^{1/\kappa})$ time to update the information covered by Invariants I and III. 
The total time is therefore $O(\kappa n^{1/\kappa})$.  By Invariant I, a \textbf{Weight}$(e_0)$
query takes $O(\kappa)$ time to sum all of $e_0$'s ancestors' $w(\cdot)$-values.  
Consider an \textbf{Add}$(e_0,e_1,\delta)$ or \textbf{RangeMin}$(e_0,e_1)$ operation.
By Invariant II, the interval $(e_0,\ldots,e_1)$ is covered by $O(\kappa n^{1/\kappa})$ $\mathcal{T}$-nodes,
and furthermore, those nodes can be arranged into less than $2\kappa$ contiguous intervals of siblings.
Thus, an \textbf{Add}$(e_0,e_1)$ can be implemented in $O(\kappa n^{1/\kappa})$ time by
adding $\delta$ to the $w(\cdot)$-values of these nodes and rebuilding the affected 
range-min structures from Invariant III.
A \textbf{RangeMin} is reduced to $O(\kappa)$ range-minimum queries (from Invariant III) 
and adjusting the answers by the $w(\cdot)$-values of their ancestors (Invariant I).
Each range-min query takes $O(1)$ time and there are $O(\kappa)$ ancestors with relevant $w(\cdot)$-values.
Thus \textbf{RangeMin} takes $O(\kappa)$ time.

\medskip

We have shown that the dynamic tree operations necessary for an \MSSP{} structure
can be implemented with a flexible tradeoff between update time and query time.
Moreover, this lower bound meets the \Patrascu-Demaine lower bound~\cite{PatrascuD06}.
We leave it as an open problem to implement the full complement of operations supported
by Link-Cut trees, with update time $O(\kappa n^{1/\kappa})$ and query time $O(\kappa)$.

\section{Multiple Holes and Nonsimple Cycles}\label{sect:MultipleHoles}

We have assumed for simplicity that all regions are bounded by a
simple cycle, and therefore have a single hole.  We now show
how these assumptions can be removed.  

Let us first illustrate how a 
region $R$ may get
a hole with a non-simple 
boundary cycle.
The hierarchical decomposition 
algorithm of Klein, Mozes, and Sommer~\cite{KleinMS13} produces a binary decomposition tree, 
of which our $\vec{r}$-decomposition is a coarsening.
It proceeds by finding a separating
cycle (as in Miller~\cite{Miller86}), and recursively decomposes
the graph inside the cycle and outside the cycle.\footnote{The Klein et al.~\cite{KleinMS13} algorithm rotates between finding
separators w.r.t.~number of vertices, number of boundary vertices, and number of holes, but this is not relevant to the present discussion.}
At intermediate stages the working graph contains several holes,
but Miller's theorem~\cite{Miller86} only guarantees that
a small cycle separator exists if the graph is triangulated.
To that end, the decomposition \cite{KleinMS13} puts an artificial vertex inside each hole and triangulates the hole.  See Figure~\ref{fig:SimpleCycle}(a,b).  If the cycle separator $C$
(blue cycle in Figure~\ref{fig:SimpleCycle}(b)) includes
a hole-vertex $v$, we splice out $v$ and replace it with an
interval of the boundary of the hole.  If $C$ also includes
edges on the boundary of the hole (Figure~\ref{fig:SimpleCycle}(c)), the modified
cycle may not be simple.  If this is the case, 
we ``cut'' along non-simple parts of the cycle,
replicating all such vertices and their incident cycle edges.
We then join pairs of identical vertices with zero-length
edges (pink edges in~Figure~\ref{fig:SimpleCycle}(c)),
and triangulate with large-length edges.
This transformation clearly preserves planarity and
does not change the underlying metric.\footnote{Given a
$\dist_G(u,v)$ query, we can map it to $\Dist(u',v',R_0)$,
where $u'$ and $v'$ are any of the copies of $u$ and $v$, 
respectively, and $R_0=\{u'\}$.}

\begin{figure}
    \centering
    \begin{tabular}{@{\hspace{1cm}}c@{\hspace{1cm}}c}
    \scalebox{.3}{\includegraphics{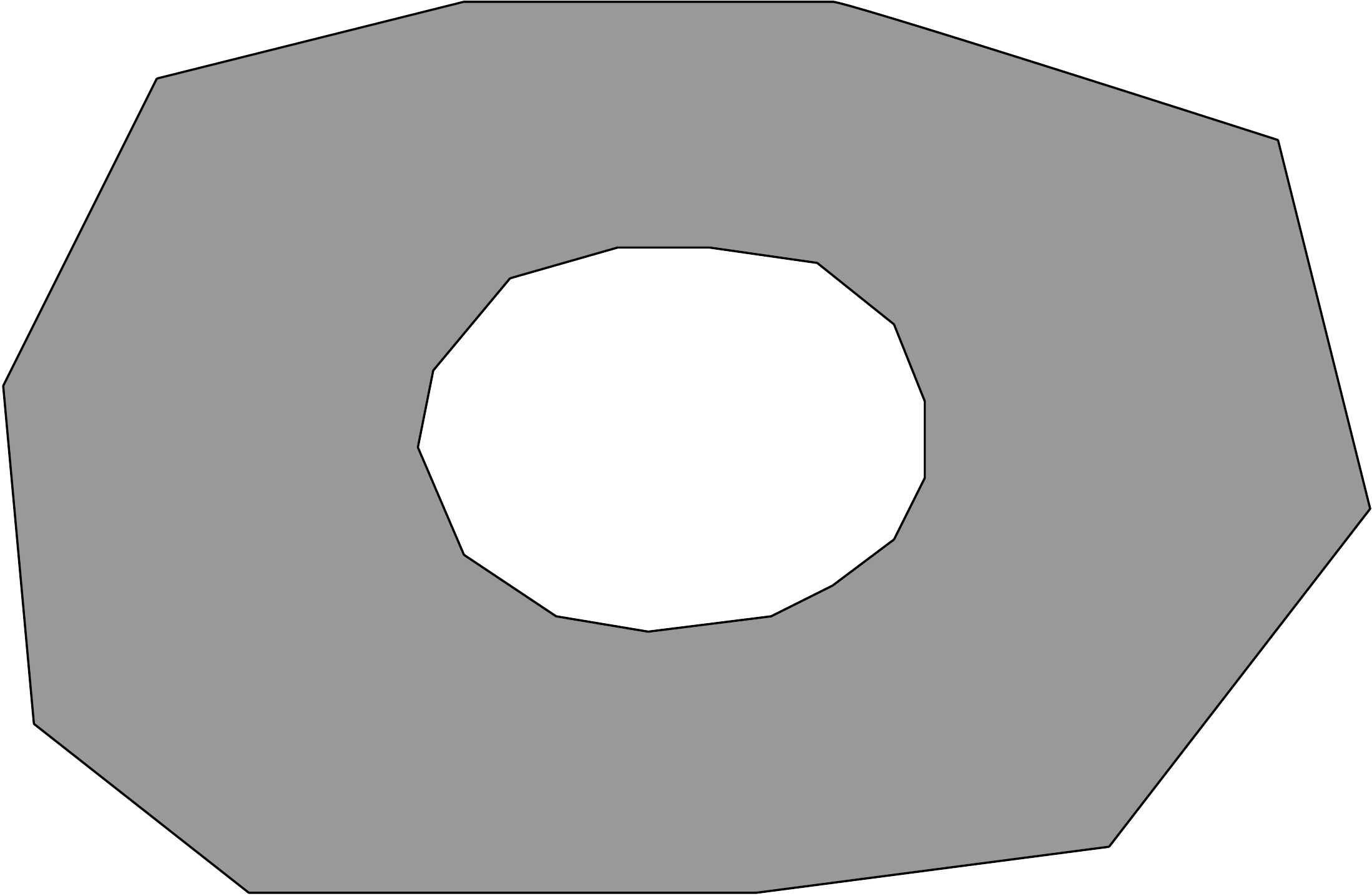}}
    &\scalebox{.3}{\includegraphics{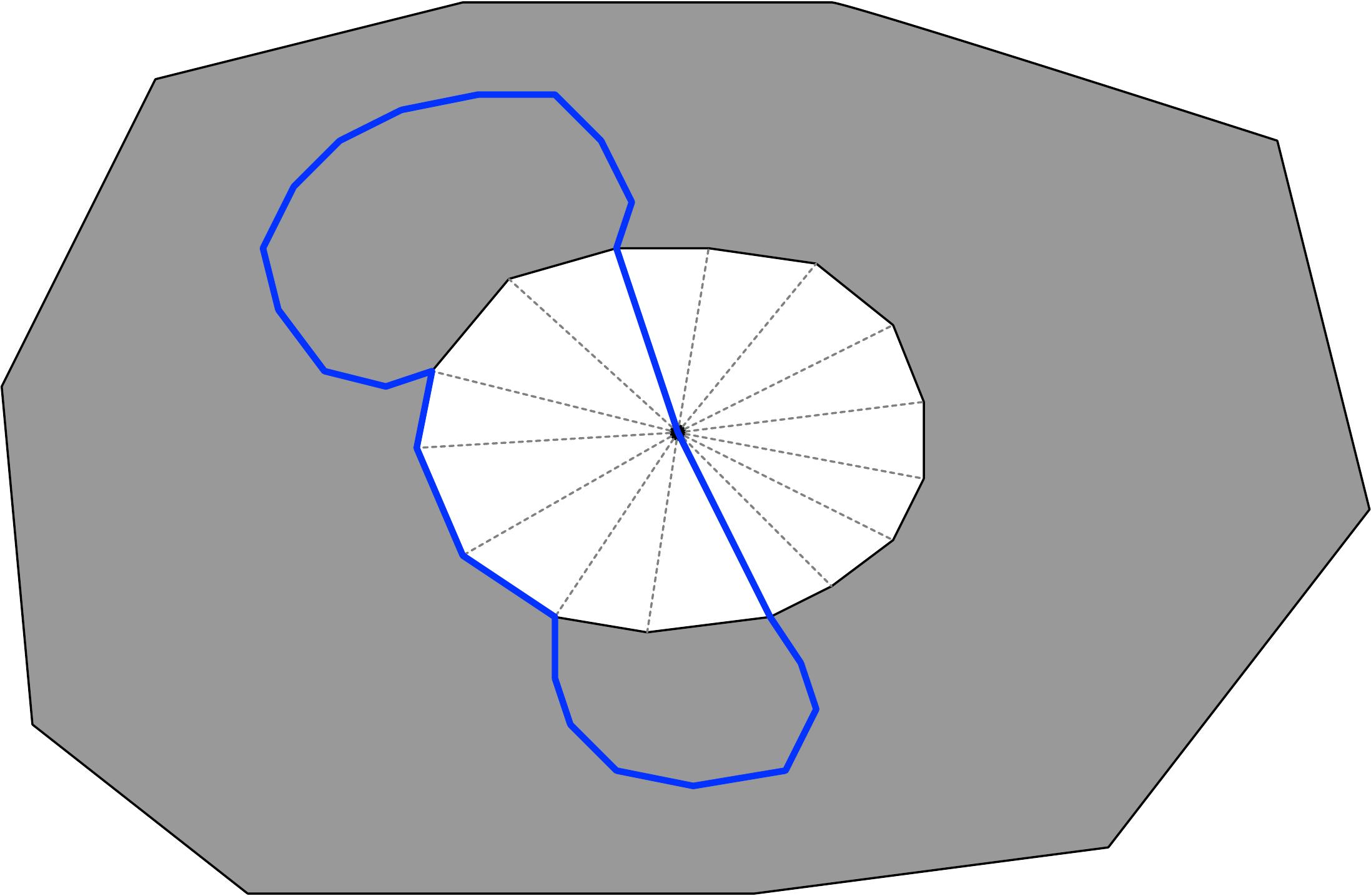}}\\
    &\\
    {\bf (a)} & {\bf (b)}\\
    &\\
    \scalebox{.3}{\includegraphics{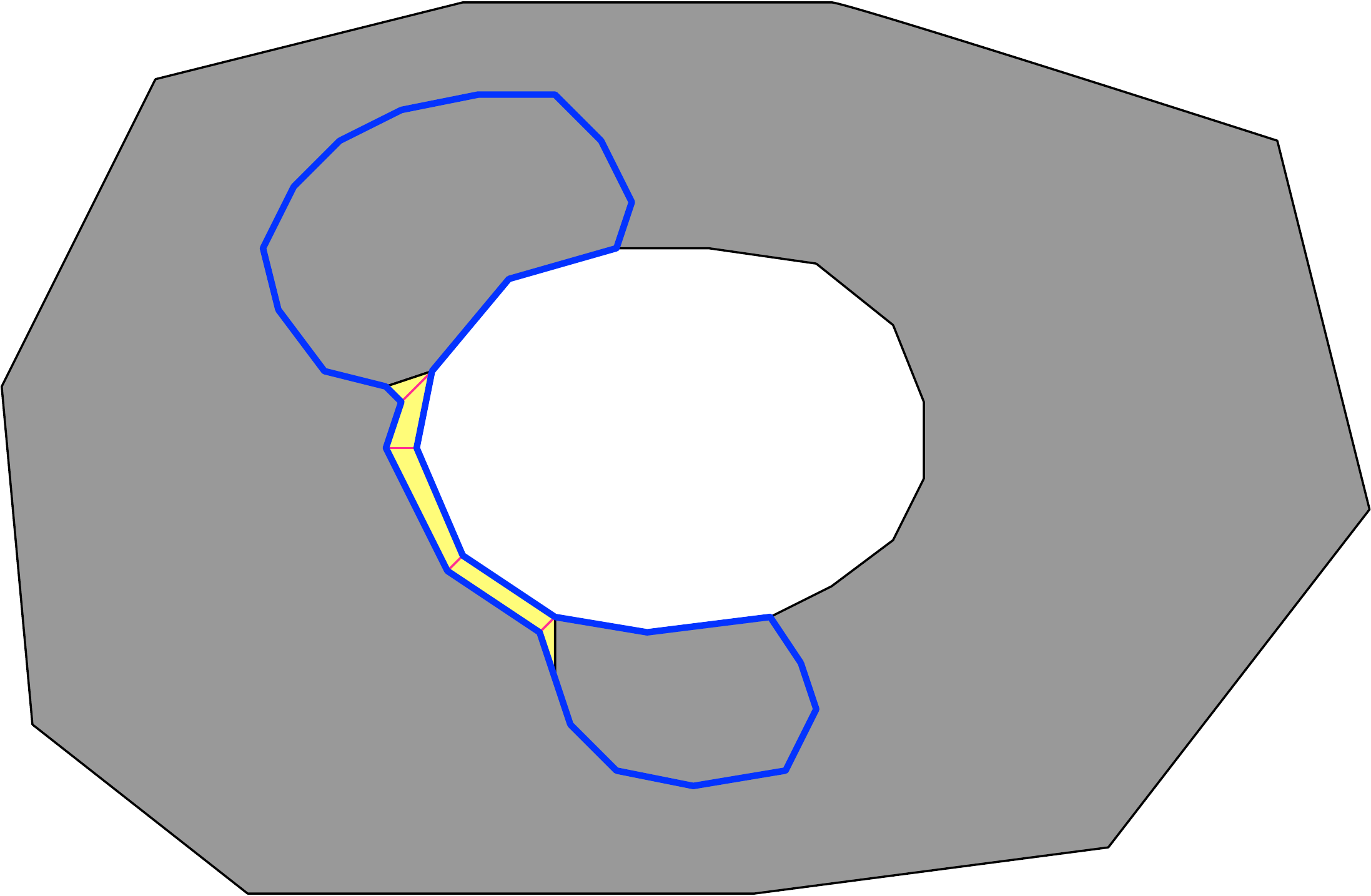}}
    &\scalebox{.3}{\includegraphics{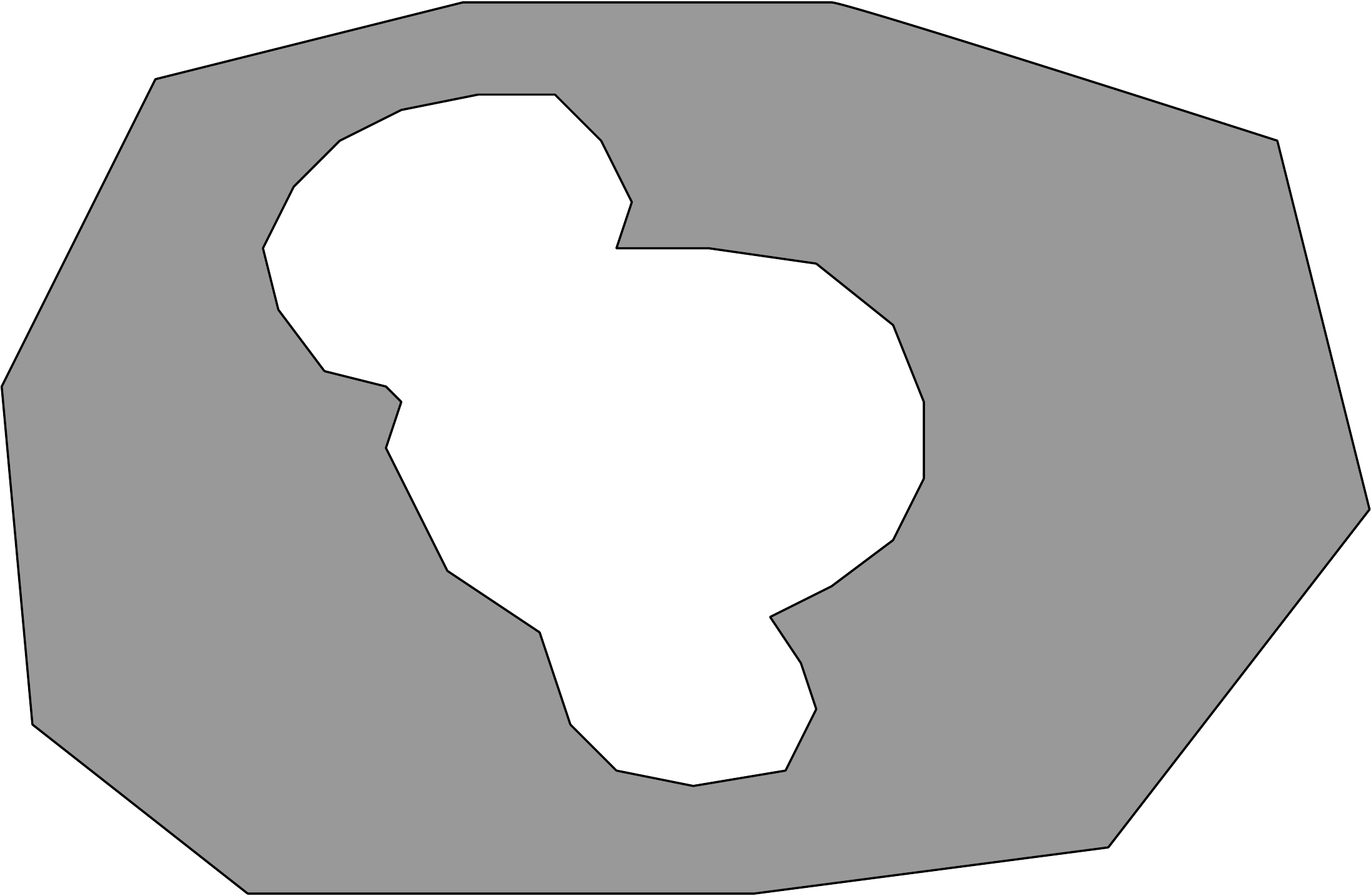}}\\
    &\\
    {\bf (c)} & {\bf (d)}
    \end{tabular}
    \caption{{\bf (a)} A subgraph with two holes. 
    {\bf (b)} We put a vertex in each hole and triangulate the hole.
    (The triangulation of the exterior hole is not drawn, for clarity.)
    A simple cycle sparator (blue curve) is found in this graph.
    {\bf (c)} The cycle is mapped to a possibly non-simple cycle
    in the original graph that avoids hole-vertices.  We cut along non-simple parts of the cycle, duplicating the vertices and their adjacent edges on the cycle.
    {\bf (d)} The graph remaining after removing the subgraph enclosed 
    by the cycle from {\bf (c)}.}
    \label{fig:SimpleCycle}
\end{figure}

Turning to the issue of multiple holes, 
we first make some observations about their structural organization.
Fix any hole $g$ of region $R_{i+1}$ and let $R_i$ be a child of 
$R_{i+1}$.  There is a unique
hole $\parent_{R_i}(g)$ in $R_i$ such that $g$ lies 
in $R_i^{\parent_{R_i}(g),\out}$, 
which we refer to as the \emph{parent} of $g$ in $R_i$.
(Note that the ancestry of holes goes in the opposite direction 
of the ancestry of regions in the $\vec{r}$-decomposition.)
In a distance query we only deal with a series of regions $R_0=\{u\},R_1,\ldots,R_m=G$.
The holes of these regions form a hierarchy, rooted at $\{u\}$, which we view as a degenerate hole.
For notational simplicity we use ``$g$'' to refer to the set of vertices on hole $g$.

\begin{lemma}\label{lem:ancestral-holes} 
(See~\cite[\S 4.3.2]{CharalampopoulosGMW19})
There is an $\tilde{O}(n)$-space data structure that,
given $u,v$ can report in $O(m)$ time the 
regions $R_0=\{u\}, R_1, \cdots, R_{t+1}$
and holes $h_0,h_1,\ldots,h_t$ such that 
$v \in R_{i}^{h_i,\out}$, $v\not\in R_t$, and $v\in R_{t+1}$.
\end{lemma}

The method of~\cite{CharalampopoulosGMW19} simply involves
doing a least common ancestor query of $\{u\}$ and $\{v\}$ 
in the ``full'' binary decomposition returned by the~\cite{KleinMS13} algorithm (from which our $\vec{r}$-decomposition is a coarsening) 
in order to retrieve $h_t$.  The holes $h_{t-1},\ldots,h_0$
can then be found by following parent pointers in $O(m)$ time.

\subsection{Data Structures}

The following modifications are made to parts (A)--(E) of the data structure.  
In all cases the space usage is unchanged, asymptotically.

\begin{enumerate}
    \item[(A)] {\bf (\MSSP{} Structures)}
    For each $i\in [0,m-1]$, each $R_i \in \mathcal{R}_i$
    with parent $R_{i+1}$ and each hole $h_{i}$ of $R_{i}$, 
    we build a \MSSP{} structure for $R_{i}^{h_{i},\out}$
    that answers distance queries and LCA queries w.r.t.~$R_{i}^{h_i,\out}$
    for vertices in $R_{i}^{h_{i},\out}\cap R_{i+1}$.

    \item[(B)] {\bf (Voronoi Diagrams)}
    For each $i\in [0,m-1]$, each $R_{i}\in\mathcal{R}_{i}$ with parent $R_{i+1}\in \mathcal{R}_{i+1}$,
    each hole $h_{i+1}$ of $R_{i+1}$ with parent $h_i = \parent_{R_i}(h_{i+1})$, and each 
    $q\in h_{i}$, we store the dual representation
    of Voronoi diagram $\VDout(q,R_{i+1},h_{i+1})$ defined to be
    $\VD^*[R_{i+1}^{h_{i+1},\out},h_{i+1},\omega]$ with $\omega(s) = \dist_G(q,s)$.
    
    \item[(C)] {\bf (More Voronoi Diagrams)}
    For each $i\in [1,m-1]$, each $R_{i}\in\mathcal{R}_{i}$, each hole $h_{i}$ of $R_{i}$, and each $q\in h_{i}$, we store $\VDout(q,R_{i},h_{i})$, which is $\VD^*[R_{i}^{h_{i},\out},h_{i},\omega]$ with $\omega(s) = \dist_G(q,s)$.
    
    \item[(D)] {\bf (Chord Trees; Piece Trees)}
    For each $i\in[1,m-1]$, each $R_{i}\in\mathcal{R}_{i}$, each hole $h_{i}$ of $R_{i}$, and source $q\in h_{i}$, we store a chord tree $T_{q}^{R_{i},h_{i}}$ obtained by restricting the SSSP tree with source $q$ to $h_{i}$.
    An edge in $T_{q}^{R_{i},h_{i}}$ is designated a chord if the corresponding path lies in $R_{i}^{h_{i},\out}$ and is internally vertex disjoint from $h_i$. 
    $\mathcal{C}_q^{R_i,h_i},\mathcal{P}_q^{R_i,h_i},\mathcal{T}_q^{R_i,h_i}$
    are defined analogously, and data structures are built to answer $\MaximalChord$
    and $\AdjacentPiece$ with respect to $q,R_i,h_i$.  

    \item[(E)] {\bf (Site Tables; Side Tables)}
    Fix an $i$ and a Voronoi diagram $\VDout = \VDout(u',R_{i},h_{i})$ from part (B) or (C). Let $f^{*}$ be any node in the centroid decomposition of $\VDout$ with $y_{j},s_{j}$ defined as usual, $j\in\{0,1,2\}$. 
    Let $R_{i'}\in\mathcal{R}_{i'}$ be an ancestor of $R_{i}$, $i'> i$, and $h_{i'}$ be a hole of $R_{i'}$ lying in $R_{i}^{h_{i},\out}$.  We store the first and last vertices $q,x$ on the shortest $s_{j}$-$y_{j}$ path that lie on $h_{i'}$, as well as $\dist_G(u',x)$.
    
    We also store whether $R_{i'}^{h_{i'},\out}$ lies to the left or right of the site-centroid-site chord $\chord{s_{j}y_{j}y_{j-1}s_{j-1}}$, or \textbf{Null} if the relationship cannot be determined.

\end{enumerate}

\subsection{Query}

At the first call to $\Dist(u,v,R_0)$ we apply Lemma~\ref{lem:ancestral-holes}
to generate the regions $R_1,\ldots,R_{t+1}$ 
and holes $h_1,\ldots,h_t$ 
that will be accessed in all recursive calls, in $O(m)$ time.

The shortest $u$-$v$ path in $G$ will eventually cross $h_{1},\ldots,h_{t}$.
The vertex $u_i$ is now defined to be the last vertex in $h_i$ on 
the shortest $u$-$v$ path.
Given $u_i$, we find $u_{i+1}$ by solving a point location
problem in $\VDout(u_{i},R_{i+1},h_{i+1})$.
The $\Navigation$ routine focuses on the subgraph 
$R_{t}^{h_{t},\out}$ rather than $R_{t}^{\out}$. 
The general problem is no different than the single hole case, 
except that there may be $O(1)$ holes of $R_{t+1}$ lying in $R_{t}^{h_{t},\out}$, which does not cause further complications.

\newcommand{\DDG}{\mathsf{DDG}}

\section{Construction}\label{sect:construction}

As in~\cite{CharalampopoulosGMW19}, we use dense distance graphs as a tool to build our oracle. To simplify the description, we still assume that $\boundary R$ lies on a single simple cycle for every region $R$ in the $\vec{r}$-division. Generalizing to multiple holes and nonsimple cycles
is straightforward.

The \textit{dense distance graph} of a region $R$ (denoted by $\DDG[R]$) is a complete directed graph on the vertices 
of $\partial R$, in which the length of $(u,v)$ is $\dist_R(u,v)$. 
We say that this kind of DDGs are \textit{internal} and, similarly, define 
the \textit{external} DDG of a region $R$ (denoted by $\DDG[R^{\out}]$) as a complete directed graph 
on $\partial R$, in which the length of $(u,v)$ is $\dist_{R^{\out}}(u,v)$. 

The FR-Dijkstra algorithm~\cite{FakcharoenpholR06} is an efficient implementation of Dijkstra's algorithm \cite{Dijkstra59} on DDGs.
In particular, it simulates the behavior of the heap in Dijkstra's algorithm
without explicitly scanning every edge in the DDGs.
In fact, the FR-Dijkstra algorithm can run on a union of DDGs \cite{FakcharoenpholR06}. Moreover, it is shown in \cite{BorradaileSW15} that it can also run compatibly with a traditional Dijkstra algorithm.
Suppose we have a graph $H$ that consists of a subgraph of $G$ on $n_0$ vertices, 
and $k$ DDGs on $n_1,n_2,\ldots,n_k$ vertices.  
The FR-Dijkstra algorithm can be implemented on $H$ in $\tilde{O}(N)$ time,
where $N=\sum_i n_i$.

%{\bf [CLARIFY THE STEPS IN THE NEXT PARAGRAPH MORE.  WHAT IS THE TOTAL RUNNING TIME?]}
%We compute the internal DDG and external DDG for each region in the $\vec{r}$-division, and also compute the DDG in $R_{i}^{\out}\cap R_{i}$ (denoted by $\DDG[R_{i}^{\out}\cap R_{i}]$) with vertices on $\partial R_{i}$ and $\partial R_{i+1}$ and weights the distance in $R_{i}^{\out}\cap R_{i}$. The edge weights in $\DDG[R_{i}^{\out}\cap R_{i}]$ can be derived from two MSSP data structures in $R_{i}^{\out}\cap R_{i}$ with sources $\partial R_{i}$ and $\partial R_{i+1}$ respectively. Internal DDGs and external DDGs can be computed in nearly linear time, where internal DDGs can still be computed by MSSP data structures and external DDGs are computed by a top-down approach in the recursive division of $G$ using Millar's simple-cycle separator.

Before the construction of DDGs and our oracle, we first prepare Klein's \MSSP{} structures (part (F) below). Note that \MSSP{} structures in part (F) are only used in the constructions of DDGs and part (E). They are not stored in our oracle and unrelated to the \MSSP{} structures from part (A). 
\begin{enumerate}
    \item [(F)] ({\bf More \MSSP{} Structures})
    For each $i\in[0,m-1]$, each $R_{i}\in\mathcal{R}_{i}$ with parent $R_{i+1}\in\mathcal{R}_{i+1}$, we build two \MSSP{} structures for $R_{i}^{\out}\cap R_{i+1}$ with sources on $\boundary R_{i}$ and $\boundary R_{i+1}$, respectively, 
    and an \MSSP{} structure for $R_{i}$ with sources on $\boundary R_{i}$. 
    
    All these \MSSP{} structures are constructed using Klein's \MSSP{} algorithm \cite{Klein05} or the one in 
    Appendix~\ref{sect:Euler} (with $\kappa=\log n$) in $\tilde{O}(\sum_{i}\frac{n}{r_{i}}r_{i+1})=\tilde{O}(mn^{1+1/m})$ time.
\end{enumerate}

We then compute, for each region $R_{i}$ in the $\vec{r}$-division, 
the internal DDG, the external DDG, 
and the DDG of $R_{i}^{\out}\cap R_{i+1}$ 
(denoted by $\DDG[R_{i}^{\out}\cap R_{i+1}]$) 
defined as the complete graph with 
vertices $\boundary R_{i}$ and $\boundary R_{i+1}$ and 
edge weights the distances in $R_{i}^{\out}\cap R_{i+1}$. 
The internal DDG and $\DDG[R_{i}^{\out}\cap R_{i+1}]$ 
for each region $R_{i}$ can be computed 
using \MSSP{} structures in part (F)
in $\tilde{O}(r_{i})$ and $\tilde{O}(r_{i+1})$ time respectively, 
so it takes $\tilde{O}(\sum_{i}\frac{n}{r_{i}}(r_{i}+r_{i+1}))=\tilde{O}(mn^{1+1/m})$ time over all regions. 
To compute the external DDGs, 
we consider a top-down process on the $\vec{r}$-division. 
The external DDG for $R_{i}$ can be computed 
by running the FR-Dijkstra algorithm sourced from $\boundary R_{i}$
on the union of $\DDG[R_{i+1}^{\out}]$ and $\DDG[R_{i}^{\out} \cap R_{i+1}]$. 
The size of the union is
$O(\sqrt{r_{i+1}})$, 
so computing $\DDG[R_{i}^{\out}]$ takes $\tilde{O}(\sqrt{r_{i}r_{i+1}})$ time, 
and the construction time over all external DDGs is
$\tilde{O}(\sum_{i}\frac{n}{r_{i}}\sqrt{r_{i}r_{i+1}})=\tilde{O}(mn^{1+1/(2m)})$. The total construct time for all DDGs is $\tilde{O}(mn^{1+1/m})$.

With dense distance graphs, all components in the oracle can be constructed as follows.

\begin{enumerate}

\item[(A)] {\bf \MSSP{} Structures}

Recall that our \MSSP{} structure for $R_{i}^{\out}$ with sites $\boundary R_{i}$ is obtained by contracting subpaths in $R_{i+1}^{\out}$ of the SSSP trees into single edges. In order to build the \MSSP{} structure using dynamic trees, 
it suffices to compute the contracted shortest path tree for every source on $\boundary R_{i}$ and then compare the differences between the trees of two adjacent sources on $\partial R_{i}$. 

For a single source on $\boundary R_{i}$, the contracted shortest path tree can be computed with 
FR-Dijkstra algorithm on the union of subgraph $R_{i}^{\out}\cap R_{i+1}$ and $\DDG[R_{i+1}^{\out}]$ in time $\tilde{O}(r_{i+1})$. Thus, the time for constructing and comparing the shortest path trees is $\tilde{O}(r_{i+1}\sqrt{r_{i}})$. After that, an \MSSP{} structure for $R_{i}^{\out}$ can be built in time $\tilde{O}((r_{i+1}+\sqrt{r_{i}r_{i+1}})\kappa n^{1/\kappa})$. The total time to construct all \MSSP{} structures is $\tilde{O}(\sum_{i}\frac{n}{r_{i}}(r_{i+1}\sqrt{r_{i}}+r_{i+1}\kappa n^{1/\kappa}))=\tilde{O}(n^{3/2+1/m}m+n^{1+1/\kappa+1/m}m\kappa)$.

\begin{remark}
    Notice that in our \MSSP{} structures for $R_{i}^{\out}$, a contracted subpath should be ``strictly'' inside $R_{i+1}^{\out}$, which contains no vertices belonging to $R_{i}^{\out}\cap R_{i+1}$ except its endpoints. However, the underlying shortest paths represented by edges in $\DDG[R_{i+1}^{\out}]$ may not satisfy this condition. To fix this problem, we add small perturbation to all edge weights in DDGs.
    %when running SSSP algorithm on $(R_{i}^{\out}\cap R_{i+1})\cup \DDG[R_{i+1}^{\out}]$, if the new distance label for a vertex $v$ is the same with the old one, i.e. there are different shortest paths from the source to $v$, we break the tie in favor of smaller weights of the edges incoming $v$. 
    Note that this will not break the Monge property of a DDG's adjacency matrix, on which the FR-Dijkstra algorithm relies. In fact, all different shortest paths from the source to $v$ in the union represent the same unique underlying shortest path in $R_{i}^{\out}$, and we choose the one passing as many as possible vertices in the union, on which each edge of the DDG is ``strictly'' inside $R_{i+1}^{\out}$. This mechanism will also be used below.
\end{remark}

%{\bf [EXPLAIN THE CONCLUSION ABOVE.  SWITCH $\rho$ to $n^{1/m}$.]}

\item[(B/C)]{\bf Voronoi Diagrams}

The additive weights of all Voronoi diagrams can be computed by a FR-Dijkstra algorithm running on a union of proper DDGs. Specific to $\VDout(u_{i},R_{i+1})$ in (B), additive weights are given by considering the union of $\DDG[R_{i}],\DDG[R_{i}^{\out}\cap R_{i+1}],\DDG[R_{i+1}^{\out}]$ in $\tilde{O}(\sqrt{r_{i+1}})$ time. For $\VDout(u_{i},R_{i})$ in (C), we focus on the union of $\DDG[R_{i}],\DDG[R_{i}^{\out}]$ and additive weights can be computed in time $\tilde{O}(\sqrt{r_{i}})$. The overall time to compute additive weights is $\tilde{O}(\sum_{i}\frac{n}{r_{i}}\sqrt{r_{i}}\sqrt{r_{i+1}})=\tilde{O}(mn^{1+1/(2m)})$.

An efficient algorithm to compute the dual representation of a Voronoi diagram is presented in \cite{CharalampopoulosGMW19}, by considering the complete recursive decomposition of $G$. Note that the complete recursive decomposition of $G$ is a binary decomposition tree, and the $\vec{r}$-division is a coarse version of it. It can be obtained in $O(n)$ time \cite{KleinMS13}.

\begin{lemma}\label{lem:constructVD}
(Cf.~Charalampopoulos et al.~\cite{CharalampopoulosGMW19}, Theorem 12)
Given a complete recursive decomposition of $G$, in which every region has been preprocessed for the FR-Dijkstra algorithm as in \cite{FakcharoenpholR06}, the dual representation of Voronoi diagrams on the complement of a specific region $R$ with sites $\boundary R$ and arbitrary input additive weights can be computed in time $\tilde{O}(\sqrt{|G|\cdot|\boundary R|})$.
\end{lemma}

By Lemma \ref{lem:constructVD}, the total construction time for the dual representations is 
\[
\tilde{O}\left(\sum_{i}\frac{n}{r_{i}}\sqrt{r_{i}}\sqrt{n\sqrt{r_{i+1}}} + \sum_{i}\frac{n}{r_{i}}\sqrt{r_{i}}\sqrt{n\sqrt{r_{i}}}\right)
= \tilde{O}\left(n^{3/2+1/(4m)}\right),
\]
which is also the construction time for parts (B) and (C).

%{\bf [EXPLAIN THE CALCULATIONS ABOVE MORE CAREFULLY]}

\item[(D)]{\bf Chord Trees and Piece Trees}

Recall that the chord tree $T_{q}^{R_{i}}$ is obtained from the shortest path tree in $G$ sourced from $q\in \boundary R_{i}$ by contracting all paths between vertices in $\boundary R_{i}$ into
single edges. Thus, it can be computed by running FR-Dijkstra on the union of 
$\DDG[R_{i}]$ and $\DDG[R_{i}^{\out}]$ in $\tilde{O}(\sqrt{r_{i}})$ time.

%{\bf [Can you make this paragraph clearer?  An inductive-style proof might be easier.  I.e., find a chord with nothing to one side of it; remove the chord; inductively build the tree for the rest, put the chord back and add a leaf to the tree.]}
Regarding the construction of piece tree $\mathcal{T}_{q}^{R_{i}}$, we first extract all the chords on $T_{q}^{R_{i}}$ in $R_{i}^{\out}$, i.e. the chord set $\mathcal{C}_{q}^{R_{i}}$. 
We treat each chord in $\mathcal{C}_{q}^{R_{i}}$ as an undirected edge and consider the undirected planar graph $Q$ which is the union of $\mathcal{C}_{q}^{R_{i}}$ and $\boundary R$. Observe that each piece in $\mathcal{P}_{q}^{R_{i}}$ relates to a face of $Q$. The piece tree $\mathcal{T}_{q}^{R_{i}}$ can be built straightforwardly in time $\tilde{O}(\sqrt{r_{i}})$. With the graph $Q$ and the piece tree $\mathcal{T}_{q}^{R_{i}}$, the data structure supporting $\MaximalChord$ and $\AdjacentPiece$ in Lemma \ref{lem:partD} can also be constructed in time $\tilde{O}(\sqrt{r_{i}})$ for the given $q, R_{i}$.

The total time for the whole part (D) is $\tilde{O}(\sum_{i}\frac{n}{r_{i}}\sqrt{r_{i}}\sqrt{r_{i}})=\tilde{O}(nm)$.

%By establishing, for each vertex $s\in \partial R$, a sequence in clockwise order containing chords with endpoints $s$, a piece $P$ and its boundary chords can be detected by starting from an arbitrary boundary chord (and a proper direction) of $P$ and walking along a clockwise cycle as small as possible on the undirected planar graph made up of $\mathcal{C}_{u}^{R}$ and $\partial R$. We get all the pieces in the following way. Initially, both directions of chords are unlabeled. We enumerate both directions of all chords and keep starting from an unlabeled direction to scan a new piece. Once we find a new piece, we label the direction we walk along of each its boundary chord. The piece tree $\mathcal{T}_{u}^{R}$ can be established with the piece set $\mathcal{P}(T_{u}^{R})$ and the knowledge of boundary chords of each piece. Obviously, the construction time for piece trees is also $\tilde{O}(nm)$.   

\item[(E)]{\bf Site Tables and Side Tables}

We focus on the site table and side table for a specific $\VDout(u,R_{i})$, and do some preparations.
\begin{itemize}
    \item Observe that the union of 
\[
\DDG[R_{i}^{\out}\cap R_{i+1}],\DDG[R_{i+1}^{\out}\cap R_{i+2}],\ldots,\DDG[R_{m-1}^{\out}]
\]
contains exactly all boundary vertices in $R_{i}^{\out}$ of ancestors $R_{i},R_{i+1},\ldots,R_{m-1}$. 
We use $H$ to denote this union with an artificial super-source $u'$ connected to each site $s\in \boundary R_{i}$ with weight $\omega(s)$, and construct the shortest path tree $T_{H}$ in $H$ sourced from the super-source $u'$ by FR-Dijkstra algorithm, which costs $\tilde{O}(\sqrt{n})$ time.
\end{itemize}

Remember that the site table stores the first and last vertices of each site-centroid $s$-$y$ path on the boundary of each ancestor $R_{i'}$ ($i'\geq i$). We first find the last vertex $x$ on the $s$-$y$ path belonging to $H$. Assume that $y\in R_{k+1}$ but $y\notin R_{k}$, where $R_{k},R_{k+1}$ are ancestors of $R_{i}$. We can observe that $x$ is the vertex in $\boundary R_{k}\cup\boundary R_{k+1}$ with the minimal $\dist_{H}(u',x)+\dist_{R_{k}^{\out}\cap R_{k+1}}(x,v)$ (breaking ties in favor of larger $\dist_{H}(u',x)$). The former is given by $T_{H}$ and the latter can be found by querying \MSSP{} structures in (F) for $R_{k}^{\out}\cap R_{k+1}$.
The calculation of $x$ needs time $\tilde{O}(|\boundary R_{k+1}|)=\tilde{O}(\sqrt{n})$.
Observe that the $u'$-$x$ path on $T_{H}$ includes all boundary vertices of upper regions on the $s$-$y$ path. By retrieving the $u'$-to-$x$ path on $T_{H}$ in $O(\sqrt{n})$ time, we can get the required information for the site table. The construction time of a site table for $\VDout(u,R_{i})$ is $\tilde{O}(\sqrt{r_{i}}\sqrt{n})$.

In the side table, we will store the relationship (left/right/{\bf Null}) between each site-centroid-site chord $C=\overrightarrow{s_{j}y_{j}y_{j-1}s_{j-1}}$ (using the notations in Figure \ref{fig:CentroidSearch}) and each ancestor $R_{i'}^{\out}$ $(i'\geq i)$. With the technique used in the construction of site tables, we can extract all vertices of $C$ on each $\boundary R_{i'}$ from $T_{H}$, and then determine the relationship between $C$ and each $R_{i'}^{\out}$ with boundary vertices on $C$. For each $R_{i'}^{\out}$ that $C$ contains no vertices on $\boundary R_{i'}$, we pick an arbitrary vertex $z$ on $\boundary R_{i'}$. We can retrieve from $T_{H}$ the $u'$-$z$ path and find the site $s_{z}$ s.t. $z\in \Vor(s_{z})$. This can be done in $O(\sqrt{n})$ time. With $T_{H}$ and \MSSP{} structures in part (F), we can determine the pairwise relationships among $s_{j}$-$y_{j}$, $s_{j-1}$-$y_{j-1}$ and $s_{z}$-$z$ shortest paths and know whether $z$ lies to the left or right of $C$, which immediately shows the relationship between $C$ and $R_{i'}^{\out}$. The construction time for a side table of $\VDout(u,R_{i})$ is $\tilde{O}(\sqrt{r_{i}}m\sqrt{n})$.

The total time for building all site tables and side tables
is $\tilde{O}(\sum_{i}\frac{n}{r_{i}}\sqrt{r_{i}}\sqrt{r_{i+1}}m\sqrt{n})=\tilde{O}(n^{3/2+1/(2m)}m^{2})$.

\end{enumerate}

The overall construction time is $\tilde{O}(n^{3/2+1/m}+n^{1+1/m+1/\kappa})$ since $m$ and $\kappa$ should be functions of $n$ that are $O(\log n)$.

\end{document}